\documentclass[11pt, twoside]{amsart}
\usepackage{geometry, appendix}                
\usepackage{graphicx}
\usepackage{amssymb, amsthm, mathrsfs, bm}
\usepackage{epstopdf}
\usepackage{enumerate}
\usepackage{color}
\usepackage[colorlinks, allcolors=black, linkcolor=black]{hyperref}
\usepackage{upgreek}
\usepackage{graphicx}
\usepackage{graphics, cancel} 
\usepackage{color}
\usepackage{amsmath, amsthm, slashed}
\usepackage{amssymb}
\usepackage{rotating}
\usepackage{epsfig}
\usepackage{verbatim, longtable}
\usepackage{rotating}
\usepackage{graphicx}
\newcommand{\be}{\begin{equation*}}
\newcommand{\ee}{\end{equation*}}
\newcommand{\ben}[1]{\begin{equation}\label{#1}}
\newcommand{\een}{\end{equation}}
\newcommand{\bea}{\begin{eqnarray}}
\newcommand{\eea}{\end{eqnarray}}
\newcommand{\bean}{\begin{eqnarray*}}
\newcommand{\eean}{\end{eqnarray*}}

\newcommand{\R}{\mathbb{R}}
\newcommand{\A}{\A}

\newcommand{\C}{\mathbb{C}}
\newcommand{\cS}{\mathcal{S}}

\renewcommand{\O}[1]{\mathcal{O}\left( #1 \right)}

\renewcommand{\H}{\underline{H}}
\renewcommand{\A}{\mathcal{A}}

\renewcommand{\L}{\underline{L}}
\newcommand{\hL}{\hat{L}_s}

\newcommand{\scri}{\mathscr{I}}
\newcommand{\hor}{{\mathscr{H}^+}}
\newcommand{\bhR}{\mathscr{R}}

\newcommand{\cc}{\Subset }
\newcommand{\abs}[1]{\left|#1 \right|} 
\newcommand{\norm}[2]{\left|\left |#1 \right| \right |_{#2}} 
\newcommand{\ip}[3]{\left(#1,#2 \right )_{#3}} 
\newcommand{\pair}[2]{\left\langle#1,#2 \right \rangle}

\newcommand{\tn}{\tilde{\nabla}}

\newcommand{\emT}{\tilde{\mathbb{T}}}
\newcommand{\emS}{\tilde{\mathbb{S}}}
\newcommand{\curJ}{\tilde{\mathbb{J}}}

\newcommand{\eq}[1]{(\ref{#1})}
\newtheorem{Theorem}{Theorem}

\newtheorem*{conj*}{Conjecture}
\newtheorem{Lemma}[Theorem]{Lemma}
\newtheorem*{Lem*}{Lemma}
\newtheorem{Corollary}[Theorem]{Corollary}
\newtheorem{Definition}[Theorem]{Definition}

\numberwithin{equation}{section}
\numberwithin{Theorem}{section}
\newtheorem{innercustomthm}{Theorem}
\newenvironment{customthm}[1]
  {\renewcommand\theinnercustomthm{#1}\innercustomthm}
  {\endinnercustomthm}
  \newtheorem{innercustomdef}{Definition}
\newenvironment{customdef}[1]
  {\renewcommand\theinnercustomdef{#1}\innercustomdef}
  {\endinnercustomdef}

\title[Quasinormal Modes]{On quasinormal modes of asymptotically anti-de Sitter black holes}
\author{Claude M. Warnick}
\thanks{\texttt{warnick@ualberta.ca} \\
\phantom{1 } Department of Physics, 4-181 CCIS, University of Alberta, Edmonton AB T6G 2E1, Canada \\
\phantom{1 } \texttt{c.warnick@warwick.ac.uk} \\
\phantom{1 } Mathematics Institute, Zeeman Building, University of Warwick, Coventry CV4 7AL, UK}

\begin{document}
\begin{abstract}
We consider the problem of quasinormal modes (QNM) for strongly hyperbolic systems on stationary, asymptotically anti-de Sitter black holes, with very general boundary conditions at infinity. We argue that for a time slicing regular at the horizon the QNM should be identified with certain $H^k$ eigenvalues of the infinitesimal generator $\A$ of the solution semigroup. Using this definition we are able to prove directly that the quasinormal frequencies form a discrete, countable subset of $\mathbb{C}$ which in the globally stationary case accumulates only at infinity. We avoid any need for meromorphic extension, and the quasinormal modes are honest eigenfunctions of an operator on a Hilbert space. Our results apply to any of the linear fields usually considered (Klein-Gordon, Maxwell, Dirac etc.) on a stationary black hole background, and do not rely on any separability or analyticity properties of the metric. Our methods and results largely extend to the locally stationary case. We provide a counter-example to the conjecture that quasinormal modes are complete. We relate our approach directly to the approach via meromorphic continuation. \\\phantom{1}\\ \phantom{1} \hfill ALBERTA THY 3-13
\end{abstract}
\maketitle
\thispagestyle{empty}
\vspace{-.3cm}
\tableofcontents

\newpage 
\section{Introduction}

A hugely important tool in the physical sciences is the idea of spectral analysis. Many physical systems respond to stimuli by emitting radiation at certain precise characteristic frequencies, the simplest example being the normal modes of a guitar string. Understanding the characteristic frequencies is a significant step in analysing a system.

In the context of black holes, the appropriate notion of `characteristic frequencies' that one should study are the quasinormal frequencies (QNF) \cite{Kokkotas:1999bd, Berti:2009kk, Konoplya:2011qq} and their corresponding quasinormal modes (QNM). Unlike the modes of a guitar string, these represent behaviour which is both oscillatory and decaying. Accordingly, the frequencies are necessarily complex. The theory of quasinormal modes is significantly complicated by the fact that they, unlike normal modes, are not usually understood as eigenfunctions of some operator. One aim of this paper is to remedy this situation by showing that the quasinormal modes can, and indeed should, be understood as honest eigenfunctions of a operator on a Hilbert space.

Much recent progress in understanding the behaviour of linear fields on black hole backgrounds has come from an improved understanding of behaviour near the horizon. In particular the combination of a regular choice of time slicing with the redshift effect has paid dividends in understanding the decay of fields outside many stationary or ultimately stationary black holes (see for example \cite{Mihalisnotes} and references therein). The price one pays for this approach is that energies are typically no longer conserved, but instead decay in time. As a result, time evolution is not  `unitary', i.e.\ cannot be viewed as a group action preserving a norm on the Hilbert space of states.

In this paper our basic philosophy is to take the view that a regular slicing is the natural setting in which to consider fields outside a black hole. In this setting, rather than unitary evolution, the time evolution should instead be thought of as the action of a (contraction) semigroup. This is very natural, since energy and information can and will fall into the black hole. Associated to such a semigroup is a operator which generates infinitesimal time translations. We shall simply define the quasinormal modes to be certain eigenmodes of this operator. A second branch of our philosophy is the primacy of the hyperbolic problem. Although the operator whose eigenvalues we seek is (degenerate) elliptic, we make heavy use of estimates for the hyperbolic problem to derive the estimates we require for the elliptic problem.

The precise Hilbert space on which the semigroup acts will be seen to be of considerable importance. In particular, the more rapidly decaying a quasinormal mode, the higher the regularity of the Hilbert space one should consider in order to identify the mode. This has a natural interpretation, since for rough initial data one may construct slowly decaying solutions by placing a suitable lump of energy on the horizon. The higher the regularity of the initial data, the smoother such a lump must be and the more rapidly it will decay. 

We will focus our attention on asymptotically anti-de Sitter (AdS) black holes. The quasinormal modes of such black holes are of great interest in the context of the putative AdS/CFT correspondence (for example, see \cite{Kovtun:2005ev}). Classical gravity in asymptotically AdS backgrounds has also recently attracted attention, owing to a conjectured instability \cite{daf, Anderson:2006ax} of the anti-de Sitter spacetime for which numerical evidence has recently been provided \cite{Bizon}. One feature of asymptotically AdS black holes is that there is a large zoology of spacetimes, and a variety of linear systems which are considered in these backgrounds. Our methods allow us to treat such problems in full generality.

With our definition in terms of the semigroup associated to a regular slicing, we are able to prove discreteness of the quasinormal frequencies of strongly hyperbolic operators on globally stationary asymptotically AdS black holes. The class of operators we consider allows us to treat the Klein-Gordon, Dirac and Maxwell equations as well as many others that often arise.  Our definition does not require any separability of the equations under consideration, nor any real analyticity of the metric. We can show the same result holds for locally stationary black holes provided we restrict to solutions consisting of a finite number of angular modes. It can be readily extended to the asymptotically de Sitter case, where it is closely related to the approach of Vasy in \cite{vasy10} (see the discussion below). A final advantage of our approach is that it easily permits consideration of perturbations which do not vanish on the horizon. We are able to relate the results back to the standard definition of quasinormal modes as resonances in the meromorphic continuation of a resolvent.

\subsection{Relation to previous works}

Quasinormal modes of black holes have a long history of study in the physics literature, and we would not be able to do justice to the whole body of work in the space available here. We instead refer the reader to the review articles \cite{Kokkotas:1999bd, Berti:2009kk, Konoplya:2011qq} for a survey of the many results in this area. This work typically focusses on explicit metrics and systems for which variables may be separated. Of particular relevance to us are approaches which use regularity on the horizon as a boundary condition. For an example of this, see \cite{Horowitz:1999jd}. For other early work on quasinormal modes in asymptotically anti-de Sitter spacetimes, see \cite{Chan:1996yk,Chan:1999sc,Cardoso:2001vs,Cardoso:2001bb,Cardoso:2001hn}.

More mathematical study of quasinormal modes was initiated by Bachelot \cite{Bachelot, Bachelot2} for the asymptotically flat Schwarzschild black hole. More recently,  very impressive results for asymptotically de Sitter spacetimes have been obtained. Quasinormal modes for the Schwarzschild-de Sitter spacetime were defined in \cite{barreto1997distribution} by making use of the results of \cite{Mazzeo1987260} on the meromorphic continuation of the Laplace resolvent for asymptotically hyperbolic manifolds. The high frequency limit was studied in \cite{BonyHafner, sabaretto}, where an exponential decay rate was proven, together with an expansion in terms of quasinormal modes. This was then extended to slowly rotating Kerr-de Sitter by Dyatlov, who used separation of variables to define the quasinormal modes in \cite{Dyatlov:2010hq}. The high frequency limit was studied in \cite{Dyatlov:2011jd,Dyatlov:2011zz, Wunsch:2011fk, NZ}. More recently, Dyatlov has extended these results to spacetimes close to Kerr-de Sitter \cite{Dyatlov:2013hba}, using methods developed in \cite{vasy10} for more general stationary asymptotically de Sitter spacetimes. 

The seminal paper of Vasy \cite{vasy10} is of particular importance in this story. In this paper, the quasinormal modes are defined for a large class of spacetimes, which includes the Kerr-de Sitter case (but not including asymptotically anti-de Sitter spacetimes). The methods used are primarily those of microlocal analysis (as is the case for several of the papers mentioned in the previous paragraph). The current paper can be viewed as a parallel development of some of the results of Vasy, using physical space methods in place of microlocal arguments. In particular, this approach allows us to make use of the renormalised energy spaces introduced in \cite{Warnick:2012fi} in order to incorporate general boundary conditions for the fields at conformal infinity. The two key effects which we exploit in this paper, the redshift estimate and the enhanced redshift effect, have a loose analogue in the radial points estimate and the elliptic estimate of Vasy. The methods we use, however, are significantly different. Vasy goes further than our results, in obtaining high frequency results and non-trapping estimates, as well as being able to treat spacetimes containing an ergoregion, provided a certain non-trapping assumption holds.

In the anti-de Sitter case, quasinormal modes for the Klein-Gordon equation with Dirichlet boundary conditions on the Schwarzschild-AdS black hole have been studied by Gannot \cite{Gannot:2012pb}. Gannot uses a `black-box' approach to define the quasinormal modes after separation of variables and furthermore finds a sequence of quasinormal frequencies which approach the imaginary axis exponentially rapidly.  The work of Holzegel and Smulevici \cite{HolSmul,Holzegel:2013kna} in the Kerr-AdS case is closely related, although they do not directly construct the quasinormal modes.

\subsection*{A note on frequency conventions}

Because we make heavy use of the Laplace transform, it is convenient to work with a definition of quasinormal frequencies in which the time dependence of a quasinormal mode with frequency $s$ is $e^{s t}$. As a result, decaying QNM have frequencies which reside in the left half-plane. The quasinormal frequency is more often defined to be $\omega = i s$, so that the time dependence of a mode is $e^{-i \omega t}$ and decaying modes inhabit the lower half plane. All of our results may be easily restated with this convention.

\subsection{The traditional approach} \label{traditional section} 

Before discussing our results in more detail, we give here a brief summary of the usual approach to quasinormal modes for asymptotically AdS black holes as considered in the physics literature. See for example \cite{Kokkotas:1999bd, Konoplya:2011qq} for a more detailed exposition. In the simplest case of a conformally coupled Klein-Gordon field (which permits us to ignore complications at infinity), analysis of the Klein-Gordon equation may be reduced to that of  the $(1+1)$-dimensional wave equation with potential:
\ben{wscr}
\psi_{\tau \tau}-\psi_{xx} + V(x)\psi = 0,
\een 
where $0\leq x<\infty$. We require boundary conditions at $x=0$, which for concreteness we assume to be Dirichlet. $x=0$ corresponds to the conformal infinity of the black hole and $x=\infty$ to the horizon. To solve this equation for $\tau \geq 0$, we may Laplace transform in the time variable\footnote{The Fourier transform is also often used, however in the context of an initial value problem, the Laplace transform is more natural. Of course, for complex values of the spectral parameter, the two are essentially the same.}, to deduce that
\ben{lap1}
-\hat{\psi}_{xx} + (V(x)+s^2) \hat{\psi} = f, \qquad f(x) = s \psi(0, x)+ \psi_\tau(0, x).
\een
In order to construct $\psi(\tau, x)$, we need to solve \eq{lap1} for $\hat{\psi}(s,x) \in L^2(\mathbb{R})$ where $s$ lies in the half-plane $\Re(s)>c_0$ for some $c_0$. Then we can invert the Laplace transform with the Bromwich integral \cite{HilleLaplace}:
\ben{Lapinv}
\psi(\tau, x) = \lim_{T\to \infty} \frac{1}{2 \pi i} \int_{c-iT}^{c+iT} e^{s \tau} \hat{\psi}(s,x) ds, \qquad c>c_0.
\een

For simplicity\footnote{The potentials arising in the black hole context typically have $V$ vanishing exponentially rapidly for large $r$, but the analysis is not significantly different.}, we'll assume that the smooth potential $V(x)$ has support in $[0, K]$. We denote by $u^0(s,x)$ the unique smooth solution of
\ben{lap2}
-{u}_{xx} + (V(x)+s^2) u =0,
\een
with $u^0(s,0)=0, u^0_x(s,0)=1$ and by $u^\infty(s, x)$ the unique smooth solution equal to $e^{-s x}$ for $x>K$. Suppose for some $s$ with $\Re(s)>0$ that these two solutions are not independent, then $u^0(s, x) = \lambda u^\infty(s, x)=u$ is smooth and decays exponentially for large $x$. Multiplying \eq{lap2} by $\overline{s} \overline{u}$, integrating over $x$ and taking the real part, we deduce that
\be
\Re(s) \int_0^\infty \abs{u_x}^2 + (V(x) + \abs{s}^2) \abs{u}^2 dx = 0,
\ee
which if $\abs{s}^2> -\min_r V(r)$ implies $u\equiv 0$, a contradiction. Thus there exists a $c_0\geq 0$ such that $\Re(s)>c_0$ implies that $u^0(s, x)$ and $u^\infty(s,x)$ are linearly independent. We may form the Wronskian of the two functions and find
\be
W(s) = \left| \begin{array}{cc} u^\infty(s, x_0) & u^0(s, x_0) \\ u^\infty_x(s, x_0) & u^0_x(s, x_0)\end{array} \right|, \qquad W(s)\neq 0,
\ee
where $W(s)$ is independent of the point $x_0$ at which the Wronskian is evaluated. We define the Green's function for the operator \eq{lap2} by
\be
G(s; x, \xi):= \left \{ \begin{array}{ll} \frac{u^0(s, x) u^\infty(s, \xi)}{W(s)}, & x<\xi, \\ \frac{u^\infty(s, x) u^0(s, \xi)}{W(s)}, & x>\xi.\end{array}\right.
\ee
We then have that the solution of \eq{lap1} is given by:
\ben{Gint}
\hat{\psi}(s, x) = \int_0^\infty G(s; x, \xi) f(\xi) d\xi.
\een

Now, the Green's function is in fact \emph{holomorphic}\footnote{more precisely, the family of operators mapping $f$ to $\hat{\psi}$ given by \eq{Gint} is a holomorphic family of bounded operators $L^2(\R)\to L^2(\R)$ which admits a meromorphic extension as operators $L^2_c(\R) \to L^2_{loc.}(\R)$ for $\Re(s)\leq c_0$.} in $s$ provided that $\Re(s)>c_0$. Thus we may consider the analytic extension to $\Re(s)\leq c_0$. The functions $u^\infty(s, x)$ and $u^0(s, x)$ are perfectly well behaved for any complex value of $s$, so the Green's function can only fail to be holomorphic if $W(s)$ vanishes for some $s$. Since $W(s)$ is holomorphic, it can have only isolated zeros, and these occur exactly at the values of $s$ where $u^\infty$ and $u^0$ are linearly dependent. These values of $s$ are known as \emph{quasinormal frequencies} and the corresponding function $u=u^\infty \propto u^0$ is a \emph{quasinormal mode}. Quasinormal modes may be characterised as non-trivial solutions of \eq{lap2} satisfying
\ben{rad}
u = 0\ \  \textrm{ at }\ \  x=0, \qquad \quad  u = e^{-s x},\ \  \textrm{ for }\ \  x>K.
\een

Given a QNM, we may construct a solution of \eq{wscr}:
\ben{QNM1}
\psi(\tau, x) = e^{s \tau} u(s, x).
\een
 If $\Re(s)>0$ is a QNF, then $u(s, x)\in L^2(\R)$, and $u$ is simply a growing mode of \eq{wscr} with finite energy. If, however, $\Re(s)\leq 0$ then $u(s, x)\not \in L^2(\R)$. For any fixed value of $x$, \eq{QNM1} will decay in time with rate $-\Re(s)$, while oscillating with a frequency $\Im(s)$. Note that for large $x$ we have
 \be
 \psi(\tau, x) = e^{s(\tau-x)}, \quad r>K,
 \ee
 so that for $x$ large and positive, $\psi(\tau, x)$ is a right moving wave. This is often taken to implement the physical condition that nothing is `coming in from infinity'. These are usually referred to as `ingoing' boundary conditions\footnote{Recall that the horizon is at $x=\infty$, so this is a wave travelling towards the black hole}.
 
 One may hope to deform the contour of integration in \eq{Lapinv} to pick up contributions from the poles at the quasinormal frequencies $s_i$, so that a generic solution to \eq{wscr},  for late times, takes the form
 \ben{qnsum}
 \psi(\tau, x) \sim \sum_{s_i} a_i e^{s_i \tau} u_i(s_i, x),\qquad \textrm{ as }\tau\to \infty,
 \een
with constants $a_i$ determined from the initial data. Whether this can be done depends on the large $\abs{s}$ behaviour of the Green's function. In certain specific cases, this can be demonstrated explicitly \cite{BonyHafner, Dyatlov:2010hq, Beyer:1998nu}.

This discussion may be extended to the case of potentials which are not compactly supported. In this case the boundary conditions required for large $r$ are modified slightly. The Green's function may also develop branch cuts in the complex plane if the potential $V$ does not decay sufficiently rapidly, as occurs in the asymptotically flat case which we do not consider here.

\subsection{Observations on the traditional approach}

While the traditional approach to QNM is undoubtedly mathematically sound, and can be extended to cases where the relevant equations do not separate, there are several criticisms that one may level, some aesthetic and some practical. We shall list here some objections, not all of which are independent.
\begin{enumerate}[1.]
\item The ingoing boundary conditions, implying that no information is `coming in from the horizon' are rather unnatural. The general solution of the one-dimensional wave equation, 
\be
\psi(\tau, x) = \psi_+(\tau+x)+\psi_-(\tau-x),
\ee
has components which `come in from infinity' in both directions, but it is not physically reasonable to require that these vanish. The quasinormal modes are thus subject to different boundary conditions to the original problem. We shall see in \S\ref{example}, \ref{ressec} that the ingoing boundary conditions can actually be too restrictive, as they arise from considering initial data supported away from the horizon.
\item The ingoing boundary conditions \eq{rad} near the horizon imply that decay in time corresponds to an exponential growth in the $x$ coordinate.
\item The QNM with $\Re(s) < 0$ do not belong to any obvious Hilbert space (as a result of the exponential growth near infinity), so are not eigenfunctions in the usual sense.
\item Any sum of quasinormal modes cannot hope to approximate the solution for all $x$, only within an arbitrarily chosen range $x<x_0$ (see for example the results of \cite{Beyer:1998nu}).  In other words the error term implicit in \eq{qnsum} cannot be uniform in $x$.
\item The construction of quasinormal modes in terms of a meromorphic extension through the continuous spectrum obscures their physical meaning. It also seems rather arbitrary that the growing modes and the decaying modes should be treated differently.
\item As a matter of practical calculation, the boundary conditions are hard to implement when no analytic solution is available. This is because one wishes to set to zero the coefficient of an exponentially decaying term, which is hard to pick out numerically against the background of an exponentially growing term. It should be noted that methods have been developed which successfully overcome these problems, for example the method of complex scaling (also known as the Perfectly Matched Layer method \cite{Berenger1994185}).
\item The definition of QNM as poles in the meromorphic extension of the scattering resolvent can be extended to cases where the wave equation does not separate, but the meaning of the boundary conditions becomes even less clear.
\item The definition requires the presence of a bifurcate Killing horizon in the spacetime. This is an unphysical requirement, since black-holes arising as the endpoint of a gravitational collapse will not contain such a structure.
\end{enumerate}

It should be stressed that these issues have been overcome in the mathematical literature on the subject of QNM. For Schwarzschild-de Sitter, \cite{barreto1997distribution, sabaretto} address these questions. For more general spacetimes, the work of Vasy \cite{vasy10} addresses all of these points. Our goal is to establish that these issues may be resolved in a relatively straightforward way by considering a different slicing of the spacetime which is regular on the horizon.

\subsection{The regular slicing approach}

If one reflects on the problems listed above, it appears that the source of many of them is the fact that energy is conserved for \eq{wscr}. We can avoid this issue by changing to a new set of coordinates which are motivated by black hole coordinates which are regular on the horizon. Let us return to equation \eq{wscr}, and introduce new coordinates $\rho, t$ by:
\begin{align*}
\rho &= 1- \tanh{x},\\
t &= \tau - x + \tanh{x}.
\end{align*}
We have that $0<\rho\leq 1$, with $\rho \to 0$ corresponding to $x \to \infty$, so that the horizon is now at $\rho=0$ while the AdS conformal boundary has been mapped to $\rho = 1$. After dividing through by a factor $\rho(2-\rho)$, the equation \eq{wscr} becomes:
\begin{align}
\nonumber 0&= (1+(1-\rho)^2) \frac{\partial^2 \phi}{\partial t^2} -  \frac{\partial}{\partial \rho}\left(\rho(2-\rho) \frac{\partial \phi }{\partial \rho} \right)\\ & \quad  -  (1-\rho)^2 \frac{\partial^2 \phi}{\partial \rho \partial t}- \frac{\partial}{\partial \rho}\left( (1-\rho)^2 \frac{\partial \phi}{\partial t}\right) + \frac{\tilde{V}}{\rho(2-\rho)} \phi. \label{regular wave}
\end{align}
Here $\phi(t, \rho) = \psi(\tau(t, \rho), x(t, \rho))$ and $\tilde{V}(\rho) = V(\tanh^{-1}(1-\rho))$. Notice that since $V$ vanishes for large $x$, the combination $\rho^{-1} \tilde{V}$ is smooth up to $\rho=0$. Now let us Laplace transform this equation in the time variable, $t$. We will  denote this Laplace transform by $\hat{\phi}(s, \rho)$. Note that $\hat{\phi}$ is related to $\hat{\psi}$ of  \S \ref{traditional section}, but not by the na\"ive change of coordinates $x \to \rho$. We find that $\hat{\phi}$ obeys:
\begin{align}
\nonumber g&= \hat{L}_s \hat{\phi} := -  \frac{\partial}{\partial \rho}\left(\rho(2-\rho) \frac{\partial \hat{\phi} }{\partial \rho} \right)  -  s (1-\rho)^2 \frac{\partial \hat{\phi}}{\partial \rho}- s\frac{\partial}{\partial \rho}\left( (1-\rho)^2 \hat{\phi} \right) \\ & \quad+ \left[ \frac{\tilde{V}}{\rho(2-\rho)} +s^2 (1+(1-\rho)^2) \right]\hat{\phi}\label{regular operator}
\end{align}
here $g$ is a function constructed from initial data, and we have introduced the degenerate elliptic operator $\hat{L}_s$. In order to construct $\phi(t, \rho)$ by the inverse Laplace transform, we need to be able to invert $\hat{L}_s$. The key result is the following Lemma:
\begin{Lemma}\label{toy lemma}
Let $I = [0, 1]$ be the closed unit interval. Suppose that $g \in H^{k-1}(I)$, then for $s$ belonging to the half-plane $\Re(s) >(\frac{1}{2} - k)$ either:
\begin{enumerate}[i)]
\item The equation \eq{regular operator} admits a unique solution $\hat{\phi}(s, \cdot) \in H^{k}(I)\cap H^{k+1}_{loc.}(I)$ which vanishes at $\rho=1$ for any $g$,
 \\
 or:
\item  There exists  $w(s, \cdot)\in C^\infty(I)$, vanishing at $\rho=1$, which solves the homogeneous problem, i.e.\ a solution of $\hat{L}_s w= 0$.
\end{enumerate}
Moreover, possibility $ii)$ can only occur for isolated values of $s$.
\end{Lemma}
\begin{proof}
We will very briefly sketch a proof of this result, with an emphasis on showing how to derive the relevant estimates. Let us first take $k=1$, and for convenience we'll assume $V \geq 0$. 
\begin{enumerate}[1.]
\item Let $\gamma>1$ be a constant which we will fix later. Multiplying $\hat{L}_s u + \gamma u$ by $\overline{s} \overline{u}$, taking the real part  and integrating by parts we find
\begin{align*}
 &\int_0^1 \Re\left[ \overline{s} \overline{u} (\hat{L}_s u + \gamma u)\right] d\rho -\abs{s}^2 \abs{u(0)}\\ & \quad = \Re(s) \int_0^1\left\{ \rho (2-\rho) \abs{\partial_\rho u}^2 + \left( \frac{\tilde{V}}{\rho(2-\rho)} +\abs{s}^2 (1+(1-\rho)^2) + \gamma\right)\abs{u}^2\right\} d\rho .
\end{align*}
From here we readily deduce that if $\Re(s)>0$
\ben{first estimate}
E(u):= \int_0^1 \left\{ \rho \abs{\partial_\rho u}^2  + \gamma\abs{u}^2 \right\} d \rho \leq \epsilon \norm{ (\hat{L}_s  + \gamma) u)}{L^2}^2 + \frac{C^{(1)}_s}{\epsilon}\norm{u}{L^2}^2,
\een
 for any $\epsilon>0$, holds for some constant $C^{(1)}_s$, independent of $\gamma$ and $\epsilon$. Notice that the weight in front of the $\abs{\partial_\rho u}^2$ term  in this estimate degenerates at $\rho=0$.

\item Next, we multiply $\hat{L}_s u + \gamma u$ by $(-\partial_\rho \overline{u})$, take the real part and integrate by parts to obtain:
\begin{align*}
 &-\int_0^1 \Re\left[ \partial_\rho \overline{u} (\hat{L}_s u + \gamma u)\right] d\rho - \frac{1}{2}\abs{\partial_\rho u(1)}^2 - \frac{\gamma}{2} \abs{u(0)}^2 \\ & \quad =  \int_0^1 \left[ (1-\rho) + 2 \Re(s) (1-\rho)^2 \right] \abs{\partial_\rho u}^2  d\rho \\ &-  \int_0^1\Re\left\{\left( \frac{\tilde{V}}{\rho(2-\rho)} +s^2 (1+(1-\rho)^2)+ 2s(1-\rho)\right) (\partial_\rho \overline{u} ) u \right\} d\rho .
\end{align*}
Notice that the term involving an integral over $\abs{\partial_\rho u}^2$ does not have a degenerate weight at $\rho=0$ and moreover is positive provided $\Re(s)>-\frac{1}{2}$. Adding a multiple of $E(u)$ to both sides of this equation, we deduce that for $\Re(s)>-\frac{1}{2}$ we have:
\ben{second estimate}
\norm{u}{H^1} +\gamma\norm{u}{L^2}^2 \leq C^{(2)}_s\left[ \norm{ (\hat{L}_s  + \gamma) u)}{L^2}^2 + E(u)\right]
\een
holds for some constant $C^{(2)}_s$, independent of $\gamma$.
\item Now let us return to \eq{first estimate}. Notice that $\hat{L}_{s_1} - \hat{L}_{s_2}$ is a first order differential operator. Let us therefore apply \eq{first estimate} with $\epsilon$ sufficiently small to deduce that if $\Re(s) > -\frac{1}{2}$ we have that:
\bean
E(u) &\leq& \epsilon \norm{ (\hat{L}_{s+\frac{1}{2}}  + \gamma) u)}{L^2}^2 + \frac{C^{(1)}_{s+\frac{1}{2}}}{\epsilon}\norm{u}{L^2}^2 \\
&\leq &  \delta\norm{u}{H^1}^2 + C^{(3)}_{s, \delta}  \left( \norm{ (\hat{L}_{s}  + \gamma) u)}{L^2}^2  + \norm{u}{L^2}^2 \right)
\eean
holds for any $\delta>0$ for some constant $C^{(3)}_{s, \delta}$. Taking $\delta$ small enough and inserting this into \eq{second estimate} we have:
\be
\norm{u}{H^1}^2  +\gamma\norm{u}{L^2}^2  \leq C^{(4)}_s\left[ \norm{ (\hat{L}_s  + \gamma) u)}{L^2}^2 + \norm{u}{L^2}^2\right]
\ee
for $C^{(4)}_s$ independent of $\gamma$. Finally then, we take $\gamma$ large enough to absorb the $L^2$ term on the right hand side and we conclude that
\ben{final estimate}
\norm{u}{H^1}^2 \leq C  \norm{ (\hat{L}_s  + \gamma) u)}{L^2}^2,
\een
provided $\Re(s) >-\frac{1}{2}$. 
\item The estimate \eq{final estimate} shows that $\hat{L}_s + \gamma$ is injective for $\gamma$ large enough provided $\Re(s)>-\frac{1}{2}$. By considering an adjoint problem it is possible to show that $\hat{L}_s + \gamma$ is in fact surjective. Thus $(\hat{L}_s + \gamma)^{-1}$ exists and moreover is a compact operator by virtue of the Rellich-Kondrachov theorem, since it is a bounded map $L^2 \to H^1$. An application of the analytic Fredholm theorem establishes the result.
\item To increase $k$ we differentiate the original equation. Rather than show this in detail, we look at the structure of the equation near $\rho = 0$:
\ben{simplified}
g = -   \frac{\partial}{\partial \rho}\left(\rho \frac{\partial u }{\partial \rho} \right)  -  2 s  \frac{\partial u}{\partial \rho} +  \ldots
\een
where we drop terms which are subleading near $\rho =0$. Differentiating this equation, we observe that \eq{simplified} is equivalent to a system of equations for $(u, u_\rho)$ which near $\rho=0$ have the structure:
\bean
g &=& -   \frac{\partial}{\partial \rho}\left(\rho \frac{\partial u }{\partial \rho} \right)  -  2 (s+1)  \frac{\partial u}{\partial \rho} + 2 u_{\rho} + \ldots
\\ \partial_{\rho}g &=& -   \frac{\partial}{\partial \rho}\left(\rho \frac{\partial u_{\rho} }{\partial \rho} \right)  -  2 (s+1)  \frac{\partial u_{\rho} }{\partial \rho} + \ldots
\eean
where $u_{\rho} = \partial_\rho u$. We have used this fact to modify the first equation by adding and subtracting the same term. The structure of this system near $\rho=0$ is the same as that of \eq{simplified}, but with $s \to s+1$. We also find that $u_\rho$ inherits a boundary condition at $\rho = 1$. As a result we can repeat all of the estimates above, but now for a system of equations, and we will conclude that
\be
\norm{u}{H^2}^2 \leq C' \norm{ (\hat{L}_s  + \gamma') u)}{H^1}^2
\ee
holds for some $C', \gamma'$, provided $\Re(s+1) > -\frac{1}{2}$. Repeatedly differentiating we obtain the result for any $k$.
\end{enumerate}
\end{proof}

There are three key ingredients in this proof which have interpretations in terms of the geometry of the black hole:
\begin{itemize}
\item The estimate \eq{first estimate} is related to the \emph{Killing energy estimate} associated to the stationary Killing vector. It is well known that this degenerates on the horizon.
\item The estimate \eq{second estimate} is related to the \emph{redshift estimate} for the black hole horizon \cite{Mihalisnotes}. This makes crucial use of the positivity of the surface gravity.
\item The fact that commuting the equation makes the structure at the horizon more favourable is due to the \emph{enhanced redshift effect} \cite{Mihalisnotes}.
\end{itemize}
These three properties are robust, in the sense that they are present for any non-extremal black hole, and it is this fact which we shall exploit when studying general black hole spacetimes.

We have in fact shown that for given $g\in H^{k-1}(I)$, the solution $\hat{\phi}(s)$ to \eq{regular operator} is meromorphic on $\Re(s) >(\frac{1}{2} - k)$. We can thus give a new definition of quasinormal frequencies to be the locations, $s_i$, of the poles in this half-plane, with quasinormal modes given by the corresponding solutions $w_i(s_i, \cdot)$ of the homogeneous problem. By taking $k$ progressively larger, we can define the quasinormal frequencies throughout the complex plane in this fashion. Where confusion is likely to arise, we shall refer to the QNF defined in this fashion as `regularity' quasinormal frequencies, and the QNF defined in the sense of \S\ref{traditional section} as `ingoing' quasinormal frequencies.

The advantage of the `regularity' definition over the `ingoing' one is that the associated quasinormal modes naturally live in $H^{k}(I)$. We can easily translate the problem of finding quasinormal modes in this context to an honest eigenvalue problem by writing $\hat{\phi} = \phi_1$ and introducing an auxiliary function $\phi_2$. Then $\hat{\phi}$ solves the homogeneous problem if and only if the vector $\bm{\phi} = (\phi_1, \phi_2) \in H^k(I) \times H^{k-1}(I)$ solves
\be
(\mathcal{A}-s) \bm{\phi} = 0,
\ee
where
\be
\mathcal{A}\bm{\phi} = \left(\begin{array}{cc}0 & 1 \\-P_2 & -P_1\end{array}\right)\left(\begin{array}{c}\phi_1 \\\phi_2\end{array}\right),
\ee
with
\begin{align*}
P_2\phi &=  -  \frac{\partial}{\partial \rho}\left(\rho(2-\rho) \frac{\partial  \phi }{\partial \rho} \right) +  \frac{\tilde{V}}{\rho(2-\rho)} \phi, \\
P_1 \phi &=   -(1-\rho)^2 \frac{\partial \phi}{\partial \rho}- \frac{\partial}{\partial \rho}\left( (1-\rho)^2 \phi \right).
\end{align*}

The operator $\mathcal{A}$ has a natural interpretation arising from the equation \eq{regular wave}. This equation naturally defines a semigroup $\mathcal{S}(t)$, which maps initial data $(\phi_0, \phi'_0) \in H^k(I) \times H^{k-1}(I)$ to the solution $(\phi(t, \cdot), \phi_t(t, \cdot))\in H^k(I) \times H^{k-1}(I)$ of \eq{regular wave} at time $t$. The operator $\mathcal{A}$ is the infinitesimal generator of this solution semigroup. This discussion suggests how we can extend the definition to more general set-up. Before we do so, we briefly relate our new definition of quasinormal modes to the previous one.

Let us note that $s_i$ is a `regularity' quasinormal frequency if and only if there exists a corresponding solution of the equation \eq{regular wave} which is smooth up to $\rho=0$ of the form:
\be
\phi(t, \rho) = e^{s_i t} w(s_i, \rho).
\ee
Recall that in  \S\ref{traditional section} we saw that at each ingoing quasinormal frequency, we found a solution of \eq{wscr} which for large $x$ has the form
\be
\psi(\tau, x) = e^{s_i (\tau-x)}, \qquad x>K
\ee
Transforming to $(\tau, \rho)$ coordinates, we see that for $\rho$ close to $0$ this solution maps to:
\be
\psi(\tau(t, \rho), x(t, \rho)) = e^{s_i (t+\rho-1)},
\ee
which certainly extends to a smooth function at $\rho=0$. Thus any `ingoing' quasinormal frequency is a `regularity' quasinormal frequency. The converse is not true. There may be extra `regularity' quasinormal frequencies at points $s\in -\mathbb{N}$ which do not correspond to quasinormal modes in the sense of  \S\ref{traditional section}. See \S \ref{ressec} for a general discussion of this issue.

\subsection{The general case}
\begin{figure}[t]
\begin{minipage}[b]{0.45\linewidth}
\centering
\input{adsbhsch2.tex}
\caption{Static slicing of AdS-Schwarzschild}
\label{fig1}
\end{minipage}
\hspace{0.8cm}
\begin{minipage}[b]{0.45\linewidth}
\centering
\input{adsbhsch.tex}
\caption{Regular slicing of AdS-Schwarzschild}
\label{fig2}
\end{minipage}
\end{figure}

As suggested by our simple example above, in the black hole case we wish to define the quasinormal modes in terms of the generator of the solution semigroup with respect to a regular slicing. In Figure \ref{fig1} we show the standard static slicing of AdS-Schwarzschild, corresponding to the $(\tau, x)$ coordinates above. The problems with the original definition of the QNM are all issues with the region where $x$ is large and $\tau$ is finite, i.e.\ in the neighbourhood of the bifurcation sphere of the horizon\footnote{$x$ is not the usual Schwarzschild radial coordinate, but is rather a tortoise coordinate}. In Figure \ref{fig2} we show a regular slicing of the horizon (which would correspond to the $(t, \rho)$ coordinates above). In this slicing, the energy is no longer conserved, but instead is bounded by its initial value. In fact, from the redshift argument \cite{Mihalisnotes, Holbnd}, we essentially have\footnote{We shall need to be more precise about the definition of the Sobolev spaces later to capture the boundary conditions at infinity, but we elide this subtlety for the purpose of the introduction.}
\ben{eng}
\norm{\psi}{H^{k}(\Sigma_t)}^2 + \norm{T{\psi}}{H^{k-1}(\Sigma_t)}^2 \leq C \left( \norm{\psi}{H^{k}(\Sigma_0)}^2+ \norm{T{\psi}}{H^{k-1}(\Sigma_0)}^2\right),
\een
where $T$ is the timelike Killing field. This implies that the solution operator for the wave equation is in fact a $C^0$-semigroup acting on $\mathbf{H}^k(\Sigma)=H^{k}(\Sigma) \times H^{k-1}(\Sigma)$.  A consequence of the Hille-Yoshida theorem (strictly, its converse) is that the generator of this semigroup\footnote{we abuse notation by using the same letter to refer to the generator regardless of which of the $\mathbf{H}^k$ spaces it is acting in.}, $\A:D^k(\A) \to \bm{H}^{k}(\Sigma)$, is a closed, unbounded, linear operator acting on a dense subspace of $\mathbf{H}^k(\Sigma)$. As such, the spectrum of $\A$ is well defined. 

We propose the following definition for quasinormal modes, see Appendix \ref{glossary} for the definition of various terms:
\begin{customdef}{\ref{QNSdef}}
Let $L$ be a strongly hyperbolic operator on a globally stationary black hole $\bhR$. Let $(D^k(\A),\A)$ be the infinitesimal generator of the solution semigroup on $\bm{H}^k(\Sigma)$ corresponding to a regular slicing of the black hole. We say that $s\in \C$ belongs to the $\bm{H}^k$-quasinormal spectrum (denoted $\Lambda^k_{QNM}$) of $L$ if:
\begin{enumerate}[i)]
\item $\Re(s) >\frac{1}{2}  w_L + (\frac{1}{2} - k)\varkappa$,
\item $s$ belongs to the spectrum of $(D^k(\A),\A)$.
\end{enumerate}
Here $\varkappa$ is the surface gravity and $w_L$ is a real number associated to $L$, defined in Definition \ref{RSdef}. If $s$ is an eigenvalue of $(D^k(\A),\A)$, we say $s$ is an $\bm{H}^k$-quasinormal frequency. The corresponding eigenfunctions are the $\bm{H}^k$-quasinormal modes.
\end{customdef}
The reason that we have to include the term $w_L$ is that by adding a first order term to the wave operator one can essentially shift the whole spectral plane to the left or right. The term $w_L$ takes account of this fact. For practical purposes one can assume $w_L=0$. This is in particular the case for the Klein-Gordon operator.

The justification for this definition will be one of the main theorems of the paper, the analogue of Lemma \ref{toy lemma} for the general case:
\begin{customthm}{\ref{mainthm}}
Let $L$ be a strongly hyperbolic operator on a globally stationary black hole $\bhR$, with boundary conditions fixed at infinity. Let $(D^k(\A),\A)$  be the infinitesimal generator of the solution semigroup on $\bm{H}^k(\Sigma)$ corresponding to a regular slicing of the black hole. Then for $s$ in the half-plane $\Re(s) >\frac{1}{2}  w_L + (\frac{1}{2} - k)\varkappa$, either:
\begin{enumerate}[i)]
 \item $s$ belongs to the resolvent set of $(D^k(\A),\A)$,
 \\
 or:
 \item $s$ is an eigenvalue of $(D^k(\A),\A)$ with finite multiplicity. 
 \end{enumerate}
Possibility $ii)$ holds only for isolated values of $s$. The resolvent is meromorphic on the half-plane $\Re(s) >\frac{1}{2}  w_L + (\frac{1}{2} - k)\varkappa$, with poles of finite rank at points satisfying $ii)$. The residues at the poles are finite rank operators.

Furthermore, if $k_1\geq k_2$, the eigenvalues and eigenfunctions of $(D^{k_1}(\A),\A)$, $(D^{k_2}(\A),\A)$ agree for $s$ in the half-plane $\Re(s) >\frac{1}{2}  w_L + (\frac{1}{2} - k_2)\varkappa$.
\end{customthm}
In particular, this implies that the set of all quasinormal modes
\be
\Lambda_{QNF} = \bigcup_{k=1}^\infty \Lambda^k_{QNF},
 \ee
 is a countable subset of $\C$, which accumulates only at infinity. 
 
What does this mean in practice? We shall see later that for a globally stationary black hole, the Klein-Gordon equation, after multiplication by a suitable function, may be decomposed as
\be
P_{2} \psi +P_{1} T\psi+TT\psi = 0,
\ee
where $P_i$ are differential operators on $\Sigma$ of order $i$. The $\mathbf{H}^k$-quasinormal frequencies are simply those values of $s$ with $\Re(s)>(\frac{1}{2}-k)\varkappa$ for which
\ben{qnmfdef}
P_2 u +sP_1u+s^2u = 0,
\een
admits a non-trivial solution $u\in H^{k}(\Sigma)$, obeying the boundary conditions at $\scri$. This definition has been used implicitly in many calculations of QNMs (e.g.\  \cite{Horowitz:1999jd}), however we believe the relation to semigroup theory for a regular slicing is novel. 

In the case of the massless wave equation, the operator appearing in \eq{qnmfdef} is essentially the stationary d'Alembert-Beltrami operator. It should be noted, however, that this depends on the spatial slicing chosen. In the Schwarzschild-AdS case the operator in \eq{qnmfdef} is \emph{not} the same as the operator one would obtain by Laplace (or Fourier) transforming the d'Alembert-Beltrami operator written in the usual Schwarzschild coordinates\footnote{This operator is the one used for example in \cite{ Dyatlov:2010hq, Gannot:2012pb}.}. This is where the choice of a regular slicing makes itself known. Contrast for example, \eq{lap1}, \eq{regular operator}.  

Our operator is elliptic away from the horizon, however the ellipticity degenerates on the horizon. The majority of the technical content of this paper consists of handling this degeneracy. The key observation is that the redshift estimates used to establish \eq{eng} can be `recycled' into elliptic estimates. We require both the `enhanced' redshift effect which comes from commuting with the redshift vector field as well as the standard redshift argument which makes use of the redshift vector field as a multiplier.

\subsection{Advantages of the regularity approach}

We discuss now how the regularity approach has advantages over the traditional approach, and in particular deals with our previous criticisms. As previously mentioned, this is not the only approach which can handle these issues \cite{ vasy10, barreto1997distribution, sabaretto}.
\begin{enumerate}[1.]
\item[1., 2.] The condition that QNM have a certain degree of regularity at the horizon is a very natural one. We can also deal with initial data which may not vanish on the horizon more easily with a regular slicing.
\setcounter{enumi}{2}
\item The QNM now belong to the Hilbert space $\mathbf{H}^k$, and are honest eigenfunctions of a closed operator in this Hilbert space
\item The QNM, multiplied by $e^{s_i t}$ represent genuine solutions of the equation, perfectly well behaved on the full extent of the spatial slice and as $t \to \infty$. As a result, a finite sum of such terms makes sense everywhere for $t\geq 0$, and we may hope to approximate the solution for late times, uniformly on $\Sigma$, in this way.
\item There is no meromorphic extension through the continuous spectrum required. In some sense, we are constructing the meromorphic extension directly as we increase $k$.
\item Since we seek eigenvalues of a (non-self adjoint) operator on a Hilbert space, there are many methods for explicit calculation, for example spectral approaches. These methods have been used implicitly in finding QNM for aAdS black holes \cite{Horowitz:1999jd}.
\item There is no assumption of separability of the wave equation required in defining the QNM in this way, and the boundary conditions are natural.
\item The definition requires only a future Killing black hole horizon. This is more physically reasonable than requiring a bifurcate horizon.
\end{enumerate}

A final advantage of the regularity approach is that we have access to standard results in semigroup theory which may be immediately applied. We give here two examples, adapted from the results of \cite{Borichev, batkai}.
\begin{Corollary}[adapted from \cite{Borichev} Theorem 1.1]\label{schcor}
Suppose that solutions of the Klein-Gordon equation, with boundary conditions fixed, on some asymptotically AdS black hole are bounded in $\H^1\times \L^2$. Furthermore suppose that there exists no quasinormal frequency on the imaginary axis (i.e. a purely oscillatory QNM). Then for any solution $\psi$ of the Klein-Gordon equation with initial data in $D^1(\A)$, we have
\be
\norm{\psi}{\H^1(\Sigma_t)} + \norm{T \psi}{\L^2(\Sigma_t)} \to 0, \qquad \textrm{ as }t \to \infty.
\ee
\end{Corollary}
Note that we do not get an explicit rate for this decay. An explicit rate can be obtained from a resolvent estimate on the imaginary axis, but this carries slightly more information than the quasinormal spectrum\footnote{The quasinormal spectrum permits us to estimate the spectral radius of $(\A-s)^{-1}$, whereas we require the operator norm. Since $\A$ is not self-adjoint, the two need not be related}. Note also that we require the data to be in $D^1(\A)$ in order to obtain decay in $\bm{H}^1(\Sigma)$. If this result were to hold for data in $\bm{H}^1(\Sigma)$, that would imply that the decay is in fact exponential. 

In Lemma \ref{hermlem} we shall show that for the Klein-Gordon equation there can be no quasinormal frequencies on the imaginary axis, except possibly at the origin. If the Klein-Gordon mass is non-zero then there can be no QNF at the origin either. Combining this with the result above and Theorems 1.2, 1.3 of \cite{Holzegel:2012wt} we are able to `upgrade' the boundedness results of that paper to local energy decay, albeit without a quantitative rate:
\begin{Corollary}
Let $\psi$ solve the Klein-Gordon equation with non-zero mass:
\be
\Box_g \psi + \mu \psi = 0, \qquad \mu \neq 0
\ee
on the Schwarzschild-AdS exterior. Suppose further that $\psi$ satisfies homogeneous Dirichlet or Neumann boundary conditions at $\scri$. Then $\psi$  has locally decaying energy in the sense that: 
\be
\norm{\psi}{\H^1(\Sigma_{t_*})} + \norm{T \psi}{\L^2(\Sigma_{t_*})} \to 0, \qquad \textrm{ as }t \to \infty,
\ee
for $t_*$ the usual time coordinate which is regular on $\hor$ (see for example \cite{Holbnd}). In the case $\mu =0$, $\psi$ decays to a constant solution.

The same holds for Kerr-AdS backgrounds satisfying the Hawking-Reall bound $r_+^2>\abs{a} l$ and for $l^{-1} \abs{a}$ sufficiently small. In the Dirichlet case only $\abs{a}<l$ is necessary.
\end{Corollary}
 Note that for the Dirichlet case, the results of Holzegel and Smulevici \cite{HolSmul} give local energy decay with a sharp asymptotic rate of $1/\log t_*$. The step from qualitative decay to sharp quantitative decay is not trivial, and requires a detailed understanding of trapping in the spacetime.

We finally state an estimate for the `best possible' rate in terms of properties of $\Lambda_{QNF}$:
\begin{Corollary}[adapted from \cite{batkai} Propositions 3.3, 3.7]
Suppose that solutions of the Klein-Gordon equation, with boundary conditions fixed, on some asymptotically AdS black hole, are bounded in $\H^1\times \L^2$. Furthermore suppose that there exists a sequence $s_n$ of quasinormal frequencies of the Klein-Gordon equation with $\abs{\Im(s_n)} \to \infty$ as $n \to \infty$ and such that for some $C$
\be
 -C \abs{\Im(s_n)}^{-\frac{1}{\alpha}} <  \Re(s_n) \leq 0.
\ee
Then for any $\epsilon>0$, no bound of the form
\be
\norm{(\psi, T \psi)}{\L^2(\Sigma_t)\times \H^1(\Sigma_t)}\leq \frac{C}{t^{\alpha+\epsilon}}\norm{\A(\psi, T \psi)}{\L^2(\Sigma_0)\times \H^1(\Sigma_0)} , 
\ee
can hold for all solutions $\psi$ of the Klein-Gordon equation with initial data in $D^1(\A)$, with $C$ independent of $\psi$.
\end{Corollary}
In other words, the more rapidly the QNFs approach the imaginary axis, the slower the decay one may obtain. This is in fact the case for the Schwarzschild-AdS black hole, as recent work has shown \cite{Gannot:2012pb,HolSmul,Holzegel:2013kna}. In this case the QNFs approach the imaginary axis faster than any polynomial, and the best asymptotic rate that one may obtain is $1/\log{t_*}$. See also \cite{Cardoso:2004up,Natario:2004jd,Harmark:2007jy} for results in the physics literature concerning the asymptotic distribution of quasinormal modes in asymptotically AdS spacetimes.

\subsection{Outline of the paper} In \S\ref{prelim}, we set up the problem, in particular introducing the globally stationary asymptotically AdS black holes, together with the strongly hyperbolic operators we shall study on these backgrounds. In \S\ref{hypsec} we prove some estimates for the full time-dependent problem. In particular we give a new account of the redshift effect which simplifies some previous approaches. We finish this section with the proof that the solution operators for the strongly hyperbolic operators generate a $C^0$-semigroup, and use this to define the quasinormal spectrum. In \S\ref{QNMsec} we prove Theorem \ref{mainthm}, making considerable use of the estimates from \S\ref{hypsec}. In \S\ref{locstat} we generalise the results to the case of black holes which are merely \emph{locally} stationary. In this case we recover the same results restricted, however, to irreducible representations of the axial symmetry group. In \S\ref{example} we develop a simple example which directly shows how our approach may be applied. It furthermore provides an example where the quasinormal modes are manifestly not `complete' in any sense, and shows that QNM with `ingoing' boundary conditions are insufficient to describe the late time behaviour of solutions whose initial data is non-trivial on the horizon. In \S\ref{ressec} we make explicit the connection to the usual definition of QNM in terms of resonances. In particular, we show how our result implies the meromorphicity of the resolvent for the Schwarzschild-AdS black hole. We also show how our results may be applied to prove the meromorphicity of the resolvent of the Laplacian for asymptotically hyperbolic manifolds. In Appendix \ref{oscapp} we prove that for hermitian operators (such as the Klein-Gordon operator) there can be no purely oscillatory quasinormal modes, a result necessary for Corollary \ref{schcor}. In Appendix \ref{supplapp} we prove some results that, while not central to our presentation here, do not appear in a suitable form in previous work.  In Appendix \ref{glossary} we provide a glossary of symbols defined throughout the paper.

\section{Preliminaries}\label{prelim}
\subsection{Asymptotically AdS spacetimes; well posedness for the IBVP}

We first define what we mean by an `asymptotically anti-de Sitter Spacetime'.
\begin{Definition} \label{adsenddef}
Let $\mathscr{M}$ be a $(d+1)$-dimensional manifold with boundary\footnote{We take the convention that $\mathscr{M}$ includes $\partial \mathscr{M}$ as a point set, and $\mathring{\mathscr{M}} = \mathscr{M}\setminus \partial \mathscr{M}$} $\partial \mathscr{M}$, and $g$ be a smooth Lorentzian\footnote{Throughout we take the mostly plus signature convention} metric on $\mathring{\mathscr{M}}$. We say that a connected component $\scri$ of $\partial \mathscr{M}$ is an \emph{asymptotically anti-de Sitter end of $(\mathring{\mathscr{M}}, g)$ with radius $l$} if:
\begin{enumerate}[i)]
\item There exists a smooth function $r$ such that $r^{-1}$ is a boundary defining function for $\scri$.
\item For $r$ sufficiently large, we can choose $x^\alpha$, coordinates on the slices $r=\textrm{const.}$  so that locally
\be
g_{rr} = \frac{l^2}{r^2}+ \O{\frac{1}{r^4}},\qquad g_{r\alpha} = \O{\frac{1}{r^2}}, \qquad g_{\alpha\beta} = r^2 \mathfrak{g}_{\alpha \beta} + \O{1},
\ee
where $\mathfrak{g}_{\alpha\beta}dx^\alpha dx^\beta$ is a Lorentzian metric\footnote{with signature $(-+++\ldots)$. In particular the induced metric on $\scri$ is not Riemannian.} on $\scri$.
\item $r^{-2} g$ extends as a smooth metric on a neighbourhood of $\scri$.
 \end{enumerate}
We say that $r$ is an asymptotic radial coordinate and $\scri$ is the conformal infinity of this end. We denote by $\mathfrak{X}(\mathscr{M})$ the space of smooth vector fields on $\mathscr{M}$. If $\Sigma$ is any submanifold of $\mathscr{M}$, then $\mathfrak{X}_\Sigma(\mathscr{M})$ consists of those smooth vector fields on $\mathscr{M}$ tangent to $\Sigma$. Finally, $\mathfrak{X}^*(\mathscr{M})$ is the space of smooth one-form fields on $\mathscr{M}$.
\end{Definition}

\subsubsection{Well Posedness} \label{WPsec}
We shall require a slightly more general well-posedness theorem than that proven in \cite{Warnick:2012fi}. In particular we will work with a multi-component complex scalar field.  The proof of \cite{Warnick:2012fi} readily extends to this case.  It is possible to consider linear fields defined on more complicated bundles, such as $\Omega^2(\mathscr{M})$ in the case of the Maxwell field or $\textrm{Spin}(\mathscr{M})$ for the Dirac field by choosing a set of reference sections with respect to which a general section may be written as a linear sum. The coefficients of this linear sum are then scalar quantities. 

We will be interested in equations of the form\footnote{Our convention is that $\Box_g = -g^{\mu \nu}\nabla_\mu \nabla_\nu$}
\ben{wpeq}
\Box_g \psi + K[\psi] + V\psi = 0,
\een
where $\psi$ takes values in $\C^N$. Here $V$ is a $C^\infty$-smooth matrix valued potential, which near infinity has an expansion:
\ben{vasy}
V^I{}_{J} = \frac{1}{l^2}\left( \kappa_I^2-\frac{d^2}{4}  \right) \delta^I{}_{J} + \O{r^{-2}}, \qquad \textrm{ [no sum over }I].
\een
We assume that $0<\kappa_1\leq \kappa_2 \leq \ldots \leq \kappa_N$ are real, and we define
\ben{kappadef}
\kappa = \textrm{diag} (\kappa_1, \ldots, \kappa_N).
\een
$K$ is a $C^\infty$-smooth matrix valued vector field, given locally by:
\be
K[\phi]^I = \sum_{J=1}^N(K^\mu)^I{}_J{} \nabla_\mu \phi^J,
\ee
and near infinity we have, in the coordinates of Definition \ref{adsenddef}: 
\ben{kasy}
(K^r)^I{}_J = \O{r^{-1}}, \quad (K^\alpha)^I{}_J = \O{r^{-2}}.
\een
This is equivalent to the statement that $r^{2} K$ belongs to $\mathfrak{X}_\scri(\mathscr{M})$.

Now, let us define a matrix version of the twisted derivatives of \cite{Warnick:2012fi,Holzegel:2012wt}. We define the twisting matrix by:
\ben{fdef}
f := \exp\left[ \left (-\frac{d}{2} \iota+\kappa\right )\log r \right],
\een
where $\iota$ is the $N\times N$ identity matrix. Then the twisted derivatives are defined  to be
\ben{twist}
\tn_\mu \phi := f \nabla_\mu \left(f^{-1} \phi \right), \qquad \tn_\mu^\dagger \phi := -f^{-1} \nabla_\mu \left(f \phi \right),
\een
where we understand $f, f^{-1}$ to act by matrix multiplication. If $W\in \mathfrak{X}(\mathscr{M})$ is any vector field, we define
\be
\tilde{W} = W^\mu \tn_\mu.
\ee

\begin{Definition} \label{wpdefs}
Let $(\mathring{\mathscr{M}}^{d+1}, g)$ be a time oriented Lorentzian manifold with an asymptotically AdS end, with asymptotic radial coordinate $r$ which we assume extends as a smooth positive function throughout $\mathring{\mathscr{M}}^{d+1}$. Let $\Sigma$ be a spacelike surface which extends to the conformal infinity of the asymptotically AdS end, $\scri$, which for convenience we assume meets $\scri$ orthogonally with respect to the conformal metric $r^{-2} g$. Let $n_\Sigma$ be the future directed unit normal of $\Sigma$ and define
\be
\hat{n}_\Sigma = r n_\Sigma,
\ee  
to be the rescaled normal. We define $\mathcal{D}^+(\Sigma)$ to be the set of points $p\in \mathring{\mathscr{M}}$ such that every past directed inextensible timelike curve either intersects $\Sigma$ or else approaches a point on $\scri\cap I^+(\Sigma)$.

Let $f$ be as in \eq{fdef} and let $\phi : \Sigma \to \C^N$, where $\C^N$ is equipped with the usual sesquilinear product. We define the norms
\bean
\norm{\phi}{\L^2(\Sigma)}^2 &=& \int_\Sigma  \frac{\abs{\phi}^2}{r} dS_{\Sigma},  \\
\norm{\phi}{\H^1(\Sigma, \kappa)}^2 &=& \int_\Sigma \left( |\tilde{\nabla} \phi |^2 + \frac{\abs{\phi}^2}{r^2}\right) r dS_{\Sigma} \, ,
\eean
where the twisted derivative is defined as in (\ref{twist}) and we use the induced metric on $\Sigma$, together with the norm on $\C^N$ to define $|\tilde{\nabla} \phi |^2$ and $dS_\Sigma$. The derivatives here are understood in a weak sense. If $\mathcal{D} \subset \{1, \ldots, N\}$, we denote by $\H^1_\mathcal{D}(\Sigma, \kappa)$ the completion in the $\H^1(\Sigma, \kappa)$ norm of the space of smooth functions $\phi : \Sigma \to \C^N$ where $\phi^I$ is supported away from $\scri$ for $I \in \mathcal{D}$.

Given a $C^1$ function $\phi:\mathring{X} \to \C^N$, we say that the $I^{th}$ component obeys Dirichlet or Robin boundary conditions if the following hold
\begin{enumerate}[i)]
\item Dirichlet:
\be
r^{\frac{d}{2} - \kappa_I} \phi^I \to 0, \quad \textrm { as } r \to \infty.
\ee
\item Robin\footnote{This includes the Neumann boundary conditions where $\beta_I=0$.}: $0<\kappa_I<1$ \emph{and}
\be
r^{\frac{d+2}{2} + \kappa_I}\, \tilde{\nabla}_r \phi^I + l \beta_I r^{\frac{d}{2} - \kappa_I} \phi^I \to 0, \quad \textrm { as } r \to \infty,
\ee
where $\beta_I \in C^\infty(\scri)$.
\end{enumerate}
A choice of homogeneous boundary conditions for $\phi$ is a subset $\mathcal{D}\subset \{1, \ldots, N\}$  together with a choice of $\beta_I \in C^\infty(\scri)$, for $I \in  \{1, \ldots, N\}\setminus \mathcal{D}$. We require that the $I^{th}$ component of $\phi$ obeys Dirichlet boundary conditions for $I \in \mathcal{D}$ and Robin boundary conditions with Robin function $\beta_I$ for $I \not \in \mathcal{D}$.

Given a choice of homogeneous boundary conditions, we denote by $C^\infty_{bc}(\bhR; \C^N)$ the set of functions of the form $\phi = e^{(-\frac{d}{2}+\kappa)\log r} \phi_++e^{(-\frac{d}{2}-\kappa)\log r} \phi_-$ where $\phi_{\pm} \in C^\infty(\bhR; \C^N)$, with $\phi_+^I = 0$ if $I \in \mathcal{D}$ and $2\kappa_I \phi_-^I - \beta_I \phi_+^I=0$ on $\scri$ if $I \not \in \mathcal{D}$.

\end{Definition}

Note that Robin boundary conditions may only be imposed for those components with $0<\kappa_I<1$. This is an extension of a similar restriction for a single real scalar field, where it is known that Neumann or Robin boundary conditions may only be chosen for a certain range of masses, see \cite{BF, Warnick:2012fi}.

\begin{Theorem}[Well Posedness]\label{WPThm}

Fix a choice of homogeneous boundary conditions for $\psi$. Let $\uppsi\in \H^1_\mathcal{D}(\Sigma, \kappa)$, $\uppsi' \in \L^2(\Sigma)$. Then there exists a unique $\psi:\mathcal{D}^+(\Sigma) \to \C^N$ which solves
\be
\Box_g \psi + K[\psi] + V\psi = 0, \qquad \psi|_{\Sigma} = \uppsi, \quad \hat n_\Sigma \psi|_\Sigma = \uppsi',
\ee
and satisfies the boundary conditions on $\scri$. Here both the equation and boundary conditions are to be understood in a weak sense\footnote{For full details of the weak formulation of the equation and boundary conditions, see \cite{Warnick:2012fi}.}. If ${\Sigma'}$ is any spacelike surface in $\mathcal{D}^+(\Sigma)$ then $\psi|_{\Sigma'} \in \H^1_\mathcal{D}({\Sigma'}, \kappa)$, $\hat n_{\Sigma'}\psi|_{\Sigma'} \in \L^2({\Sigma'})$ and we have a continuity estimate:
\be
\norm{\psi|_{\Sigma'} }{ \H^1({\Sigma'}, \kappa)} +\norm{\hat n_{\Sigma'}\psi|_{\Sigma'} }{ \L^2({\Sigma'})} \leq C({\Sigma'}, g, K, V) \left(\norm{\uppsi}{ \H^1(\Sigma, \kappa)} +\norm{\uppsi' }{ \L^2(\Sigma)} \right).
\ee

If the initial conditions satisfy stronger regularity and asymptotic conditions, then in fact $\psi|_{\Sigma'} \in H^k_{\textrm{loc.}}({\Sigma'})$, $\hat n_{\Sigma'}\psi|_{\Sigma'} \in   H^{k-1}_{\textrm{loc.}}({\Sigma'})$ for any integer  $k\geq 2$. Furthermore, the boundary conditions hold in the sense of traces\footnote{in fact, for Dirichlet components we get a stronger result that $\psi^I \sim r^{-\frac{d}{2}-\kappa_I}$ if $0<\kappa_I<1$, and $\abs{\psi^I} \lesssim  r^{-\frac{d}{2}-1+\epsilon}$ for any $\epsilon>0$ if $\kappa_I \geq 1$.}. For sufficiently regular initial data we obtain a classical solution to the initial boundary value problem.
\end{Theorem}

We will also require a twisted version of the energy momentum tensor which extends the definition of \cite{Holzegel:2012wt} to allow for a multi-component field.

\begin{Definition}\label{emtdef}
We define the twisted energy-momentum tensor associated to \eq{wpeq} to be the symmetric $2-$tensor
\ben{twem}
\emT_{\mu\nu}[\phi] = \tn_{(\mu} \overline{\phi} \cdot \tn_{\nu)} \phi - \frac{1}{2} g_{\mu \nu} \left(\tn_\sigma  \overline{\phi} \cdot \tn^\sigma \phi + \overline{\phi}\cdot  F\phi \right),
\een
where the matrix valued function $F:\mathscr{M} \to \mathbb{M}^{n\times n}$ is given by
\be
F =  \frac{1}{2} (V+V^*)- f^{-1}\nabla_\mu \nabla^\mu f.
\ee
where $V^*$ is the hermitian conjugate of $V$. From the definition of $f$ and the asymptotics of $V$, we have $r^2 F\in C^\infty(\mathscr{M})$.
\end{Definition}
The twisted energy momentum enjoys the following important property:

\begin{Lemma}[Divergence of $\emT_{\mu\nu}$]\label{emprop}
For a sufficiently regular $\phi$, we have
\ben{divt}
\nabla_\mu \emT^\mu{}_\nu[\phi] = -\Re\left[\tn_\nu \overline{\phi}\cdot \left(\Box_g\phi+ \frac{1}{2} (V+V^*) \phi  \right) \right]+ \emS_\nu[\phi],
\een
where
\be
2 \emS_\nu[\phi] = \Re\left\{ \overline{\phi} \cdot \left[ \tn^\dagger_\nu ( F f) f^{-1} \phi\right]+  (\tn_\sigma \overline{\phi}) \cdot \left[f^{-1} (\tn^\dagger_\nu  f)  (\tn^\sigma \phi)\right]\right \}.
\ee
\end{Lemma}
\begin{proof}
The proof follows by a straightforward calculation, which we include in \S\ref{Lemma appendix}.
\end{proof}

\subsection{Asymptotically AdS Black Holes}
We now define what we mean by a stationary aAdS black hole:
\begin{Definition}\label{adsbh}
We say that $(\bhR, \hor, \scri, \Sigma, g, r, T)$ is a globally stationary, asymptotically anti-de Sitter, black hole space time with AdS radius $l$ if the following holds
\begin{enumerate}[i)]
\item $\bhR$ is a $(d+1)$-dimensional manifold with stratified boundary $\hor\cup \Sigma\cup \scri$, where $\Sigma$ is a compact manifold whose boundary has two components: $\hor\cap \Sigma$ and $\Sigma\cap\scri$ which are compact, connected, manifolds.
\item $g$ is a smooth Lorentzian metric on $\bhR\setminus \scri$.
\item $\scri$ is an asymptotically AdS end of $(\mathring{\bhR}, g)$ of AdS radius $l$, asymptotic radial coordinate $r$ and such that
\be
\bhR = \mathcal{D}^+(\Sigma).
\ee 
We assume $r$ extends to a smooth positive function throughout $\mathring{\bhR}$.
\item $\Sigma$ is everywhere spacelike with respect to $g$, whereas $\hor$ is null.
\item $T$ is a Killing field of $g$ which is normal to $\hor$,  transverse to $\Sigma$ and timelike in a punctured neighbourhood of $\hor$, thus $\hor$ is a Killing horizon generated by $T$, which we assume to be a non-extremal black hole horizon, i.e. it should have constant surface gravity $\varkappa$.
\item $T\in\mathfrak{X}_\scri(\bhR)$, as a result, w.l.o.g. we may assume\footnote{By this we mean that we may always choose an asymptotic radial coordinate $r$ which satisfies $\mathcal{L}_T r=0$ by fixing such an $r$ on $\Sigma$ and defining $r$ elsewhere by $\mathcal{L}_T r=0$.} $\mathcal{L}_T r=0$. 
 \item If $\varphi_t$ is the one-parameter family of diffeomorphisms generated by $T$, then $\bhR$ is smoothly foliated by  $\varphi_t(\Sigma):=\Sigma_t$, $t\geq 0$.
 \item $T$ is timelike on $\bhR \setminus \hor$.
\end{enumerate}
\end{Definition}
Figure \ref{Pendiag} shows a schematic Penrose diagram of the black hole exterior region, $\bhR$.

\emph{Remarks: } we note that the assumption that there is exactly one Killing horizon and one asymptotically anti-de Sitter end is not necessary for our results. In particular, we can assume that $\bhR$ is bounded by $\Sigma$ and one or two non-extremal Killing horizons. In this way, our results may be readily extended to asymptotically de Sitter spacetimes. The compactness of $\Sigma$ is, however, an important feature which cannot be easily relaxed, so our results do not extend to the asymptotically flat case in an obvious fashion. 

The assumption in $v)$ that $T$ is timelike in a neighbourhood of $\hor$ is superfluous, given $viii)$. It is convenient to have the definition stated in this way so that when we later relax assumption $viii)$ to consider black holes which are merely locally stationary in \S\ref{locstat} we only have to make minimal changes. As a matter of terminology, prior to \S\ref{locstat} we may refer to globally stationary black holes as simply stationary.

\begin{figure}[t]
\input{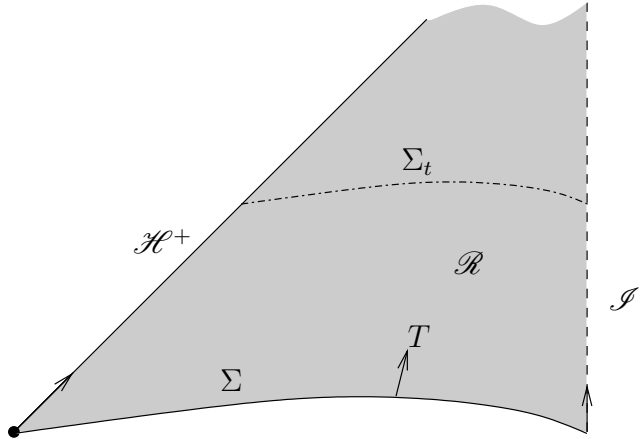}
\caption{A schematic Penrose diagram for a stationary, asymptotically anti-de Sitter, black hole spacetime.\label{Pendiag}}
\end{figure}

Now, as in \cite{Holzegel:2012wt}, we define some geometrical objects associated to the stationary slicing defined by $\Sigma_t$. Firstly, we denote
\ben{sigA}
\sigma:=-g(T,T), \qquad A:= -\frac{1}{g^{-1}(dt, dt)}.
\een
The asymptotic conditions imply that $\sigma, A \sim r^2$ near $\scri$. The definition also implies that $A>0$ on $\bhR$, while $\sigma\geq0$, with $\sigma$ vanishing precisely on $\hor$. The condition that $\hor$ is a non-extremal black hole horizon generated by $T$ implies that on the horizon\footnote{For a function $u$, the gradient $\nabla u$ is a vector field defined by $\nabla u[f] = g^{-1}(du, df)$.}
\ben{sgdef}
\nabla \sigma = 2 \varkappa T, 
\een
for some $\varkappa>0$, constant on the horizon. We refer to $\varkappa$ as the \emph{surface gravity}. We consider the constancy of $\varkappa$ to be part of our definition, however this follows automatically in the case when $(\mathscr{R}, g)$ satisfies Einstein's equations with matter obeying the dominant energy condition \cite[\S12.5]{Wald}. We believe that a modified version of our results will persist for non-constant $\varkappa$, provided that we have bounds $0<\underline{\varkappa}\leq \varkappa \leq \overline{\varkappa} <\infty$ for constants $\underline{\varkappa},\overline{\varkappa}$. Since the physically relevant case has constant $\varkappa$, we shall not pursue this further.

Rather than make use of equation \eq{wpeq} as it stands, it's convenient to consider the operator obtained after multiplying by $A$, so that the $TT\phi$ term appears with unit coefficient. Thus we define

\begin{Definition} \label{LLdef}
We say that an operator $L$ acting on functions $\phi \in C^2(\mathring{\bhR}; \C^N)$ is a strongly hyperbolic operator on $\bhR$ if it takes the form: 
\ben{Ldef}
L \phi := A (\Box_g \phi + V\phi) +W[\phi],
\een
where $V$ satisfies the assumptions of \S\ref{WPsec} and  $W$ is a  matrix-valued vector field which belongs to $\mathfrak{X}_\scri(\bhR)$. It is further assumed that $\mathcal{L}_T V = 0$, $\mathcal{L}_T W = 0$. 

We define the adjoint operator to $L$ to be
\ben{Ldefdual}
L^\dagger \phi := A (\Box_g \phi + V^*\phi) -W^*[\phi] - A\textrm{div}_g\left( \frac{W^*}{A}\right) \phi.
\een
$L^\dagger$ is itself a strongly hyperbolic operator on $\bhR$.
\end{Definition}
We will also consider $L$ acting on larger spaces of functions, with the derivatives understood to act in a weak sense. After dividing by $A$ the equation $L\phi=0$ is equivalent to an equation of the form \eq{wpeq}, so the well posedness results of \S2 apply.

\subsubsection{AdS-Schwarzschild} \label{schwarzschild section}
To justify the definitions we have given in this section, we will show that they apply in particular to the well known AdS-Schwarzshild black hole. We first define the manifold with stratified boundary to be:
\be
\mathscr{R} = [0, \infty)_t \times [0, 1]_\rho \times S^2_{\omega},
\ee
and identify the boundary components as:
\bean
\Sigma &=& \{0\} \times [0, 1] \times S^2, \\
\hor &=& [0, \infty) \times \{0 \} \times S^2, \\
\scri &=& [0, \infty) \times \{1\} \times S^2.
\eean
Clearly $\Sigma$ is compact, as are $\Sigma \cap \hor\simeq \Sigma \cap \scri \simeq S^2$.

Now, let $r_{+}$ be the largest root of the cubic
 \be
 x-2 m +\frac{x^3}{l^2}=0,
\ee
and define the function $r : \mathscr{R} \setminus \scri$ by:
\be
r(t, \rho, \omega) = \frac{1}{1-\rho} + r_+-1.
\ee
We may take $(t, r, \omega) \in [0, \infty)_t\times [r_+, \infty)_r \times S^2_\omega$ as coordinates\footnote{of course we will require two coordinate patches to cover the sphere, but for clarity we will ignore this issue.} on $\mathscr{R} \setminus \scri$, and with respect to these coordinates we define
 the metric
\ben{gdd}
g = -\left( 1-\frac{2 m}{r} +\frac{r^2}{l^2}\right)dt^2 + \frac{4m}{r}\frac{1}{1+\frac{r^2}{l^2}} dr dt + \frac{1+\frac{2 m}{r} +\frac{r^2}{l^2}}{\left(1+\frac{r^2}{l^2} \right)^2}dr^2 + r^2 d\Omega_2^2,
\een
where $d\Omega_2^2$ is the canonical metric on the unit $2$-sphere and $m$, $l$ are constants. Taking $T = \frac{\partial}{\partial t}$ in these coordinates, we clearly see that $T\in \mathfrak{X}_\scri (\bhR)$, $T$ is timelike away from $\hor$ and becomes null on $\hor$. We can conclude:
\begin{Lemma}
The AdS-Schwarzschild black hole $(\bhR, \hor, \scri, \Sigma, g, r, T)$ is a globally stationary, asymptotically anti-de Sitter, black hole space time with AdS radius $l$, and surface gravity
\be
\varkappa = \frac{2 m}{r_+^3}\left( \frac{3 m -r_+}{1+\frac{r_+^2}{l^2}}\right).
\ee
\end{Lemma}
We have included the AdS-Schwarzschild example as it is straightforward to state the metric. However, our definition of globally stationary, asymptotically anti-de Sitter, black hole space times also includes the slowly rotating Kerr-AdS black hole, see \cite[\S 5]{Holzegel:2012wt} for a definition. 

With the problems that we have in mind, an important result is:
\begin{Lemma}
A sufficiently smooth function $\psi$ obeys the Klein-Gordon equation in the AdS-Scwarzschild background
\be
\Box_g \psi + \frac{a}{l^2} \psi = 0,
\ee
if and only if it satisfies the equation
\be
L\psi = 0,
\ee
for a strongly hyperbolic operator on $\bhR$ with $W =0$, $V = \frac{a}{l^2}$. The condition under which Neumann or Robin boundary conditions may be imposed, $0<\kappa<1$, translates to the Breitenlohner-Freedman condition $-\frac{9}{4}<a <-\frac{5}{4}$.
\end{Lemma}

\section{The hyperbolic problem}\label{hypsec}

We will now discuss the regular wave semigroup approach to quasinormal modes for the equation \eq{wpeq}. We first derive some estimates which hold for any suitably smooth functions on $\bhR$, a stationary asymptotically anti-de Sitter black hole. From these we shall deduce that the solution operator of \eq{wpeq} is a $C^0$ semigroup on a family of Hilbert spaces. Our approach to deriving this fact may appear to be somewhat indirect, however, in the process we derive some estimates which will be extremely useful later on.

\subsection{Hyperbolic estimates}
All of our estimates are obtained by the vector field method. The vector field method usually consists of applying the divergence theorem to suitably chosen vector fields, known as energy currents, integrated over a spacetime slab. For our purposes, it will be useful to work with a form of the divergence theorem where the integration is performed only over a spacelike slice.
\begin{Lemma}[Local-in-time divergence theorem] \label{litdivlem}
Suppose $K$ is a vector field on $\bhR$, an aAdS black hole spacetime. Then for sufficiently regular $K$ the following relation holds
\ben{litdiv}
\frac{d}{dt} \int_{\Sigma_t} K^\mu n_\mu\ dS - \int_{\scri \cap \Sigma_t} K^\mu m_\mu \ \sqrt{A}d \mathcal{K} = -\int_{\hor\cap \Sigma_t} K^\mu T_\mu\ d\sigma - \int_{\Sigma_t} \nabla_\mu K^\mu\ \sqrt{A} dS.
\een
Here $dS$ is the induced volume form  and $n$ the future directed unit normal on $\Sigma_t$; $d\sigma$ the induced volume form on $\hor\cap \Sigma_t$ and $T$ the future directed Killing generator of $\hor$ normalised by $dt(T)=1$. We understand the integral over $\scri \cap \Sigma_t$ to be defined as follows
\ben{scrint}
\int_{\scri \cap \Sigma_t} K^\mu m_\mu\ \sqrt{A} d \mathcal{K}  = \lim_{r \to \infty} \int_{\{r=\textrm{const.}\} \cap \Sigma_t} K^\mu m_\mu\ \sqrt{A}d \mathcal{K},
\een
where $d \mathcal{K}$ is the induced volume element on  $\{r=\textrm{const.}\} \cap \Sigma_t$ and $m$ is the outward normal unit vector tangent to $\Sigma_t$.
\end{Lemma}
\begin{proof}
This can be shown either directly from the formula for the divergence of a vector field in local coordinates, or else by differentiating with respect to time the usual divergence theorem as applied to the region bounded by $\Sigma_t$, $\Sigma_0$, $\hor$ and $\scri$.
\end{proof}

Before we begin to derive estimates for the equation, we first derive an adjointness relation between $L$ and $L^\dagger$ that justifies our choice of notation.
\begin{Corollary}[Adjointness relation]\label{adjrel}
Let $L$ be a strongly hyperbolic operator on a black hole background $\bhR$. Define an inner product on the spacelike slice $\Sigma_t$ by
\be
\pair{\phi_1}{ \phi_2}_{\Sigma_t} = \int_{\Sigma_t} \overline{\phi_1} \cdot \phi_2 \frac{1}{\sqrt{A}} dS.
\ee 
This gives rise to a norm equivalent to the $\L^2(\Sigma_t)$ norm. Furthermore, defining
\be
K_\mu = \overline{\phi_1}\cdot (\tn_\mu \phi_2) -  (\tn_\mu \overline{\phi_1})\cdot \phi_2 - \frac{1}{A} \overline{\phi_1}\cdot W_\mu \phi_2,
\ee
we have that for $\phi_1, \phi_2\in C^2(\mathring{\bhR}; \C^N)$:
\be
\pair{\phi_1}{L\phi_2}_{\Sigma_t} -\langle L^\dagger \phi_1, \phi_2\rangle_{\Sigma_t} = \frac{d}{dt} \int_{\Sigma_t} K^\mu n_\mu\ dS - \int_{\scri \cap \Sigma_t} K^\mu m_\mu \ \sqrt{A}d \mathcal{K} +\int_{\hor\cap \Sigma_t} K^\mu T_\mu\ d\sigma. 
\ee
\end{Corollary}
\begin{proof}
We simply apply Lemma \ref{litdivlem} to the vector field $K_\mu$, which gives the result after a brief calculation.
\end{proof}

\subsubsection{The Killing estimate} The first estimate we shall derive from \eq{litdiv} will control the rate of growth of the energy associated to the vector field $T$. Since $T$ becomes null on the horizon this will only give us control of $L^2$ norms of the field and its derivatives tangential to the horizon, i.e. we do not control a transverse derivative at the horizon. We will later correct this oversight with a second estimate, the \emph{redshift estimate}, which will control the full $H^1$ norm at the horizon.

\begin{Definition}\label{Kildef}
Fix homogeneous boundary conditions for the strongly hyperbolic operator $L$, as in Definition \ref{wpdefs}, and assume that the Robin functions are stationary, i.e. $\mathcal{L}_T \beta_I = 0$. For a function $\phi\in C^\infty_{bc}(\bhR; \C^N)$, we define:
\be
(\curJ^T_\gamma)^\mu[\phi] = T^\nu \emT_\nu{}^\mu[\phi] - \frac{1}{2A} \gamma \abs{\phi}^2 T^\mu.
\ee
Here $\gamma$ is a constant which we will fix later. We also define the Killing energy on $\Sigma_t$ to be
\be
E_\gamma(t)[\phi] = \int_{\Sigma_t} (\curJ^T_\gamma)^\mu[\phi] n_\mu\ dS + \frac{1}{2}\sum_{I \in \{1, \dots, N\}\setminus \mathcal{D}} \int_{\scri \cap \Sigma_t}  \overline{\phi}^I \phi^I \beta_I r^{-2\kappa_I} \sqrt{A} d \mathcal{K}.
\ee
The powers of $r$ are such that the second integral, defined in the same fashion as \eq{scrint}, is finite.
\end{Definition}

\begin{Theorem}[The Killing estimate]\label{killthm}
For $\phi\in C^\infty_{bc}(\bhR; \C^N)$, we have:
\begin{enumerate}[i)]
\item Given $\gamma \geq 0$, there exists $c>0$ independent of $\gamma$ such that
\be
E_\gamma(t)[\phi] \leq  c \left( \norm{\phi}{\H^1(\Sigma_t, \kappa)}^2 + \norm{T\phi}{\L^2(\Sigma_t)}^2+ \gamma \norm{\phi}{\L^2(\Sigma_t)}^2\right),
\ee
for all $\phi$.
\item There exists $\gamma_0$ such that for any $\gamma > \gamma_0$ and for any $X \in \mathfrak{X}_\hor(\bhR)$ we can find $C_{X, \gamma_0}>0$ independent of $\gamma$ such that
\be
\norm{\tilde{X} \phi}{\L^2(\Sigma_t)}^2+(\gamma-\gamma_0) \norm{\phi}{\L^2(\Sigma_t)}^2 \leq C_{X, \gamma_0} E_\gamma(t)[\phi],
\ee
for all $\phi$.
\item There exists a constant $C$, independent of $\gamma$, such that for any $\epsilon>0$ the following estimate holds:
\ben{killest}
\frac{d}{dt} E_{\gamma}(t)[\phi] \leq \epsilon\left( \norm{\phi}{\H^1(\Sigma_t, \kappa)}^2 + \norm{(L+\gamma)\phi}{\L^2(\Sigma_t)}^2\right) + \frac{C}{\epsilon} \norm{T\phi}{\L^2(\Sigma_t)}^2.
\een
\item \emph{If $\phi$ is additionally assumed to vanish on the horizon}, there exists a constant $C$ independent of $\gamma$ such that for any $\epsilon>0$ the following estimate holds, 
\ben{killest2}
-\frac{d}{dt} E_{\gamma}(t)[\phi] \leq \epsilon\left( \norm{\phi}{\H^1(\Sigma_t, \kappa)}^2 + \norm{(L+\gamma)\phi}{\L^2(\Sigma_t)}^2\right) + \frac{C}{\epsilon} \norm{T\phi}{\L^2(\Sigma_t)}^2.
\een

\end{enumerate}
\end{Theorem}
\begin{proof}
\begin{enumerate}[1.]
\item The fact that tangential (but not transverse) derivatives are controlled by the Killing energy at the horizon for $\gamma$ sufficiently large can be confirmed by a calculation in local coordinates. The requirement that $\gamma$ be sufficiently large comes from the fact that we do not require $V$ to be positive, so we have to add a multiple of the $L^2$-norm of $\phi$ to ensure that the energy as defined is coercive. Given this, the only barrier to establishing $i)$ and $ii)$ is to deal with the surface term in the definition of $E_\gamma$. To do this, we note that a twisted Sobolev trace identity (see Lemma \ref{twisted trace} in the appendix) implies that if $I \in \{1, \ldots, N\}\setminus \mathcal{D}$ then for any $\delta>0$, there exists $C_\delta$ such that
\ben{traceeq}
 \int_{\scri \cap \Sigma_t} \abs{\phi_I}^2 r^{-2\kappa_I} \sqrt{A} d \mathcal{K} \leq \delta \norm{\chi \phi}{\H^1(\Sigma_t, \kappa)}^2 + C_\delta \norm{\chi \phi}{\L^2(\Sigma_t)}^2,
\een
where $\chi$ is some smooth function equal to $1$ near $\scri$ and vanishing outside a neighbourhood of $\scri$. With this in hand, the algebraic properties of the twisted energy-momentum tensor together with an analysis of the behaviour near infinity give the required inequalities.
\item To establish $iii)$ we apply Lemma \ref{litdivlem} with the vector field $K=\curJ^T_\gamma$.  We take the terms one at a time. The first term in \eq{litdiv} may be left alone as it will simply give rise to the time derivative of the first term in the definition of the energy.
\item Consider now the second term. Since $\mathcal{L}_T r = 0$ we have that $g(T, m) = 0$ and so
\bean
\int_{\scri \cap \Sigma_t} (\curJ^T_\gamma)^\mu m_\mu\ \sqrt{A} d \mathcal{K} &=& \lim_{r \to \infty} \sum_I \int_{\{r=\textrm{const.}\} \cap \Sigma_t} T^\mu m^\mu \tn_{(\mu}\overline{\phi}^I\tn_{\nu)}\phi^I \ \sqrt{A}d \mathcal{K},\\ &=& \lim_{r \to \infty} \sum_I \int_{\{r=\textrm{const.}\} \cap \Sigma_t} r l^{-1} \Re(\tn_t \overline{\phi}^I\tn_r \phi^I) \ \sqrt{A}d \mathcal{K}.
\eean
Now, if $I\in \mathcal{D}$, then the Dirichlet boundary condition implies that this limit vanishes.  If $I \in  \{1, \dots, N\}\setminus \mathcal{D}$ then we use the Robin boundary condition, together with the fact that $\tn_t = \nabla_t$ to deduce that
\be
\int_{\scri \cap \Sigma_t} (\curJ^T_\gamma)^\mu m_\mu\ \sqrt{A} d \mathcal{K} = - \frac{1}{2}\frac{d}{dt}\sum_{I \in \{1, \dots, N\}\setminus \mathcal{D}} \int_{\scri \cap \Sigma_t}  \overline{\phi}^I \phi^I \beta_I r^{-2\kappa_I} \sqrt{A} d \mathcal{K}.
\ee
Thus the two terms on the left hand side of \eq{litdiv} combine to give $dE_\gamma/dt$.
\item Now consider the surface term in \eq{litdiv} evaluated on the horizon. We have 
\be
-\int_{\hor\cap \Sigma_t} (\curJ^T_\gamma)^\mu T_\mu\ d\sigma = -\int_{\hor\cap \Sigma_t} \abs{T\phi}^2 d\sigma \leq 0.
\ee
\item Finally we consider the divergence term for which we make use of Lemma \ref{emprop}. Since $\mathcal{L}_Tr = 0$ and $\mathcal{L}_TV = 0$, we have that $T^\mu \emS_\mu = 0$. We deduce that:
\bean
\nabla_\mu  (\curJ^T_\gamma)^\mu &=& -\Re\left[ T\overline{\phi}\cdot \left(\Box_g\phi+ \frac{1}{2}(V+V^*) \phi + \frac{\gamma}{A} \phi  \right) \right] \\&=&  -\Re\left[ \frac{1}{A} T\overline{\phi}\cdot \left(L \phi+ \gamma \phi  \right) \right]+\Re\left[ \frac{1}{A} T\overline{\phi}\cdot  W \phi\right] -\Re\left[ \frac{1}{2} T\overline{\phi}\cdot (V-V^*) \phi\right],
\eean
with $W \in \mathfrak{X}_\scri(\bhR)$. Recalling that $A \sim r^2$ and $(V-V^*)\sim r^{-2}$  near $\scri$, and applying the Cauchy-Schwarz inequality, we deduce that we can find a $C$, independent of $\gamma$, such that for any $\epsilon>0$
\be
\abs{ \int_{\Sigma_t} \nabla_\mu K^\mu\ \sqrt{A} dS} \leq  \epsilon\left( \norm{\phi}{\H^1(\Sigma_t, \kappa)}^2 + \norm{(L+\gamma)\phi}{\L^2(\Sigma_t)}^2\right) + \frac{C}{\epsilon}\norm{T\phi}{\L^2(\Sigma_t)}^2.
\ee
Taking $2.$-$5.$ together we conclude that $iii)$ holds.
\item Statement $iv)$, follows in exactly the same way as $iii)$, however since $\phi$ is assumed to vanish on the horizon we do not need to consider the surface term on $\hor$ and can simply estimate the divergence term on the right.
\end{enumerate}
\end{proof}

It may seem somewhat strange that we have separated the $T\phi$ term from the rest of the derivative terms in our estimates. The reason for this will become apparent later when we consider the degenerate elliptic problem we obtain by Laplace transforming in $t$. Then $T \phi\to s  \phi$ and we can absorb this term on the left hand side by taking $\gamma$ large enough. It may also seem surprising that we consider the case where the function $\phi$ is constrained to vanish on the horizon. This is because when we come to the elliptic problem we will make use of the fact that the existence problem for a closed, densely defined, linear operator on a Hilbert space is related to the uniqueness problem for its dual operator. Loosely speaking, the dual problem to the forward evolution problem in the exterior of a horizon is the backwards evolution problem with Dirichlet conditions on the horizon. We will be more concrete about the dual problem later when we come to the elliptic problem proper, but the estimates we have derived will be of use.

We will now briefly note a corollary for the scalar Klein-Gordon equation, which follows from the proof of the main theorem using an approximation argument to weaken the smoothness assumptions.
\begin{Corollary}\label{cor1}
Suppose $\psi\in L^2(\R_+, \H^1(\Sigma, \kappa))$ with $T\psi \in L^2(R_+, \L^2(\Sigma))$ is a weak solution of the equation
\ben{scaleq}
(\Box_g + V)\psi = 0, \qquad \textrm{ in }\bhR,
\een
where $V$ is an hermitian matrix valued potential obeying the conditions of \S\ref{WPsec}, and $\psi$ is subject to compatible homogeneous stationary boundary conditions at $\scri$, and to the initial conditions:
\be
\psi|_{\Sigma_0} = \uppsi, \qquad T \psi|_{\Sigma_0} = \uppsi',
\ee
for  $\uppsi\in \H^1_\mathcal{D}(\Sigma, \kappa)$, $\uppsi' \in \L^2(\Sigma)$. Suppose further that Theorem \ref{killthm} $ii)$ holds with $\gamma_0<0$. Then in fact $E_0(t)[\psi] \in C^0(\R_+)$ with the estimate:
\be
\sup_{t\geq 0}\left( E_0(t)[\psi] \right)  \leq E_0(0)[\psi] \leq C\left(\norm{\uppsi}{\H^1(\Sigma, \kappa)}^2 + \norm{\uppsi'}{\L^2(\Sigma)}^2 \right).
\ee
\end{Corollary}
\textbf{Remark:} In the paper \cite{Holzegel:2012wt}, we gave conditions under which $\gamma_0$ may be taken to be negative in terms of the eigenvalue of a certain degenerate elliptic operator. Our proof was for a one-component field, but can be extended.

\subsubsection{The redshift estimate}

We now consider an estimate based on a stationary multiplier vector field which remains timelike at the horizon. In fact, we can choose essentially any such vector field: the slightly delicate construction of \cite{Mihalisnotes} is not actually necessary in order to gain control over the transverse derivatives. The key consideration is the following:
\begin{Lemma}\label{deformation}
Let $K\in \mathfrak{X}_\scri(\bhR)$ be stationary. Then we have
\ben{deften}
\mathcal{L}_K g = -2\varkappa K^\flat \otimes_s dt+r^2 \omega_0 \otimes_s dt + r^2 \sum_a \epsilon_a \omega_a \otimes_s \omega_a,
\een
for some finite collection of one forms $\omega_i \in \mathfrak{X}^*(\bhR)$ such that $\omega_i(T)=0$ on the horizon and $\epsilon_a = \pm1$.
\end{Lemma}
\begin{proof}
First we note that the condition that $K \in \mathfrak{X}_\scri(\bhR)$, together with the asymptotic behaviour of $g$, implies that $r^{-2} \mathcal{L}_K g$ extends smoothly to $\bhR$.  Next we note the redshift identity, which states that
\be
d\left[ g(T, T) \right] = - 2 \varkappa T^\flat,
\ee
on the horizon\footnote{$T^\flat = T_\mu dx^\mu$ is the one form obtained by `lowering an index' on $T$ using the metric $g$}. From here, we deduce that for any stationary vector field $K$, i.e. one satisfying $ \mathcal{L}_T K=-\mathcal{L}_K T =0$, we have
\be
(\mathcal{L}_K g)(T, T) = \mathcal{L}_K[ g(T, T)] -2 g(\mathcal{L}_KT, T) =-2 \varkappa g(K, T). 
\ee 
The result follows from simple linear algebra considerations at the horizon, together with a partition of unity argument away from the horizon.
\end{proof}

We may now move on to discuss the redshift estimate.
\begin{Definition}\label{RSdef}
Fix homogeneous boundary conditions for the strongly hyperbolic operator $L$, as in Definition \ref{wpdefs}, and assume that the Robin functions are stationary, i.e. $\mathcal{L}_T \beta_I = 0$. Fix also any smooth timelike stationary vector field $N$ which agrees with $n_{\Sigma_t}$ near $\hor$ and with $T$ near $\scri$. For a function $\phi\in C^\infty_{bc}(\bhR; \C^N)$, we define:
\be
(\curJ^N_\gamma)^\mu[\phi] = N^\nu \emT_\nu{}^\mu[\phi] - \frac{1}{2A}\gamma \abs{\phi}^2 N^\mu.
\ee
We also define the redshift energy on $\Sigma_t$ to be
\be
\mathcal{E}_\gamma(t)[\phi] = \int_{\Sigma_t} (\curJ^N_\gamma)^\nu[\phi] n_\nu\ dS +  \frac{1}{2}\sum_{I \in \{1, \dots, N\}\setminus \mathcal{D}} \int_{\scri \cap \Sigma_t}  \overline{\phi}^I \phi^I \beta_I r^{-2\kappa_I} \sqrt{A} d \mathcal{K}. 
\ee
Finally, recalling that $L \phi = A (\Box_g\phi+V) + W[\phi]$ for $W \in \mathfrak{X}_\scri(\bhR)$ a stationary matrix valued vector field, we define
\bean
w_L &=& \sup_{\Sigma_t\cap \hor} \left \{ \frac{1}{A} g(T, \Re[\overline{\xi}\cdot W\xi])\ |\ \xi\in \C^N, \abs{\xi} = 1 \right\}, \\
w_L^* &=& \inf_{\Sigma_t\cap \hor} \left \{ \frac{1}{A} g(T, \Re[\overline{\xi}\cdot W\xi])\ |\ \xi\in \C^N, \abs{\xi} = 1 \right\}.
\eean
Note that these are finite as a result of the compactness of spatial sections of the horizon and furthermore
\be
w_{L^\dagger} = -w_L^*, \qquad w_{L^\dagger}^* = -w_L.
\ee
\end{Definition}

We then have the following result:
\begin{Theorem}[Redshift estimate]\label{redshift}
\begin{enumerate}[i)]
\item There exist constants $C, c>0$ independent of $\gamma$ such that for any sufficiently large $\gamma$ we have 
\be
c \mathcal{E}_\gamma(t)[\phi] \leq \left( \norm{\phi}{\H^1(\Sigma_t; \kappa)}^2+\norm{T\phi}{\L^2(\Sigma_t)}^2+\gamma \norm{\phi}{\L^2(\Sigma_t)}^2\right) \leq C \mathcal{E}_\gamma(t)[\phi],
\ee
for all $\phi\in C^\infty_{bc}(\bhR; \C^N)$.
\item For sufficiently large $\gamma$, for any $\epsilon>0$ there exists $C_{\epsilon}$ independent of $\gamma$ such that
\ben{rsest}
\frac{d}{dt} \mathcal{E}_{\gamma}(t)[\phi] \leq ( w_L-\varkappa+\epsilon) \mathcal{E}_{\gamma}(t)[\phi] +C_{\epsilon} \left ( \norm{(L+\gamma)\phi}{\L^2(\Sigma_t)}^2+E_\gamma(t)[\phi] \right),
\een
for all $\phi\in C^\infty_{bc}(\bhR; \C^N)$.
\item If additionally $\phi$ vanishes on $\hor$, then for sufficiently large $\gamma$, for any $\epsilon>0$ there exists $C_{\epsilon}$ independent of $\gamma$ such that
\ben{dualrsest}
-\frac{d}{dt} \mathcal{E}_{\gamma}(t)[\phi] \leq (- w_L^*+\varkappa+\epsilon) \mathcal{E}_{\gamma}(t)[\phi] +C_{\epsilon} \left ( \norm{(L+\gamma)\phi}{\L^2(\Sigma_t)}^2+E_\gamma(t)[\phi] \right).
\een
\end{enumerate} 
In all three cases, ``sufficiently large $\gamma$'' may be taken to mean both $\gamma>2 \gamma_0$ and $\gamma\geq 0$ hold.
\end{Theorem}
\begin{proof}
The thing to bear in mind when proving this theorem is that of the type of terms we will encounter, the only terms we don't already control by the Killing energy are those involving transverse derivatives at the horizon, i.e.\ terms like $\tilde{N}\phi$. Quadratic terms involving $\tilde{N}\phi$ occur when the $\tilde{W}\phi$ term in $L$ is multiplied by $\tilde{N}\phi$ as well as from the deformation tensor term, as a result of Lemma \ref{deformation}. These give rise to the presence of $w_L$ and $-\varkappa$ respectively in the final estimate. Terms of the form $\tilde{N}\phi \tilde{K}\phi$ for $K$ tangent to the horizon can be dealt with using the Cauchy-Schwarz inequality, and give rise to the loss of $\epsilon\mathcal{E}_\gamma$. All other derivative terms may be estimated by $E_\gamma$. In more detail, we have:
\begin{enumerate}[1.]
\item Note that in \eq{traceeq}, we may assume that $\chi$ is supported in the region where $N=T$. A very similar argument to the case of the Killing energy then shows that if $\gamma>\gamma_0$, $\mathcal{E}_\gamma(t)[\phi]$ is positive. Since $N$ is everywhere timelike, $\mathcal{E}_\gamma(t)[\phi]$ controls all derivatives, including those transverse to the horizon. Thus $i)$ holds
\item We now apply \eq{litdiv} with the vector field $K=\curJ^N_\gamma$. The surface term on $\scri$ can be dealt with exactly as in the proof of Theorem \ref{killthm}, so the terms on the left hand side of \eq{litdiv} combine to give $d\mathcal{E}/dt$.
\item Note that $\gamma>2\gamma_0, \gamma \geq 0$ implies $\gamma>\gamma_0$. This implies that $-\curJ^N_\gamma$ is timelike and future directed everywhere. As a result, we conclude that
\be
-\int_{\hor \cap \Sigma_t} (\curJ^N_\gamma)^\mu T_{\mu} d \sigma \leq 0,
\ee
so the term on the horizon has a good sign.
\item We now need to estimate the divergence term. For this, we need the results of Lemma \ref{emprop}. We deduce
\bea\label{divN}
\nabla_\mu  (\curJ^N_\gamma)^\mu &=& -\Re\left[ \tilde{N}\overline{\phi}\cdot \left(\Box_g\phi+ \frac{1}{2}(V+V^*) \phi + \frac{\gamma}{A} \phi  \right) \right] + N^\mu \emS_{\nu}[\phi] +\Pi^N_{\mu \nu} \emT^{\mu \nu} \\ \nonumber &&\quad -  \frac{1}{2A}\gamma \abs{\phi}^2 (\nabla^\mu N_\mu)-\frac{\gamma}{A} (N^\mu (\nabla_\mu f)f^{-1})\overline{\phi})\cdot\phi +\frac{\gamma N(A)}{2A^2} \abs{\phi^2}.
\eea
where
\be
\Pi^N = \frac{1}{2}\mathcal{L}_N g
\ee
is the deformation tensor of $N$. We take the terms in \eq{divN} one by one.
\item We have:
\bean
\Re\left[ \tilde{N} \overline{\phi}\cdot \left(\Box_g\phi+ \frac{1}{2}(V+V^*) \phi + \frac{\gamma}{A} \phi  \right) \right]  &=& \Re\left[ \frac{1}{A} \tilde{N}\overline{\phi}\cdot \left(L \phi +\gamma \phi \right) \right]  - \Re \left[ \frac{1}{A} \tilde{N}\overline{\phi}\cdot \tilde{W}\phi \right] +\\&&\quad + \frac{1}{2}\Re[\tilde{N}\overline{\phi} \cdot (V-V^*)\phi].
\eean
We may write $W = -f \sqrt{A} N + W'$, where $W'$ is tangent to $\hor$ and $f$ is a smooth matrix valued function supported near the horizon, with $w_L^* \abs{\xi}^2\leq \Re[\xi^*\cdot f\xi] \leq w_L\abs{\xi}^2$. Recall $\gamma>\gamma_0$. As a result, $E_\gamma(t)[\phi]$ controls all derivatives tangent to the horizon, as well as the $\L^2$ norm of $\phi$ with a constant that can be chosen independent of $\gamma$ (since $\gamma$ is bounded away from $\gamma_0$). Making use of this, we deduce that for any $\delta$, there exists $C_\delta$ independent of $\gamma$ such that
\bean
&&\int_{\Sigma_t} \Re\left[ \tilde{N} \overline{\phi}\cdot \left(\Box_g\phi+ \frac{1}{2}(V+V^*) \phi+ \frac{\gamma}{A} \phi  \right) \right] \sqrt{A} dS \\&& \qquad \leq (w_L+\delta) \mathcal{E}_{\gamma}(t)[\phi] +C_{\delta} \left ( \norm{(L+\gamma)\phi}{\L^2(\Sigma_t)}^2+E_\gamma(t)[\phi] \right).
\eean

\item Recall that the correction term arising from the use of twisted derivatives, $\emS_\nu$ is given by
\be
2 \emS_\nu[\phi] = \Re\left\{ \overline{\phi} \cdot \left[ \tn^\dagger_\nu ( Ff ) f^{-1}\phi\right]+  (\tn_\sigma \overline{\phi}) \cdot \left[f^{-1} (\tn^\dagger_\nu  f)  (\tn^\sigma \phi)\right]\right \}.
\ee
Now, using the fact that  $g^{-1} = N \otimes_s K_0 + \sum_a K_a \otimes K_a$, for $K_i$ stationary and tangent to $\hor$, together with the fact that $N[f]=N[F] = 0$ near $\scri$, we deduce that for any $\delta>0$, there exists $C_\delta$ independent of $\gamma$ such that
\be
\abs{\int_{\Sigma_t} N^\mu \emS_{\nu}[\phi] \sqrt{A} dS} \leq \delta \mathcal{E}_{\gamma}(t)[\phi] +C_{\delta} E_\gamma(t)[\phi].
\ee
\item Now consider the remaining terms in \eq{divN}, those proportional to $\gamma$. These can clearly be dominated by
\be
R_1:=\gamma C_{N, f, A} \norm{\phi}{\L^2(\Sigma_t)}^2,
\ee
with $C_{N, f, A}$ independent of $\gamma$. We know that
\ben{engamma}
(\gamma-\gamma_0) \norm{\phi}{\L^2(\Sigma_t)}^2 \leq E_\gamma(t)[\phi].
\een
Making use of the fact that $\gamma < 2(\gamma-\gamma_0)$, we deduce that
\be
R_1 \leq 2 C_{N, f, A} E_\gamma(t)[\phi],
\ee
with $C_{N, f, A}$ independent of $\gamma$. 
\item Now, making use of Lemma \ref{deformation} with $K=N$, we have that for any $\delta$, there exists $C_{\delta}$ such that 
\be
-\int_{\Sigma_t} \Pi^N_{\mu \nu} \emT^{\mu \nu} \sqrt{A} dS \leq (-\varkappa +\delta)  \mathcal{E}_{\gamma}(t)[\phi] +C_{\delta} E_\gamma(t)[\phi].
\ee
Taking all the estimates in 2.-8. together, and making $\delta$ sufficiently small, we arrive at \eq{rsest}.
\item For part $iii)$, we make essentially the same estimates, but using $-N$ as a multiplier. The only place where we have to modify the argument is to check that the surface term on the horizon vanishes since $\phi$ is assumed to vanish there.
\end{enumerate}
\end{proof}

The first corollary of this theorem follows from a simple application of Gronwall's Lemma:
\begin{Corollary}\label{cor2}
\begin{enumerate}[i)]
\item Suppose $\psi\in L^2(\R_+, \H^1(\Sigma, \kappa))$ with $T\psi \in L^2(R_+, \L^2(\Sigma))$ is a weak solution of the equation
\ben{Leq}
L\psi = 0, \qquad \textrm{ in }\bhR,
\een
where $L$ is as above and $\psi$ is subject to compatible homogeneous stationary boundary conditions at $\scri$, and to the initial conditions:
\be
\psi|_{\Sigma_0} = \uppsi, \qquad T \psi|_{\Sigma_0} = \uppsi',
\ee
for  $\uppsi\in \H^1_\mathcal{D}(\Sigma, \kappa)$, $\uppsi' \in \L^2(\Sigma)$. Then in fact $\psi\in C^0(\R_+, \H^1(\Sigma, \kappa))$ with $T\psi \in C^0(R_+, \L^2(\Sigma))$, with the estimate
\ben{supest}
\sup_{t\geq 0} \left(\norm{\psi}{\H^1(\Sigma_t, \kappa)}^2 + \norm{T\psi}{\L^2(\Sigma_t)}^2 \right)\leq C e^{M t} \left(\norm{\uppsi}{\H^1(\Sigma, \kappa)}^2 + \norm{\uppsi'}{\L^2(\Sigma)}^2 \right).
\een
for some constants $C, M$ depending on $g, W, V$.

\item Suppose further that $W=0$ and that Theorem \ref{killthm} $ii)$ holds with $\gamma_0<0$. Then \eq{supest} holds with $M=0$.
\end{enumerate}
\end{Corollary}
\begin{proof}
Suppose $\uppsi$, $\uppsi'$ are in $C^\infty_{bc}(\bhR; \C^N)$ and launch a smooth solution of \eq{scaleq}. Fix a sufficiently large $\gamma$. The right hand side of \eq{rsest} is controlled by $M \mathcal{E}_\gamma(t)[\psi]$, after making use of $L\psi=0$. An application of Gronwall's Lemma immediately produces the estimate \eq{supest}. An approximation argument allows us to relax the condition that the initial data launch a smooth solution.

Suppose now $W=0$ and we may assume $\gamma_0<0$. Again consider initial data launching a smooth solution. By Corollary \ref{cor1} we have that $E_0(t)[\psi]\leq E_0(0)[\psi]$. We have $\mathcal{E}_0(t)[\psi] \sim \left(\norm{\psi}{\H^1(\Sigma_t)}^2 + \norm{T\psi}{\L^2(\Sigma_t)}^2 \right)$ and for any $\epsilon>0$ we have:
\be
\frac{d}{dt} \mathcal{E}_{0}(t)[\psi] \leq -(\varkappa-\epsilon) \mathcal{E}_{0}(t)[\psi] +C_{\epsilon} E_0(t)[\psi].
\ee
Taking $\epsilon<\varkappa$, an application of Gronwall's Lemma immediately produces the estimate \eq{supest} with $M=0$. By approximation we can again relax the assumption that the initial data are smooth.
\end{proof}
The presence of the $w_L$ term in the estimates \eq{rsest} is certainly necessary. Suppose for simplicity that a single component $\psi$ obeys
\be
(\Box_g+V)\psi = 0,
\ee
with, say,  Dirichlet conditions at infinity and for $V$ such that $\mathcal{E}_\gamma(t)[\psi]$ remains bounded by the above result. Consider now the function $\Psi = e^{k t} \psi$, $k>0$, which obeys
\be
\Box_g \Psi + 2 k (dt)^\sharp[\Psi] +(V- k^2 g^{-1}(dt, dt) )\Psi = 0.
\ee
Obviously any sharp bound for $\mathcal{E}_\gamma(t)[\Psi]$ will be $e^{2 k t}$ worse than for $\mathcal{E}_\gamma(t)[\psi]$, so clearly the presence of the $\tilde{W}$ term has an effect on the decay one can obtain, and this is reflected in our estimates.

Corollary \ref{cor2}, $ii)$ is the now-classical redshift result for an asymptotically AdS black hole. The key point is that the positivity of the surface gravity allows us to `upgrade' an estimate for the Killing energy into an estimate for the full $H^1$ norm at the horizon. There was nothing special about our choice of normal vector field to use as a multiplier -- any vector field which is future directed, stationary, timelike, and transverse to the horizon will work equally well, but it's somewhat convenient to assume it is normal to the spatial slice. In \cite{Mihalisnotes} a redshift vector field is constructed such that the deformation term controls \emph{all} derivatives on the horizon. By contrast, we note that any future directed, transverse, stationary vector field gives rise to a deformation tensor which has a good sign for the \emph{transverse derivatives} and we make use of the Killing energy to control the tangential derivatives. Our argument does not rely on the asymptotic AdS structure for these facts (in fact this rather complicates matters) so this approach can be applied to simplify any redshift argument of this form.

\subsubsection{Commuting the equation}\label{comsec} The last major result that we shall require for the full time-dependent problem is a commutation with vector fields tangent to $\scri$. We may state it as follows:
\begin{Theorem}\label{commute}
Fix stationary homogeneous boundary conditions on $\scri$. Suppose $\phi\in C^\infty_{bc}(\bhR; \C^N)$. We define $f$ by
\ben{Leqn}
L\phi = f.
\een
\begin{enumerate}[i)]
\item There exists a finite set of vector fields $K_a$, $a = 1, \ldots, M$ which span\footnote{In the sense that any element of $\mathfrak{X}_\scri(\bhR)$ may be written (not necessarily uniquely) as a linear combination of $K_a$'s with smooth coefficients.} $\mathfrak{X}_\scri(\bhR)$, such that if we define $\Phi = (\phi, K_1 \phi, \ldots, K_M \phi)$, then $\Phi$ obeys an equation
\ben{Lpeqn}
L' \Phi = f'.
\een
Here $L'$ is a strongly hyperbolic operator constructed from $L$, which acts on vectors in $\C^{N'}$, $N'=N(M+1)$, and $f'$ is defined by
\ben{fpdef}
f'_0 = f, \qquad f'_a = K_a f - \frac{K_a A}{A} f.
\een
Each component of $\Phi$ separately obeys the same boundary conditions as $\phi$, so $\Phi$ inherits stationary homogeneous boundary conditions. We also have:
\ben{neww}
w_{L'} = w_{L} - 2 \varkappa.
\een

\item Conversely, suppose that some $\Phi\in C^\infty_{bc'}(\bhR; \C^{N'})$ where $bc'$ are the inherited boundary conditions, and furthermore suppose $\Phi$ satisfies \eq{Lpeqn} where $f'$ has the form \eq{fpdef} for some $f$. Then, defining $\phi := \Phi_0$ and $\delta\Phi := (\Phi_a - K_a \phi)_{a=1,\ldots, M}$, we have that $\delta\Phi$ satisfies:
\be
L'' \delta \Phi =0,
\ee
for a strongly hyperbolic operator $L''$ acting on vectors of dimension $MN$ such that $w_{L''} = w_{L'}$. If the initial conditions imply that $\delta\Phi \equiv 0$, then $\phi$ solves \eq{Leqn}.
\end{enumerate}
\end{Theorem}
The relation \eq{neww} is very important, as we shall see later that it is this fact which permits us to successively define the quasinormal modes on a larger and larger subset of the plane of complex frequencies. This property is (a generalisation of) the enhanced redshift effect of Dafermos and Rodnianski \cite{Mihalisnotes}.

We will first prove a technical Lemma, which allows us to construct a suitable set of $K_a$'s. Theorem \ref{commute} will then follow by commuting the original equation with these vector fields.
\begin{Lemma}\label{vflem}
There exists a finite collection of vector fields $K_a$, $a= 1, \ldots, M$ with the following properties:
\begin{enumerate}[i)]
\item $K_a$ are stationary elements of $\mathfrak{X}(\bhR)$.
\item $K_1$ agrees with $N$ near $\hor$ and vanishes near $\scri$. \label{K1def}
\item $K_2$ agrees with $r^{-1} dr^\sharp$ near $\scri$  and vanishes near $\hor$. \label{K2def}
\item $K_a$ are tangent to $\hor$ and satisfy $dr(K_a)=0$ near $\scri$ for $a=3, \ldots M$.
\item If $X\in \mathfrak{X}_\scri(\bhR)$ is any stationary vector field tangent to $\scri$, then there exist (not necessarily unique) functions $x^a \in C^\infty(\bhR)$ such that
\be
X = \sum_a x^a K_a.
\ee
\item We have $r^2 \nabla \left[r^{-2} g(K_a, K_b)\right] \in \mathfrak{X}_\scri(\bhR)$
and 
\be
{}^{K_a}\Pi =  r^{-2} \sum_{b,c} f^{bc}_a K_b^\flat \otimes_s K_c^\flat,
\ee
for stationary functions $f^{bc}_a=f^{cb}_a\in C^\infty(\bhR)$ such that $\nabla f_a^{bc} \in \mathfrak{X}_\scri(\bhR)$ and
\ben{horvals}
\begin{array}{rcl}
r^{-2} f^{11}_1 &=& \frac{\varkappa}{\sqrt{A}}, \\
r^{-2} f^{11}_a &=& 0, \qquad a\neq 1,
\end{array}
\een
on $\hor$.
\item Finally, we have
\be
A {K_aV}, \frac{K_aA}{A}\in C^\infty(\bhR), \qquad [K_a, W] \in \mathfrak{X}_\scri(\bhR),
\ee
where $V$, $W$ obey the assumptions of \S\ref{WPsec}.
\end{enumerate}
\end{Lemma}
\begin{proof}
First let us pick any stationary $K_1, K_2$ as in $\ref{K1def}), \ref{K2def})$. By modifying $\Sigma$ near $\scri$ if necessary, let us assume that $\Sigma_t$ is normal to $\scri$, and in particular $g^{-1}(dt, dr) = \O{r^{-2}}$. We note that any coordinate chart $(U,\varPsi)$ on $\Sigma$ can be pushed forward to a tubular coordinate patch on $\bhR$ by the map
\be 
\begin{array}{r c l}
\R_+ \times U & \to & \bhR, \\
(t, x) & \mapsto& \varphi_t\circ\varPsi(x).
\end{array}
\ee
In such a coordinate chart the metric functions are independent of $t$. Now pick any point $p$ on $\Sigma$. There are three possibilities:
\begin{enumerate}[1.]
\item $p\in \mathring{\Sigma}$. Then we can pick a coordinate chart $(U, \varPsi)$ such that $p \in \varPsi(U)$,  $\varPsi(U) \cc \mathring{\Sigma}$. Now pick a cut-off function $\chi$ which is the identity in a neighbourhood of $p$ and is supported inside $U$. We define vector fields on the tubular coordinate patch $\mathbb{R}_+\times U$ by $X^{(p)}_\mu = \chi \partial_\mu$, where $\partial_\mu$ are the coordinate vector fields $\partial_0 = \partial_t$, $\partial_i = \partial_{x^i}$. Since the coordinate chart is stationary, these are stationary vector fields. They are also clearly elements of $\mathfrak{X}(\bhR)$ tangent to $\hor$ and $\scri$, and furthermore ${}^{X^{(p)}_\mu}\Pi$ vanishes near $\hor$ and $\scri$.

\item $p \in \Sigma \cap \scri$. In this case we may choose a tubular coordinate patch containing $p$ such that the metric is independent of $t$ and takes the form\footnote{We say a tensor is $\O{\rho}$ if the components with respect to the coordinate basis are $\O{\rho}$.}
\be
g = \frac{d\rho^2 + \mathfrak{g}_{\alpha \beta}(x^i)dx^\alpha dx^\beta+ \O{\rho^2}}{\rho^2},
\ee
where $\rho = r^{-1}$, $x^\alpha = (t, x^l)$, $l = 1, d-1$. Now note that $K_2 = -\rho \partial_\rho + \O{\rho^3}$ in this chart. We calculate
\be
{}^{K_2}\Pi = \frac{\mathfrak{g}_{\alpha \beta}(x^i)dx^\alpha dx^\beta+ \O{\rho^2}}{\rho^2}.
\ee

Choose cut-off functions $\chi_1(\rho), \chi_2(x^l)$ such that $\chi_1(\rho) = 1$ for $\rho$ close to $0$ and $\chi_1(\rho)\chi_2(x^l)=1$ in a neighbourhood of $p$. Then define $X^{(p)}_\alpha = \chi_1(\rho)\chi_2(x^l)  \partial_\alpha$. We note that $d\rho(X^{(p)}_\alpha)=0$ near $\scri$ and we calculate:
\bean
{}^{X^{(p)}_\gamma} \Pi = \frac{\mathfrak{\pi}^{(\gamma)}_{\alpha \beta}(x^i)dx^\alpha dx^\beta + \O{\rho^2}}{\rho^2} ,
\eean
near $\scri$, where $\mathfrak{\pi}^{(\gamma)}$ is the deformation tensor of $X^{(p)}_\gamma|_\scri$ with respect to the metric $\mathfrak{g}$ on $\scri$.
\item $p \in \Sigma \cap \hor$. We can pick a chart $(U, \varPsi)$ on $\Sigma$ such that $\varPsi(U)$ contains $p$, $\hor$ is given locally by $x_1=0$ and $\partial_{x^1}$ is normal to $\hor$. We pick a cut-off function $\chi$ which is the identity near $p$ and has support in $U$. In the tubular neighbourhood in $\bhR$ generated by $U$, we define $X^{(p)}_0 =  \chi \partial_{t}, X^{(p)}_1 = \chi x_1 \partial_{x_1}, X^{(p)}_l = \chi \partial_{x_l}$ for $l = 2, \ldots, d$. These vectors are all smooth and tangent to $\hor$ and $\scri$. 
\end{enumerate}
Thus for each point $p\in \Sigma$ we have constructed a tubular neighbourhood $V_{(p)}$ and a set of vector fields $\{X^{(p)}_\mu \in \mathfrak{X}(\bhR)\}$ which are stationary and tangent to $\hor$, $\scri$. Furthermore, if $K \in \mathfrak{X}(\bhR)$ is a stationary vector field tangent to $\hor$, $\scri$ supported in $V_{(p)}$ then there exist stationary $k^\mu \in C^\infty(\bhR)$ such that
\be
K = k^\mu X^{(p)}_\mu.
\ee
By compactness of $\Sigma$, we may choose a finite set $\{p_i \in \Sigma\}$ such that $V_{(p_i)}$ cover $\bhR$. Taking $\{K_a\}_{a = 3, \ldots, M}$ to be the union of $\{X^{(p_i)}_\mu\}$, we thus have a finite set of vectors $K_a$, $a=1, \ldots, M$ such that for any $X\in C^\infty(\bhR)$ tangent to $\scri$, there exist (not necessarily unique) functions $x^a \in C^\infty(\bhR)$ such that
\be
X = \sum_a x^a K_a.
\ee
The collection of vector fields so constructed agrees with $i)-v)$ above. Now, recall Lemma \ref{deformation}. Since ${}^K\Pi=\frac{1}{2} \mathcal{L}_K g $ and $dt = -\frac{1}{\sqrt{A}} N^\flat$, the conditions on the deformation tensors on the horizon \eq{horvals} follow immediately. The calculations of $2.$ above allow us to control the deformation tensors near $\scri$, so that we can verify that $vi)$ holds. Finally, $vii)$ follows from a calculation in local coordinates.
\end{proof}

We will also require the following Lemma, which permits us to commute the wave operator with a vector field. It can be proven straightforwardly with a calculation in local coordinates:

\begin{Lemma}\label{commutation lemma}
Let $K\in\mathfrak{X}(\bhR)$, with deformation tensor ${}^K \Pi$. Then for any sufficiently smooth $\phi$ we have
\be
K \Box_g\phi = \Box_gK\phi + 2 \nabla_\mu \left({}^K\Pi^{\mu\nu} \nabla_\nu \phi\right) - \nabla_\nu({}^K\Pi_\mu{}^\mu) \nabla^\nu \phi.
\ee
\end{Lemma}
We may now prove the main theorem of this section.
\begin{proof}[Proof of Theorem \ref{commute}]

The proof of part $i)$ of the theorem proceeds by commuting the equation $L\phi = f$ with the full set of vector fields $K_a$ constructed in the previous Lemma. This naturally gives rise to a system of equations for the field 
\ben{defPhi}
\Phi = (\Phi_i)_{i = 0, \ldots, M} = (\phi, K_1 \phi, \ldots, K_M \phi).
\een
This procedure enlarges the dimension of the target space for the scalar field from $N$ to $N(M+1)$, and as a consequence we have $NM$ first order differential constraints of the form $K_a \Phi_0 - \Phi_a = 0$. 

The first thing we must verify is that the enlarged system is a strongly hyperbolic system in the sense we have previously defined. This primarily requires us to verify the fall-off of various coefficients near infinity. We then wish to arrange that for the new system of equations the condition $w_{L'} = w_L - 2\varkappa$ holds. For this we make use of the freedom we have to modify the commuted system of equations near the horizon by adding multiples of the constraints. The commuted equations modified in such a way will still be satisfied for solutions  of the original equation $L\phi = f$. Using this freedom together with the behaviour of the deformation tensors of $K_a$ at the horizon we can indeed arrange that $w_{L'} = w_L - 2\varkappa$.

Finally, for part $ii)$ of the equation we must verify that the constraints are propagated by the equation, so that if they are satisfied on the initial data then they are satisfied throughout. The process by which we constructed the enlarged system essentially guarantees that this will be the case.

We now proceed to verify these steps in detail:
\begin{enumerate}[1.]
\item We start with the equation
\be
L \phi = A(\Box_g\phi + V\phi)+ W \phi = f,
\ee
and commute with the vector fields $K_a$ constructed in Lemma \ref{vflem}. Using Lemma \ref{commutation lemma} we have
\begin{align}
\nonumber K_a f - \frac{K_aA}{A} f &= A(\Box_gK_a \phi + V K_a \phi)+ W K_a \phi\\ &+ 2A \nabla_\mu \left({}^{K_a}\Pi^{\mu\nu} \nabla_\nu \phi\right)  \nonumber  \\&- A\nabla_\nu({}^{K_a} \Pi_\mu{}^\mu) \nabla^\nu \phi + \left[K_a, W \right ] \phi - \frac{K_a A}{A} W \label{comm1}  \phi\\&+A(K_a V) \phi.\nonumber 
\end{align}
\item  Consider the second line of \eq{comm1}. We have, making use of \eq{horvals}:
\begin{align*}
\nabla_\mu \left({}^{K_a}\Pi^{\mu\nu} \nabla_\nu \phi\right) &=  \nabla_\mu \left(r^{-2} f^{bc}_a (K_b)^\mu K_c \phi\right) \\ &=   K_b \left(r^{-2} f^{bc}_a\right ) K_c \phi + r^{-2} f^{bc}_a (\textrm{div }K_b) K_c \phi +  r^{-2} f^{bc}_a K_b K_c \phi .
\end{align*}
Here we understand repeated indices $b, c$ to be summed over $1, \ldots, M$. Now, since $f^{bc}_a \in C^\infty(\bhR)$ and $K_b$ is tangent to $\scri$, we deduce that $K_b \left(r^{-2} f^{bc}_a\right )  \in C^\infty(\bhR)$ and moreover, this function is $\O{r^{-2}}$ as we approach $\scri$. Thus $A K_b \left(r^{-2} f^{bc}_a\right )  \in C^\infty(\bhR)$. Now, consider the second term. Taking the trace of \eq{horvals} we have
\be
{}^{K_a}\Pi_\mu{}^\mu = \textrm{div }K_a =  r^{-2} f^{bc}_a g(K_b, K_c).
\ee
Since $r^{-2}g$ extends to $\scri$ as a smooth metric, we deduce that $\textrm{div }K_a \in C^\infty(\bhR)$ and thus $A r^{-2} f^{bc}_a (\textrm{div }K_b) \in C^\infty(\bhR)$. We also deduce from this calculation and part $vi)$ of Lemma \ref{vflem} that $A \nabla( \textrm{div }K_a) \in \mathfrak{X}_\scri(\bhR)$. Thus we can re-write \eq{comm1} as:
\begin{align}
K_a f - \frac{K_aA}{A} f &= A(\Box_gK_a \phi + V K_a \phi)+  \mathsf{W}_a{}^b K_b \phi \label{comm21}  \\& + \mathsf{W}_a\phi + A \mathsf{V}_a\phi.
\nonumber 
\end{align}
Here $\mathsf{W}_a{}^b\in\mathfrak{X}_\scri(\bhR)$ is given by
\be
\mathsf{W}_a{}^b =  \delta_a{}^b W + 2 A r^{-2} f^{cb}_a K_c,
\ee
$\mathsf{W}_a \in \mathfrak{X}_\scri(\bhR)$ is given by
\begin{align*}
\mathsf{W}_a &= 2A  \left[K_b \left(r^{-2} f^{bc}_a\right )  + r^{-2} f^{bc}_a (\textrm{div }K_b) \right] K_c \\& \qquad - A \nabla( \textrm{div }K_a) + \left[K_a, W\right] - \frac{K_aA}{A} W,
\end{align*}
and $\mathsf{V}_a \in C^\infty(\bhR)$ with $\mathsf{V}_a=\O{r^{-2}}$ is given by:
\be
\mathsf{V}_a = K_a V.
\ee
This verifies that the various coefficients have the correct behaviour near $\scri$ to ensure that the enlarged system (after replacing $K_a \phi \to \Phi_a$) will be of strongly hyperbolic form. 

\item Next we need to check that we can arrange that the enlarged system has $w_{L'} = w_L-2\varkappa$, which is a condition on the first order derivative terms at the horizon. Let us look at the term involving $\mathsf{W}_a{}^b$. We have, adding and subtracting the term $2  \varkappa \sqrt{A}  K_1 K_a\phi$:
\bean
&&  \mathsf{W}_a{}^b K_b \phi = (W+2 \iota \varkappa \sqrt{A} K_1) K_a\phi \\&& \quad +2A  r^{-2} f^{bc}_a K_b K_c \phi -2  \varkappa \sqrt{A}  K_1 K_a\phi.
\eean
Here, we recall that $\iota$ is the $N\times N$ identity matrix. Now consider the terms on the second line. We add a multiple of a commutator identity to this term to get:
\bean
A r^{-2} f^{bc}_a K_b K_c \phi - \varkappa \sqrt{A}  K_1 K_a\phi &=& A \left( r^{-2} f^{bc}_a-\frac{\varkappa}{\sqrt{A}} \delta_1^b \delta_a^c \right) K_b K_c \phi \\
&& \quad + A  \alpha^b_a (K_b K_1 \phi - K_1 K_b \phi +  \left[K_1, K_b \right] \phi) \\ &=& A\left( r^{-2} f^{bc}_a-\frac{\varkappa}{\sqrt{A}} \delta_1^b \delta_a^c + \alpha^b_a \delta_1^c - \delta^b_1 \alpha^c_a \right) K_b K_c \phi  \\
&& \quad + A  \alpha_a^b \left[K_1, K_b \right] \phi. \\
\eean
 Let us take $\alpha^b_a$ to be vanishing near $\scri$ and such that on $\hor$ they take the values $\alpha_a^1 = 0$ and $\alpha_a^b = r^{-2} f_a^{1b} - \varkappa A^{-\frac{1}{2}} \delta_a^b$ for $b \neq 1$. Then defining
 \be
 \mathsf{W}_a^{bc} = \left( r^{-2} f^{bc}_a-\frac{\varkappa}{\sqrt{A}} \delta_1^b \delta_a^c + \alpha^b_a \delta_1^c - \delta^b_1 \alpha^c_a \right),
 \ee
 we deduce that $\mathsf{W}_a^{1 b}$ vanishes on $\hor$. Next, notice that since $K_a$ span $\mathfrak{X}_\scri(\bhR)$ we may write
 \begin{align*}
\mathsf{W}_a + 2A \alpha_a^b \left [ K_1, K_b\right] \phi =  A \mathsf{V}_a{}^b K_b \phi.
\end{align*}
for some $\mathsf{V}_a{}^b\in C^\infty(\bhR)$, which are $\O{r^{-2}}$ near $\scri$.
 
 We return to \eq{comm21} and re-write it as:
 \begin{align}
K_a f - \frac{K_aA}{A} f &= A(\Box_gK_a \phi + (V\delta_a{}^b + \mathsf{V}_a{}^b) K_b \phi + \mathsf{V}_a\phi) \label{comm22} \\  & + (\delta_a{}^b (W + 2\iota \varkappa \sqrt{A} K_1)+  2 \iota  \mathsf{W}_a^{cb} K_c) K_b \phi .
\nonumber 
\end{align}
\item Introducing $\Phi_i$ as in \eq{defPhi} we may write \eq{comm22} as
\be
K_a f - \frac{K_aA}{A} f = A(\Box_g \Phi_a + V'_a{}^j \Phi_j )+ W'_a{}^j \Phi_j,
\ee
where summation over $j = 0, \ldots, M$ is implied and we have defined:
\be
W'_a{}^b = W + 2\iota \varkappa \sqrt{A} K_1+  2 \iota  \mathsf{W}_a^{cb} K_c, \qquad W'_a{}^0 = 0,
\ee
as well as
\be
V'_a{}^b = V \delta_a{}^b + \mathsf{V}_a{}^b, \qquad V_a{}^0 = \mathsf{V}_a.
\ee
This gives us good equations for $\Phi_a$. We need to supplement these with an equation for $\Phi_0 = \phi$. For this, we recall the equation for $\phi$:
\begin{align*}
f &= A(\Box_g\phi + V\phi)+ W \phi \\
&= A(\Box_g\Phi_0 + V\Phi_0)+ (W+2\iota \eta  \varkappa \sqrt{A}  K_1) \Phi_0 - 2 \varkappa \eta \sqrt{A} \Phi_1,
\end{align*}
for $\eta$ a cut-off function equal to the identity on $\hor$ and vanishing on $\scri$.  Here we have added a multiple of the constraint $K_1 \Phi_0 - \Phi_1 = 0$. We therefore take
\be
W'_0{}^i = (W+2\eta \iota \varkappa \sqrt{A} K_1)\delta_0{}^i,
\ee
and
\be
V'_0{}^i = V \delta_0{}^i -2\eta \iota \varkappa \sqrt{A}\delta_1{}^i.
\ee
\item We finally have that, defining $f'_0 = f$, $f'_a = K_a f - \frac{K_aA}{A} f$:
\be
f'_i = A(\Box_g \Phi_i + V'_i{}^j \Phi_a )+ W'_i{}^j \Phi_j =: (L' \Phi)_i, \label{comm23}  
\ee
with $W'_i{}^j$ and $V'_i{}^j$ satisfying the conditions at infinity, so that $L'$ is a strongly hyperbolic operator. Since $\phi$ is only differentiated in directions tangent to $\scri$, each $\Phi_a$ will inherit the same boundary conditions as $\phi$. Moreover we may verify that on $\hor$, as a consequence of $W_a^{1 b}$ vanishing there we have:
\be
g(T, W'_i{}^j) = (g(T, W)- 2  A \varkappa \iota) \delta_i{}^j,
\ee
so that $w_{L'}=w_L - 2\varkappa$. This completes the proof of part $i)$ of the theorem.

\item To prove part $ii)$ of the theorem, we consider solutions of the extended system of equations, and we wish to show that the constraints are propagated. Define $\delta \Phi_a = K_a \Phi_0 - \Phi_a$. We note that the $0^{th}$ component of \eq{Lpeqn} is given by
\be
L \Phi_0 + 2 \eta \varkappa \sqrt{A}\, \iota \delta \Phi_1 = f'_0.
\ee
Commuting this with $K_a$, we can repeat precisely the process of parts $1.)-4.)$ above to deduce that if the generalised system of equations holds, we have:
\begin{align*}
K_a f'_0 - \frac{K_aA}{A} f'_0 &= A(\Box_g (K_a\Phi_0) + V'_a{}^b (K_b \Phi_0) )+ W'_a{}^b K_b \Phi_0 + W'_a{}^0 \Phi_0+  A V'_a{}^0 \Phi_0 \\& +  2 K_a(\eta \varkappa \sqrt{A}\, \delta \Phi_1) - 2 \frac{K_aA}{A}(\eta \varkappa \sqrt{A}\, \delta \Phi_1).
\end{align*}
Now, suppose that the constraint holds for the source term:
\be
0 =   K_a f'_0 - \frac{K_a A}{A} f'_0 -  f'_a.
\ee
Inserting the result above, and the equations for $f'_a$, we deduce that
\begin{align*}
0 &= A(\Box_g (\delta \Phi_a) + V'_a{}^b (\delta \Phi_a) )+ W'_a{}^b \delta \Phi_b \\& +  2 K_a(\eta \varkappa \sqrt{A} \delta \Phi_1) - 2 \frac{K_aA}{A}(\eta \varkappa \sqrt{A} \delta \Phi_1).
\end{align*}
Clearly then, this can be written as
\begin{align*}
(L'' \delta\Phi)_a &:= A(\Box_g (\delta \Phi_a) + V''_a{}^b (\delta \Phi_a) )+ W''_a{}^b \delta \Phi_b = 0,
\end{align*}
where $L''$ is a strongly hyperbolic operator. We can check that
\begin{align}
g(T, W''_a{}^b) &= g(T, W_a{}^b) - 2  A \varkappa \iota \delta_a{}^1 \delta_1^b \nonumber\\ \nonumber
&= (g(T, W)- 2  A \varkappa \iota) \delta_i{}^j - 2  A \varkappa \iota \delta_a{}^1 \delta_1^b,
\end{align}
so that $w_{L''} = w_{L'}$. As a consequence the constraints are propagated. If initially $\delta\Phi = 0$ and $f'$ satisfies the constraints then $\delta\Phi$ remains zero and we can reduce the system back to the uncommuted equation.
\end{enumerate}
\end{proof}

A consequence of this theorem is that a result similar to Corollary \ref{cor2} holds, but with higher regularity Sobolev spaces replacing $\H^1$, $\L^2$. Before we can state this corollary, we must first define the higher regularity Sobolev spaces. This is a slightly technical process, because the spaces must encode the `compatibility' conditions at $\scri$. We achieve this by making use of the fact that the equation $L\psi=0$ permits us to determine all the derivatives of $\psi$ transverse to $\Sigma_t$ given only the restrictions of $\psi$ and $T\psi$ to $\Sigma_t$, provided they are sufficiently regular.
\begin{Definition}\label{Hkdef}
Suppose $L$ is a strongly hyperbolic operator on an aAdS black hole region $\bhR$, and let us fix compatible stationary homogeneous boundary conditions. Suppose $(\uppsi, \uppsi') \in H^k_{loc.}(\Sigma, \C^N) \times H^{k-1}_{loc.}(\Sigma, \C^N)$. There exists a unique $k$-jet $\Uppsi$ defined on $\Sigma_0$ determined by
\be
\left. \Uppsi \right|_{\Sigma_0} = \uppsi, \qquad \left. T \Uppsi \right|_{\Sigma_0} = \uppsi', \qquad \left. T^{i} L\Uppsi \right |_{\Sigma_0} = 0, \quad 0\leq i \leq k-2.
\ee
This satisfies $\left. T^{i} \Uppsi \right |_{\Sigma_0} \in H^{k-i}_{loc.}(\Sigma)$. Pick a set of vector fields $K_a$, $a = 1, \ldots, M$ satisfying the conditions of Lemma \ref{vflem}. We say that $(\uppsi, \uppsi') \in \mathbf{H}^k(\Sigma)$ if 
\be
\begin{array}{lcl}
\left. K_{A} \Uppsi \right |_{\Sigma_0} \in \H^1_\mathcal{D}(\Sigma, \kappa), &\quad& \abs{A}\leq k-1,\\
\left. K_{A} \Uppsi \right |_{\Sigma_0} \in \L^2(\Sigma), &\quad& \abs{A} = k,
\end{array}
\ee
where $A$ is a multi-index. We define a norm on $\mathbf{H}^k(\Sigma)$ by
\be
\norm{(\uppsi, \uppsi')}{\mathbf{H}^k(\Sigma)}^2 = \sum_{\abs{A}\leq k-1} \norm{\left. K_{A} \Uppsi \right |_{\Sigma_0}}{\H^1(\Sigma, \kappa)}^2 + \sum_{\abs{A}=k} \norm{\left. K_{A} \Uppsi \right |_{\Sigma_0}}{\L^2(\Sigma)}^2.
\ee
A different choice of $K_a$ gives rise to an equivalent norm. We note that the definitions depend on both $L$ and the choice of boundary conditions.
\end{Definition}

With this definition, we have:
\begin{Corollary}\label{cor3}
Fix stationary, homogeneous, boundary conditions for $L$, a strongly hyperbolic operator on an asymptotically AdS black hole $\bhR$. Suppose $\psi\in L^2(\R_+, \H^1(\Sigma, \kappa))$ with $T\psi \in L^2(R_+, \L^2(\Sigma))$ is a weak solution of the equation
\ben{Leq2}
L\psi = 0, \qquad \textrm{ in }\bhR,
\een
with the initial conditions:
\be
\psi|_{\Sigma_0} = \uppsi, \qquad T \psi|_{\Sigma_0} = \uppsi',
\ee
for  $(\uppsi, \uppsi') \in \mathbf{H}^k(\Sigma)$. Then in fact $(\psi, T\psi) \in C^0(\R_+, \mathbf{H}^k(\Sigma))$, with the estimate
\ben{supest2}
\sup_{t\geq 0} \norm{ \left. (\psi, T\psi)\right|_{\Sigma_t}}{\mathbf{H}^k(\Sigma)}^2 \leq C e^{M t} \norm{(\uppsi, \uppsi')}{\mathbf{H}^k(\Sigma)}^2.
\een
for some constants $C, M$ depending on $g, W, V, k$.
\end{Corollary}
\begin{proof}
This follows from commuting the equation with the vector fields $K_a$ an appropriate number of times and applying Corollary \ref{cor2} $i)$ to the resulting system.
\end{proof}
In fact, we may choose $M$ to be independent of $k$. This can be shown by repeatedly commuting the equation with $T$ and making use of elliptic estimates, see \cite[\S3.3.4]{Mihalisnotes}  for the case of a single component function obeying the wave equation outside a Schwarzschild black hole. Crucial to this result is the fact that commuting \emph{decreases} $w_L$ by $2\varkappa$. We shall not repeat this argument in detail.

\subsection{The solution operator semigroup}

Let us now define the solution operator for a strongly hyperbolic equation on an asymptotically AdS black hole background. It is essentially the operator which maps the state of the field at time $0$ to the state of the field we find at time $t$ by evolving with the equation $L\psi = 0$.
\begin{Definition}\label{Sdef}
Fix stationary, homogeneous, boundary conditions for $L$, a strongly hyperbolic operator on an asymptotically AdS black hole $\bhR$. Let $(\uppsi, \uppsi') \in \mathbf{H}^1(\Sigma)$, and let $\psi$ be the unique weak solution of
\ben{prob1}
L \psi = 0 \textrm{ in }\bhR, \qquad \left. \psi\right|_{\Sigma_0} = \uppsi, \quad  \left. T\psi\right|_{\Sigma_0} = \uppsi',
\een
subject to the boundary conditions at $\scri$. We define the solution operator $\cS(t)$ associated to the hyperbolic problem \eq{prob1} to be
\be
\begin{array}{rcrcl}
\cS(t) &:& \mathbf{H}^1(\Sigma) & \to & \mathbf{H}^1(\Sigma), \\
&& (\uppsi, \uppsi')&\mapsto& \left. (\psi, T\psi)\right|_{\Sigma_t}.
\end{array}
\ee
This definition makes sense as a result of Corollary \ref{cor2}.
\end{Definition}

The primary reason for considering the solution operator $\cS(t)$ is that it has the structure of a semigroup. See for example \cite[\S7.4]{Evans} or \cite{Hille} for the definition of a $C^0$ semigroup. We state this as:
\begin{Theorem}\label{semigroup}
The family of operators $\cS(t)$ define a $C^0$-semigroup on $\mathbf{H}^k(\Sigma)$. 
\end{Theorem}
\begin{proof}
First we note that Corollary \ref{cor3} implies that $\cS(t)$ maps $\mathbf{H}^k(\Sigma)$ (thought of as a subspace of $\mathbf{H}^1(\Sigma)$) to itself. To show that $\cS(t)$ defines a $C^0$ semigroup, we must check three conditions:
\begin{enumerate}[1.]
\item \emph{$\cS(0) = I$, the identity on $\mathbf{H}^k(\Sigma)$}. This follows immediately from the definition of $\cS(0)$ and the well-posedness theorem for problem \eq{prob1}.
\item \emph{We have $\cS(t+t') = \cS(t)\cS(t')$ for all $t, t'\geq 0$}. This is a consequence of the stationarity of $\bhR$. The problem \eq{prob1} is invariant under $t\to t+t'$, $t'\geq0$, so solving the equation on the interval $[t', t+t']$ is equivalent to solving on $[0, t]$. Noting that to solve the equation on $[0, t+t']$ we can first solve on $[0, t']$ then on $[t', t+t']$ we are done.
\item \emph{$\cS(t)$ is continuous in the strong operator topology}. This follows from Corollary \ref{cor3}, since it is the statement that  for each $x \in \mathbf{H}^k(\Sigma)$,  the map $t \mapsto \cS(t)x$ is a $C^0$  curve in $\mathbf{H}^k(\Sigma)$.
\end{enumerate}
Thus $\cS(t)$ verifies the conditions to be a $C^0$-semigroup on $\mathbf{H}^k(\Sigma)$. 
\end{proof}

A standard result concerning $C^0$-semigroups is that associated with each semigroup is a closed  operator, in general unbounded, called the \emph{infinitesimal generator} of the semigroup.
\begin{Definition}\label{Adef}
We define
\be
D^k(\A) := \left\{\bm{\psi}\in \bm{H}^k(\Sigma): \lim_{t \to 0_+} \frac{\cS(t)\bm{\psi} - \bm{\psi}}{t}\textrm{ exists in }\bm{H}^k(\Sigma) \right\},
\ee
and
\be
\A \bm{\psi} := \lim_{t \to 0_+} \frac{\cS(t)\bm{\psi} - \bm{\psi}}{t}, \qquad \textrm{ for } \bm{\psi} \in D^k(\A).
\ee
We call the unbounded operator $(D^k(\A), \A)$ the infinitesimal generator of the semigroup $\cS(t)$ on $\mathbf{H}^k(\Sigma)$. Here $D^k(\A)$ is the domain of $\A$.
\end{Definition}
We note that $(D^{k-1}(\A), \A)$ extends $(D^{k}(\A), \A)$, i.e.\ $D^k(\A)\subset D^{k-1}(\A)$ and the operators agree where they are both defined. The main property of the generators that we shall require is \cite[\S7.4: Theorem 2] {Evans}, \cite[Theorem 11.6.1]{Hille}:
\begin{Theorem}\label{gentheorem}
\begin{enumerate}[i)]
\item The domain $D^k(\A)$ is dense in $\bm{H}^k(\Sigma)$.
\item $(D^k(\A), \A)$ is a closed operator.
\item The resolvent $(\A - s)^{-1}$ exists and is a bounded linear transformation of $\bm{H}^k(\Sigma)$ onto $D^k(\A)$ for $s$ in the half-plane $\Re(s) > M$, where $M$ is the exponent in Corollary \ref{cor3}.
\end{enumerate}
\end{Theorem}
It will be useful to have an expression for $\A$ in terms of the operator $L$. To obtain this, we note that we can write:
\be
L \phi = P_2 \phi + P_1 T \phi + T T \phi,
\ee
where $P_i$ are matrix differential operators of order $i$ on $\Sigma$. $P_2$ is elliptic, degenerating on the horizon. For more details of this decomposition, see  \cite[\S4]{Holzegel:2012wt} (note that in the notation of that paper $P_2 = L$ and $P_1 =B$). Recalling that the semigroup evolves the pair $(\psi, T\psi)|_{\Sigma_t}$, we re-write the equation $L \psi = 0$ as a first order equation:
\be
T\left(\begin{array}{c}\psi \\ \psi' \end{array}\right)=\left(\begin{array}{cc}0 & 1 \\-P_2 & -P_1\end{array}\right)\left(\begin{array}{c}\psi \\ \psi' \end{array}\right).
\ee
From here we deduce
\ben{Akdef}
\A\left(\begin{array}{c}\psi \\ \psi' \end{array}\right)= \left(\begin{array}{cc}0 & 1 \\-P_2 & -P_1\end{array}\right)\left(\begin{array}{c}\psi \\ \psi' \end{array}\right).
\een
Note the operators for different $k$ differ only in their domain of definition.

\subsubsection{Definition of the quasinormal frequencies and modes}

We are now in a situation to be able to define the quasinormal modes. We note that since $(D^k(\A),\A)$ is a closed, densely defined operator, it makes sense to consider its spectrum $\sigma^k(\A)$. We know from Theorem \ref{gentheorem} that $\sigma^k(\A)\subset \{\Re(s)\leq M\}$. We would like to identify the eigenvalues of $(D^k(\A),\A)$ with the quasinormal frequencies. Unfortunately not all of the spectrum corresponds to what we would like to identify as QNFs. In particular, by considering specific examples one expects that every point in the half plane $\{ \Re(s)<c_k\}$ is an eigenvalue of $(D^k(\A),\A)$ for an appropriate $c_k$. We accordingly define the quasinormal frequencies and modes as follows:
\begin{Definition}\label{QNSdef}
Let $L$ be a strongly hyperbolic operator on a stationary black hole $\bhR$. Let $(D^k(\A),\A)$ be the infinitesimal generator of the solution semigroup on $\bm{H}^k(\Sigma)$. We say that $s\in \C$ belongs to the $\bm{H}^k$-quasinormal spectrum of $L$ if:
\begin{enumerate}[i)]
\item $\Re(s) >\frac{1}{2}  w_L + (\frac{1}{2} - k)\varkappa$,
\item $s$ belongs to the spectrum of $(D^k(\A),\A)$.
\end{enumerate}
Here $w_L$ is defined as in Definition \ref{RSdef} and $\varkappa$ is the surface gravity. If $s$ is an eigenvalue of $(D^k(\A),\A)$, we say $s$ is an $\bm{H}^k$-quasinormal frequency. The corresponding eigenfunctions are the $\bm{H}^k$-quasinormal modes.
\end{Definition}
We note that for each $k$ we restrict ourselves to eigenvalues in a right half-plane, however as $k$ increases the boundary of the half-plane moves further and further to the left. In other words the more rapidly decaying a quasinormal mode is, the higher the regularity we require in order to identify it. We can make sense of this in a heuristic fashion as follows. The horizon can be thought of as an unstable trapping surface. In the ray optics approximation, a photon which begins on the horizon and moves towards infinity will simply sit there for all time. Of course the ray optics approximation assumes an ideal `infinitely localised' wave packet. The higher the regularity we assume for our initial data, the further we are forced to depart from this idealisation. Crudely, one can estimate from dimensional considerations that a highly localised wave packet on the horizon with finite norm in $\bm{H}^k$ will decay\footnote{we refer here to decay of the wave packet itself at the horizon, imagining that we can somehow separate this from the rest of the field} roughly like $e^{-k \varkappa}$. As a result, any quasinormal mode with $\Re(s) \ll -k \varkappa$ will never be observable in the late time behaviour of the field if we consider data in $\bm{H}^k$. We need more regular data to identify the more rapidly decaying quasinormal modes.

\section{Properties of the quasinormal spectrum} \label{QNMsec}

We are now ready to state the main theorem of the paper. 
\begin{Theorem}\label{mainthm}
Let $L$ be a strongly hyperbolic operator on a stationary black hole $\bhR$, with boundary conditions fixed at infinity. Let $(D^k(\A),\A)$  be the infinitesimal generator of the solution semigroup on $\bm{H}^k(\Sigma)$. Then for $s$ in the half-plane $\Re(s) >\frac{1}{2}  w_L + (\frac{1}{2} - k)\varkappa$, either:
\begin{enumerate}[i)]
 \item $s$ belongs to the resolvent set of $(D^k(\A),\A)$,
 \\
 or:
 \item $s$ is an eigenvalue of $(D^k(\A),\A)$ with finite multiplicity.
 \end{enumerate}
Possibility $ii)$ holds only for isolated values of $s$. The resolvent is meromorphic on the half-plane $\Re(s) >\frac{1}{2}  w_L + (\frac{1}{2} - k)\varkappa$, with poles of finite rank at points satisfying $ii)$. The residues at the poles are finite rank operators.

Furthermore, if $k_1\geq k_2$, the eigenvalues and eigenfunctions of $(D^{k_1}(\A),\A)$, $(D^{k_2}(\A),\A)$ agree for $s$ in the half-plane $\Re(s) >\frac{1}{2}  w_L + (\frac{1}{2} - k_2)\varkappa$.
\end{Theorem}

In particular, this implies that the quasinormal spectrum of $L$ consists of at most a countable discrete set of quasinormal frequencies, and each quasinormal frequency corresponds to a finite number of quasinormal modes.

We give here a brief overview of the proof of Theorem \ref{mainthm}. The first step is to relate the invertibility of $(\A-s)$ to that of a simpler operator, $\hL$, which is obtained by formally Laplace transforming the operator $L$. The advantage of considering $\hL$ is that for sufficiently large $\gamma$, the hyperbolic estimates of \S\ref{hypsec} enable us to show that for $s$ in the half-plane $\Re(s) >\frac{1}{2}  w_L + (\frac{1}{2} - k)\varkappa$, we can find $\gamma$ such that $(\hL + \gamma)^{-1}$ exists and can be shown to be compact.  Crudely speaking, the compactness arises since it is a bounded operator from $L^2$ into $H^1$. The appropriate Sobolev spaces here are of course twisted Sobolev spaces, so we have to call on a twisted version of the Rellich-Kondrachov theorem proven in \cite{Holzegel:2012wt} to establish compactness. Although $(\hL + \gamma)^{-1}$ depends on $s$ in a less direct way than $(\A-s)^{-1}$, it is nevertheless holomorphic in $s$. An application of the analytic Fredholm theorem establishes a Fredholm alternative style theorem for $\hL$ from which the result follows.

\subsection{The Laplace transformed operator}\label{LTop}
When trying to construct the spectrum of  $(D^k(\A),\A)$, an important issue to confront is that $\A$ is not uniformly elliptic. In fact, it degenerates at the horizon as a result of $T$ becoming null. This is a problem in applying the usual approach to elliptic eigenvalue problems based on G\r{a}rding's inequality, since the `obvious' G\r{a}rding estimate fails to control the full $H^1$ norm at the horizon. In order to derive useful estimates for $\A$, we need to make use of the fact that $\A$ degenerates in a very particular way. In the time dependent problem studied above, we were able to get good control in spite of the degeneration by using the redshift vector field. It turns out that for the resolvent $(\A-s)^{-1}$ we can make use of precisely the same approach, however wherever the vector field $T$ acts in the time dependent problem, we instead multiply by $s$. Since $T$ is a Killing field all of the calculations may be repeated as before. 

We can in fact directly recycle the estimates above. In order to do this, we first observe that any function $u:\Sigma \to \C^N$, may be lifted uniquely to a function $\upsilon:\bhR \to \C^N$ satisfying
\be
\left. \upsilon \right|_{{\Sigma}_0} =  u, \quad T\upsilon = 0.
\ee
It's convenient to define some new functions spaces at this point. Recall that in \S\ref{comsec} we constructed a set of vector fields $K_a$. We say that $u \in \H^{1, k}_\mathcal{D}(\Sigma, \kappa)$, resp. $u \in \L^{2, k}(\Sigma)$ if
\be
\left.
\begin{array}{l}
\left. K_{A} \upsilon \right |_{\Sigma_0} \in \H^1_\mathcal{D}(\Sigma, \kappa), \\
\left. K_{A} \upsilon \right |_{\Sigma_0} \in \L^2(\Sigma), 
\end{array}
\right \} \qquad \abs{A} \leq k,
\ee
where $A$ is a multi-index and $K_a$ are as in Lemma \ref{vflem}.

We now define a second order operator, the Laplace transform of $L$. On $C^\infty_{bc}(\Sigma; \C^N)$ we define: 
\ben{lhatdef}
\hL u =  P_2 u + s P_1 u + s^2 u.
\een
We fix a family of domains for $\hat{L}_s$ by requiring $D^{k}(\hat{L}_s)$ for $k =1,2, \ldots$  to be the closure of $C^\infty_{bc}(\Sigma; \C^N)$ with respect to the graph norm $\norm{u}{\L^{2,{k-1}}(\Sigma)}^2 + ||\hL u ||_{\L^{2,{k-1}}(\Sigma)}^2$.  On each domain, $\hL$ defines a closed, densely defined operator.

For any value of $t$, we may express $\hL$ in terms of $L$ as follows:
\ben{lhatrel}
\hL u =\left .e^{-st} L\left( e^{st} \upsilon \right) \right|_{\Sigma_t}.
\een
The justification for considering $\hL$ is the following Lemma:
\begin{Lemma}\label{Lslemma}
The resolvent $(\A - s)^{-1}$ exists and is a bounded linear transformation of $\bm{H}^k(\Sigma)$ onto $D^k(\A)$ if and only if $\hL^{-1}:\L^{2,{k-1}}(\Sigma) \to D^k(\hL)$ exists as a bounded operator with $D^k(\hL) \subset \H^{1, k-1}_\mathcal{D}(\Sigma, \kappa)$. Furthermore, we have
\ben{residen}
(\A - s)^{-1} = \left(\begin{array}{cc}-1 & 0 \\-s & 1\end{array}\right) \left(\begin{array}{cc}\hL^{-1} & 0 \\0 & 1\end{array}\right)\left(\begin{array}{cc}P_1 + s & 1 \\1 & 0\end{array}\right).
\een
\end{Lemma}
\begin{proof}
By direct calculation, we may show that
\be
(\A - s) = \left(\begin{array}{cc}0 & 1 \\1 & -(P_1 + s) \end{array}\right) \left(\begin{array}{cc}\hL & 0 \\0 & 1\end{array}\right)\left(\begin{array}{cc}-1 & 0 \\-s & 1\end{array}\right).
\ee
The first and last operators on the right may be easily inverted, so that formally at least \eq{residen} holds. Checking the domains of definition of the various operators, one sees that if $\hL^{-1}:\L^{2,{k-1}}(\Sigma) \to D^k(\hL)$ exists:
\be
(\A - s)^{-1} \circ (\A - s) = \textrm{id}_{D^k(\A)}, \qquad (\A - s) \circ (\A - s)^{-1} = \textrm{id}_{\bm{H}^k(\Sigma)}.
\ee
It is straightforward to check that the right hand side of \eq{residen} defines a bounded operator on $\bm{H}^k(\Sigma)$ provided $\hL^{-1}$ is a bounded operator from $\L^{2, k-1}(\Sigma)$ to $\H^{1, k-1}_\mathcal{D}(\Sigma, \kappa)$. The converse follows in a similar fashion since we can re-write \eq{residen} as an equation for $\hL^{-1}$.
\end{proof}

As a result of this Lemma, we can focus our attention on the invertibility of $\hL$. Since $\hL^{-1}$ improves the regularity, it will turn out that by considering $\hL$ rather than $(\A-s)$ we gain some compactness, which allows us to derive a result similar to the Fredholm alternative for $\A$, even though $(\A-s)^{-1}$ is not compact.

We now derive some estimates for $\hL$, first with domain $D^1(\hL)$, but later for the higher regularity domains.
\begin{Theorem}\label{Linj}
Let $L$ be a strongly hyperbolic operator on an asymptotically AdS black hole $\bhR$, and fix stationary homogeneous boundary conditions at $\scri$. Fix $\Omega \subset \{s \in \C: \Re(s) >\frac{1}{2}  w_L -\frac{1}{2}\varkappa\}$ a compact domain. Then there exists $\gamma_1$ depending on $\Omega, L$ such that $\hL + \gamma$ is injective for any $\gamma>\gamma_1$ and $D^1(\hL)\subset \H^1(\Sigma, \kappa)$ for any $s$ in $\Omega$. Furthermore, we have the estimate:
\ben{Lsest}
\norm{u}{\H^1(\Sigma, \kappa)}^2 \leq C \norm{(\hL + \gamma) u}{\L^2(\Sigma)}^2,
\een
for any $u \in D^1(\hL)$, where $C$ depends only on $\Omega, L, \gamma$.
\end{Theorem}
\begin{proof}
Obviously it suffices to prove estimate \eq{Lsest} for $u\in C^\infty_{bc}(\Sigma; \C^N)$. The rest of the conclusions then follow by the density of this set in $D^1(\hL)$. We now establish that estimate:
\begin{enumerate}[1)]
\item Recall that $\upsilon$ is the stationary lift of $u$ to a function on $\bhR$. We first apply the estimate \eq{killest} from part $ii)$ of Theorem \ref{killthm} to the function $e^{(s+c)t}\upsilon$ for some $c$ to be determined. We assume $\gamma>\gamma_0$. We calculate:
\bean
2\Re[c+s] E_\gamma(t)[e^{st}\upsilon] &\leq& \epsilon \left(\norm{e^{st} \upsilon}{\H^1(\Sigma, \kappa)}^2 + \norm{(L+\gamma)e^{st}\upsilon+ (c P_1 + 2 c+c^2)e^{st}\upsilon}{\L^2(\Sigma)}^2 \right)\\ && \quad + \frac{C\abs{c+s}^2}{\epsilon}\norm{e^{st}\upsilon}{\L^2(\Sigma)}.
\eean
Choosing $c$ such that $\Re(c+s)>0$ on $\Omega$ and expanding the second term on the right hand side, making use of the fact that $P_1$ is a bounded map from $\H^1(\Sigma, \kappa)$ to $\L^2(\Sigma)$, and noting that the term involving $P_1$ is multiplied by $\epsilon$, we deduce that given $\delta>0$, there exists a constant $C_\delta$ depending on $\Omega$ but independent of $\gamma$ such that for any $s \in \Omega$:
\ben{ellkillest}
E_\gamma(t)[e^{st}\upsilon] \leq \delta \left(\norm{e^{st} \upsilon}{\H^1(\Sigma, \kappa)}^2 + \norm{(L+\gamma)e^{st}\upsilon}{\L^2(\Sigma)}^2 \right)+ C_{\delta}\norm{e^{st}\upsilon}{\L^2(\Sigma)}.
\een

\item Now let us look at the redshift estimate \eq{rsest} from part $ii)$ of Theorem \ref{redshift}, and this time we apply it to the function $e^{s t} \upsilon$. We deduce
\be
\left(2 \Re(s) - w_L + \varkappa-\epsilon\right) \mathcal{E}_\gamma(t)[e^{s t} \upsilon] \leq C_\epsilon\left(\norm{(L+\gamma)e^{st}\upsilon}{\L^2(\Sigma)}^2+E_\gamma(t)[e^{st}\upsilon] \right).
\ee
Now we use \eq{ellkillest} to control the term $C_\epsilon E_\gamma(t)[e^{st}\upsilon]$, and we choose $\delta$ sufficiently small that we have
\be
\delta C_\epsilon \norm{e^{st} \upsilon}{\H^1(\Sigma, \kappa)}^2 \leq \epsilon \mathcal{E}_\gamma (t)[e^{st}\upsilon].
\ee
This can be done independently of $\gamma$ by  part $i)$ of Theorem \ref{redshift}. Picking $\epsilon$ sufficiently small that $\left(2 \Re(s) - w_L + \varkappa-2\epsilon\right)>0$ on $\Omega$, we deduce that
\ben{ellrsest}
\mathcal{E}_\gamma(t)[e^{s t} \upsilon] \leq C_1\left(\norm{(L+\gamma)e^{st}\upsilon}{\L^2(\Sigma)}^2+\norm{e^{st}\upsilon}{\L^2(\Sigma)}^2 \right),
\een
for $C_1$ independent of $\gamma$.
\item Now, note that $\mathcal{E}_\gamma(t)[\phi] = \mathcal{E}_{\gamma'}(t)[\phi] + (\gamma-\gamma') \norm{\phi}{\L^2(\Sigma)}^2$. Let us pick $\gamma$ such that $\gamma-C_1>\gamma'$. Then we can absorb the last term on the right of \eq{ellkillest} to obtain:
\ben{ellrsest2}
\mathcal{E}_{\gamma'}(t)[e^{s t} \upsilon] \leq C\norm{(L+\gamma)e^{st}\upsilon}{\L^2(\Sigma)}^2,
\een
for some $\gamma'>\gamma_0$. Multiplying by $e^{-2 \Re(s)t}$, both sides of the inequality become independent of time. Making use of  part $i)$ of Theorem \ref{redshift}, we can deduce that:
\be
\norm{u}{\H^1(\Sigma, \kappa)}^2\leq C e^{-2 \Re(s)t}\mathcal{E}_{\gamma'}(t)[e^{s t} \upsilon],
\ee
so that together with \eq{lhatrel}  we conclude that 
\be
\norm{u}{\H^1(\Sigma, \kappa)}^2 \leq C \norm{(\hL + \gamma) u}{\L^2(\Sigma)}^2,
\ee
holds for some $C, \gamma$ which depends on $\Omega, L$ and for any $s \in \Omega$.
\end{enumerate}
\end{proof}
\textbf{Remark:} Note that a consequence of the estimate \eq{Lsest} together with the fact that $(D^1(\hL), \hL)$ is closed is the fact that $\hL+\gamma$ has closed range in $\L^2(\Sigma)$.

\subsubsection{The adjoint of $\hL$}

The previous section establishes injectivity of $\hL+\gamma:D^1(\hL)\to \L^2(\Sigma)$, for suitable $s$. To prove that this operator is invertible we need to supplement this with a statement of surjectivity. Since $\hL+\gamma$ is a closed, densely defined operator on a Hilbert space with a closed range, surjectivity follows from injectivity of the adjoint operator. We give a brief proof of this fact here for completeness:
\begin{Lemma}\label{surjinjd}
Let $A$ be a closed, densely defined operator on a Hilbert space $H$ with closed range. Then $A$ is surjective if and only if the adjoint\footnote{for the purposes of this section, ${}^*$ will denote the Hilbert space adjoint.}  $A^*$ is injective.
\end{Lemma}
\begin{proof}
Pick $v \in \textrm{Ran}(A)^\perp$. For any $u \in D(A)$ we have
\be
(v, Au) = 0,
\ee
from whence it follows that $v \in D(A^*)$. We deduce that for any $u \in D(A)$
\be
(A^* v, u) = 0.
\ee
Since $A$ is densely defined, this implies that $A^* v = 0$. Conversely, if $A^* v = 0$ then $(v, Au) = 0$ for any $u \in D(A)$. Thus $\textrm{Ran}(A)^\perp = \textrm{Ker}(A^*)$. Now since $\textrm{Ran}(A)$ is a closed subspace of $H$, $\textrm{Ran}(A) = \overline{\textrm{Ran}(A)} = \textrm{Ran}(A)^\perp{}^\perp$, whence $\textrm{Ran}(A) = \textrm{Ker}(A^*)^\perp$ and the result follows.
\end{proof}

Now, let us define the adjoint operator. We first define the adjoint on $C^\infty_{bc*}(\Sigma; \C^N)$ which consists of functions in $C^\infty_{bc}(\Sigma; \C^N)$ vanishing near the horizon. Suppose $u\in C^\infty_{bc*}(\Sigma; \C^N)$, and again denote its stationary lift to $\bhR$ by $\upsilon$. Acting on $C^\infty_{bc*}(\Sigma; \C^N)$, we define $\hL^\dagger$ by
\ben{Lhatrel2}
\hL^\dagger u = \left. e^{ \overline{s} t} L^\dagger (e^{- \overline{s} t} \upsilon)\right |_{\Sigma_t}.
\een
We fix the domain of $\hat{L}_s^\dagger$, $D(\hat{L}_s^\dagger)$  to be the closure of $C^\infty_{bc*}(\Sigma; \C^N)$ with respect to the graph norm $\norm{u}{\L^2(\Sigma)}^2 + ||\hL^\dagger u ||_{\L^2(\Sigma)}^2$.  With this domain, $\hL^\dagger$ defines a closed, densely defined operator.

\begin{Lemma}
The operator $\hL^\dagger$ is the adjoint of $\hL$ with respect to the inner product:
\be
\ip{u'}{u}{\L^2(\Sigma)} = \int_\Sigma \overline{u'}\cdot {u} \frac{1}{\sqrt{A}} dS,
\ee
i.e.\ we have  $\hL^* = \hL^\dagger$.
\end{Lemma}
\begin{proof}
It will in fact be easier for us to establish $( \hL^\dagger)^* = \hL$, which is equivalent to the required result since $\hL$, $\hL^\dagger$ are closed operators. From the definition of the adjoint, $u \in D(( \hL^\dagger)^*)$ if there exists $f \in \L^2(\Sigma)$ such that
\ben{adjdef}
\ip{\hL^\dagger u'}{u}{\L^2(\Sigma)} = \ip{u'}{f}{\L^2(\Sigma)}, \qquad \forall \ \ u' \in C^\infty_{bc*}(\Sigma; \C^N).
\een
and for such $u$, we then define $(\hL^\dagger)^*u:=f$.

Now suppose $u \in C^\infty_{bc}(\Sigma; \C^N)$, $u' \in C^\infty_{bc*}(\Sigma; \C^N)$ and denote their lifts to $\bhR$ by $\upsilon, \upsilon'$. Recall Corollary \ref{adjrel}. Setting $\phi_1 = e^{-\overline{s} t} \upsilon'$, $\phi_2 = e^{s t} \upsilon$ we deduce:
\ben{adjeq}
\ip{\hL^\dagger u'}{ u}{\L^2(\Sigma)}=\ip{u'}{\hL u}{\L^2(\Sigma)}.
\een
Here we use the fact that $K_\mu$ is independent of time by construction, vanishes near $\hor$ since $\phi_1$ does and furthermore vanishes at $\scri$ by the boundary conditions. Since $\hL$, $\hL^\dagger$ are closed, the result holds for $u\in D^1(\hL)$, $u' \in D(\hL^\dagger)$. We immediately deduce that $\hL \subset (\hL^\dagger)^*$. To establish the result, it will suffice to show that any element of $(\hL^\dagger)^*$ can be approximated in the norm $\norm{u}{\L^2(\Sigma)}^2 + ||\hL u ||_{\L^2(\Sigma)}^2$ by elements of $C^\infty_{bc}(\Sigma; \C^N)$.

Now from  \eq{adjdef}, \eq{adjeq} we deduce that if $u \in D(( \hL^\dagger)^*)$ then $\hL u$ exists in a weak sense and belongs to $\L^2(\Sigma)$. As a result, a standard elliptic estimate implies that $u \in H^2_{loc.}(\Sigma)$. Let $\Sigma_c$ be a compact subset of $\Sigma$ which is a neighbourhood of $\scri$, has a smooth boundary, and remains a non-zero distance from $\hor$. Let $\chi:\Sigma \to [0, 1]$ be a cut-off function such that $\chi\equiv 1$ in a neighbourhood of $\scri$ and $\textrm{supp}(d \chi) \cc \Sigma_c$.  We deduce that $\hL (\chi u) := g \in \L^2(\Sigma)$. Now let us define $ \hL^\dagger|_{\Sigma_c}$ to be the closure of the operator $\hL^\dagger$ acting on functions in $C^\infty_{bc}(\Sigma_c; \C^N)$ which vanish on the inner boundary of $\Sigma_c$. Applying \eq{adjeq} with $\chi u'$ as a test function, we deduce that $\chi u \in D((\hL^\dagger|_{\Sigma_c})^*)$. Since $\hL^\dagger$ is uniformly elliptic on $\Sigma_c$, it is a straightforward matter to show that $D((\hL^\dagger|_{\Sigma_c})^*) = D(\hL^\dagger|_{\Sigma_c})$ and $(\hL^\dagger|_{\Sigma_c})^* = \hL|_{\Sigma_c}$. As a result, for any $\epsilon>0$ we can approximate $\chi u$ with an element, $u_1^\epsilon$ of $C^\infty_{bc*}(\Sigma; \C^N)$ such that
\be
\norm{\chi u - u_1^\epsilon}{\L^2(\Sigma)} + \norm{\hL(\chi u - u_1^\epsilon)}{\L^2(\Sigma)} < \epsilon.
\ee
Now consider $(1-\chi)u$. This vanishes near $\scri$, so given $\epsilon>0$, with a simple mollification we can find a $u_2^\epsilon$ which is smooth on $\Sigma$ and vanishes near infinity such that
\be
\norm{(1-\chi) u - u_2^\epsilon}{\L^2(\Sigma)} + \norm{\hL((1-\chi) u - u_2^\epsilon)}{\L^2(\Sigma)} < \epsilon.
\ee
Taking these together, the result follows.
\end{proof}

Having identified the adjoint of $\hL$, we next show that for $s$ belonging to a subset of the complex plane, $\hL^\dagger+\gamma$ is indeed injective.
\begin{Theorem}\label{Lsur}
Let $L$ be a strongly hyperbolic operator on an asymptotically AdS black hole $\bhR$, and fix stationary homogeneous boundary conditions at $\scri$. Fix $\Omega \subset \{s \in \C: \Re(s) >\frac{1}{2}  w_L +\frac{1}{2}\varkappa\}$ a compact domain. Then there exists $\gamma_1$ depending on $\Omega, L$ such that $\hL^\dagger + \gamma$ is injective for $\gamma>\gamma_1$ and $D(\hL^\dagger)\subset \H^1(\Sigma, \kappa)$ for any $s$ in $\Omega$. Furthermore, we have the estimate:
\ben{Lsest2}
\norm{u}{\H^1(\Sigma, \kappa)}^2 \leq C \norm{(\hL^\dagger + \gamma) u}{\L^2(\Sigma)}^2,
\een
for any $u \in D(\hL^\dagger)$, where $C$ depends only on $\Omega, L$.
\end{Theorem}
\begin{proof}
The proof follows in precisely the same way as for Theorem \ref{Linj}, however we use the estimates previously derived for functions on $\bhR$ vanishing on the horizon. Once again it suffices to prove estimate \eq{Lsest2} for $u\in C^\infty_{bc*}(\Sigma; \C^N)$. The rest of the conclusions then follow by the density of this set in $D(\hL^\dagger)$. We now establish that estimate:
\begin{enumerate}[1)]
\item Recall that $\upsilon$ is the stationary lift of $u$ to a function on $\bhR$. We first apply the estimate \eq{killest2} from part $ii)$ of Theorem \ref{killthm} to the function $e^{(-\overline{s}-c)t}\upsilon$ for some $c$ to be determined, and where this time we use $L^\dagger$ in place of $L$. We assume $\gamma>\gamma_0$. We calculate:
\bean
2\Re[c+s] E_\gamma(t)[e^{-\overline{s}t}\upsilon] &\leq& \epsilon \left(\norm{e^{-\overline{s}t} \upsilon}{\H^1(\Sigma, \kappa)}^2 + \norm{(L^\dagger+\gamma)e^{-\overline{s}t}\upsilon+ (c P'_1 + 2 c+c^2)e^{-\overline{s}t}\upsilon}{\L^2(\Sigma)}^2 \right)\\ && \quad + \frac{C\abs{c+\overline{s}}^2}{\epsilon}\norm{e^{-\overline{s}t}\upsilon}{\L^2(\Sigma)}.
\eean
Choosing $c$ such that $\Re(c+s)>0$ on $\Omega$ and expanding the second term on the right hand side, making use of the fact that $P'_1$ is a bounded map from $\H^1(\Sigma, \kappa)$ to $\L^2(\Sigma)$, and noting that the term involving $P_1' \upsilon$ is multiplied by $\epsilon$, we deduce that given $\delta>0$, there exists a constant $C_\delta$ depending on $\Omega$ but independent of $\gamma$ such that for any $s \in \Omega$:
\ben{ellkillest2}
E_\gamma(t)[e^{-\overline{s}t}\upsilon] \leq \delta \left(\norm{e^{-\overline{s}t} \upsilon}{\H^1(\Sigma, \kappa)}^2 + \norm{(L^\dagger+\gamma)e^{-\overline{s}t}\upsilon}{\L^2(\Sigma)}^2 \right)+ C_{\delta}\norm{e^{-\overline{s}t}\upsilon}{\L^2(\Sigma)}.
\een

\item Now let us look at the redshift estimate \eq{dualrsest} from part $ii)$ of Theorem \ref{redshift}, and this time we apply it to the function $e^{-\overline{s} t} \upsilon$, again with $L$ replaced by $L^\dagger$. We pick some $\lambda>\lambda_0$, recall that $w^*_{L^\dagger} = -w_L$, and deduce
\be
\left(2 \Re(s) - w_L - \varkappa-\epsilon\right) \mathcal{E}_\gamma(t)[e^{-\overline{s} t} \upsilon] \leq C_\epsilon\left(\norm{(L^\dagger+\gamma)e^{-\overline{s}t}\upsilon}{\L^2(\Sigma)}^2+E_\gamma(t)[e^{-\overline{s}t}\upsilon] \right).
\ee
Now we use \eq{ellkillest2} to control the term $C_\epsilon E_\gamma(t)[e^{-\overline{s}t}\upsilon]$, and we choose $\delta$ sufficiently small that we have
\be
\delta C_\epsilon \norm{e^{-\overline{s}t} \upsilon}{\H^1(\Sigma, \kappa)}^2 \leq \epsilon \mathcal{E}_\gamma (t)[e^{-\overline{s}t}\upsilon].
\ee
This can be done independently of $\gamma$ by  part $i)$ of Theorem \ref{redshift}. Picking $\epsilon$ sufficiently small that $\left(2 \Re(s) - w_L - \varkappa-2\epsilon\right)>0$ on $\Omega$, we deduce that
\ben{ellrsest3}
\mathcal{E}_\gamma(t)[e^{-\overline{s} t} \upsilon] \leq C_1\left(\norm{(L^\dagger+\gamma)e^{-\overline{s}t}\upsilon}{\L^2(\Sigma)}^2+\norm{e^{-\overline{s}t}\upsilon}{\L^2(\Sigma)}^2 \right),
\een
for $C_1$ independent of $\gamma$.
\item Now, note that $\mathcal{E}_\gamma(t)[\phi] = \mathcal{E}_{\gamma'}(t)[\phi] + (\gamma-\gamma') \norm{\phi}{\L^2(\Sigma)}^2$. Let us pick $\gamma$ such that $\gamma-C_1>\gamma'$. Then we can absorb the last term on the right of \eq{ellkillest2} to obtain:
\ben{ellrsest4}
\mathcal{E}_{\gamma'}(t)[e^{-\overline{s} t} \upsilon] \leq C\norm{(L^\dagger+\gamma)e^{-\overline{s}t}\upsilon}{\L^2(\Sigma)}^2,
\een
for some $\gamma'>\gamma_0$. Multiplying by $e^{2 \Re(s)t}$, both sides of the inequality become independent of time. Making use of  part $i)$ of Theorem \ref{redshift}, we can deduce that:
\be
\norm{u}{\H^1(\Sigma, \kappa)}^2\leq C e^{2 \Re(s)t}\mathcal{E}_{\gamma'}(t)[e^{-\overline{s} t} \upsilon],
\ee
so that together with \eq{Lhatrel2}  we conclude that 
\be
\norm{u}{\H^1(\Sigma, \kappa)}^2 \leq C \norm{(\hL^\dagger + \gamma) u}{\L^2(\Sigma)}^2,
\ee
holds for some $C, \gamma$ which depends on $\Omega, L$ and for any $s \in \Omega$.
\end{enumerate}
\end{proof}

\textbf{Remarks:} Note that we don't claim that $\hL^\dagger+\gamma$ for sufficiently large $\gamma$ fails to be injective outside the half-plane $\Re(s) > \frac{1}{2} w_L + \frac{1}{2}\varkappa$. In fact, we shall later show that injectivity holds at least for $\Re(s) > \frac{1}{2} w_L - \frac{1}{2}\varkappa$, which is the range in which we previously established that $\hL$ is injective. For $\Re(s) < \frac{1}{2} w_L + \frac{1}{2}\varkappa$ it is no longer the case that $D(\hL^\dagger)\subset \H^1(\Sigma, \kappa)$, so it becomes more difficult to establish injectivity with energy estimates.

\begin{Theorem}\label{l2inv}
Let $L$ be a strongly hyperbolic operator on an asymptotically AdS black hole $\bhR$, and fix stationary homogeneous boundary conditions at $\scri$. Fix $\Omega \subset \{s \in \C: \Re(s) >\frac{1}{2}  w_L +\frac{1}{2}\varkappa\}$ a compact domain. Then there exists $\gamma_1$ depending on $\Omega, L$ such that $\hL+ \gamma:D^1(\hL)\to \L^2(\Sigma)$ is invertible for $\gamma>\gamma_1$ and $(\hL+ \gamma)^{-1}$ maps $\L^2(\Sigma)$ into $\H^1(\Sigma, \kappa)$. Furthermore, we have the estimate:
\ben{Lsest3}
\norm{(\hL + \gamma)^{-1}f}{\H^1(\Sigma, \kappa)}^2 \leq C \norm{f}{\L^2(\Sigma)}^2,
\een
for any $f \in \L^2(\Sigma)$, where $C$ depends only on $\Omega, L, \gamma$.
\end{Theorem}
\begin{proof}
By Theorem \ref{Lsur} and Lemma \ref{surjinjd}, we know that for $\gamma$ large enough, $(\hL + \gamma): D^1(\hL) \to \L^2(\Sigma)$ is surjective. By Theorem \ref{Linj}, $(\hL+ \gamma)$ is injective, after increasing $\gamma$ if necessary. Thus $(\hL + \gamma)^{-1}$ exists. The estimate \eq{Lsest3} is simply \eq{Lsest} from Theorem \ref{Linj}.
\end{proof}

\subsubsection{Commuting the operator}
Recall that in \S\ref{comsec} we constructed a set of vector fields $K_a$ such that commuting the equation $Lu = f$ with $K_a$ we once recovered an equation of the form $L'u = f'$ for a strongly hyperbolic operator $L'$ which acted on a higher dimensional space of scalar functions. In this way we were able to establish that higher regularity norms are propagated by the equation $Lu = 0$. In this section we will make use of a `Laplace transformed' version of this argument to show that we can uniquely solve the equation $(\hL+\gamma)u=f$ for $s$ in a larger set, provided that we assume $f$ belongs to a more regular function space.

\begin{Theorem}\label{Lsinv}
Let $L$ be a strongly hyperbolic operator on an asymptotically AdS black hole $\bhR$, and fix stationary homogeneous boundary conditions at $\scri$. Let $k \geq 0$ be an integer and fix $\Omega \subset \{s \in \C: \Re(s) >\frac{1}{2}  w_L-\left(k+\frac{1}{2}\right)\varkappa \}$ a compact domain. Then there exists $\gamma_k$ depending on $\Omega, L, k$ such for $\gamma>\gamma_k$ the equation 
\ben{eqn56}
(\hL+ \gamma) u = f,
\een
admits a unique solution $u$ for any $f \in \L^{2, k}(\Sigma)$. Furthermore $u\in \H^{1, k}(\Sigma, \kappa)$ with the estimate:
\ben{Lsest4}
\norm{u}{\H^{1,k}(\Sigma, \kappa)}^2 \leq C \norm{f}{\L^{2,k}(\Sigma)}^2,
\een
where $C$ depends only on $\Omega, L, \gamma$.
\end{Theorem}
\begin{proof}
Let us define the Laplace transformed commutators $\hat{K}_a$ by
\be
\hat{K}_a u = \left. e^{- s t} K_a e^{s t} \upsilon \right |_{\Sigma_t}.
\ee
The right hand side here is independent of $t$. These operators map $\H^{1, k}_\mathcal{D}(\Sigma, \kappa)$ to $\H^{1, k-1}_\mathcal{D}(\Sigma, \kappa)$ and $\L^{2, k}(\Sigma)$ to $\L^{2, k-1}(\Sigma)$. We will prove the theorem by inductively commuting with $\hat{K}_a$. Unfortunately, we can initially only invert $(\hL+ \gamma)$ when $s$ belongs to $\Omega \subset \{s \in \C: \Re(s) >\frac{1}{2}  w_L-\left(k-\frac{1}{2}\right)\varkappa \}$. To improve this to the full range we make use of an approximation argument after commuting.
\begin{enumerate}[1.]
\item We first establish that the theorem holds when $\Omega \subset \{s \in \C: \Re(s) >\frac{1}{2}  w_L-\left(k-\frac{1}{2}\right)\varkappa \}$. The $k=0$ case is then simply Theorem \ref{l2inv}.
\item Fix $k > 0$,  $\Omega \subset \{s \in \C: \Re(s) >\frac{1}{2}  w_L-\left(k-\frac{1}{2}\right)\varkappa \}$ a compact subset and assume that the theorem holds for $k-1$ with $\Omega \subset \{s \in \C: \Re(s) >\frac{1}{2}  w_L-\left(k-\frac{3}{2}\right)\varkappa \}$. We assume that $f$ is smooth in the interior of $\Sigma$ for the time being.
\item Commuting \eq{eqn56} with the operators $\hat{K}_a$ and applying Theorem \ref{commute}, we deduce that if a solution to \eq{eqn56} exists, the vector $u'=(u, \hat{K}_au)$ must satisfy the equation
\ben{eqn58}
(\hat{L}'_s+ \gamma) u' = f',
\een
for a new strongly hyperbolic operator $L'$, with $w_{L'} = w_L - 2 \varkappa$. 
\item Now $\Omega \subset \{s \in \C: \Re(s) >\frac{1}{2}  w_{L'}-\left(k-\frac{3}{2}\right)\varkappa \}$, and $f' \in \L^{2, k-1}(\Sigma)$, so by the induction assumption, for sufficiently large $\gamma$, there exists a unique solution $u' = (u, u_a)$ to \eq{eqn58}. From standard elliptic estimates, we know that $u'$ will in fact be smooth in the interior of $\Sigma$. Making use of the second part of Theorem \ref{commute} we can show that $w = (u_a - \hat{K}_a u)$ solves $(\hat{L}''_s+\gamma) w = 0$ which forces $w=0$ provided $\gamma$ is sufficiently large. From this we further deduce that in fact $u$ solves \eq{eqn56}.
\item From the definitions of the norms, the estimate
\be
\norm{u'}{\H^{1,{k-1}}(\Sigma, \kappa)}^2 \leq C \norm{f'}{\L^{2,{k-1}}(\Sigma)}^2,
\ee
which follows from the induction assumption implies \eq{Lsest4}. A continuity argument allows us to relax the assumption that $f$ is smooth to an assumption that $f\in \L^{2, k}(\Sigma)$.
\item Now consider the case $k=0$ once again. If we assume that $f\in \L^{2, 1}(\Sigma)$, then by the above we know that for sufficiently large $\gamma$, \eq{eqn56} admits a solution in $\H^{1, 1}(\Sigma, \kappa)$ when $s\in \Omega \subset \{s \in \C: \Re(s) >\frac{1}{2}  w_{L}-\frac{1}{2} \varkappa \}$. Now by Theorem \ref{Linj} this solution is in fact unique in $\H^1(\Sigma, \kappa)$, and satisfies
\be
\norm{u}{\H^{1}(\Sigma, \kappa)}^2 \leq C \norm{f}{\L^{2}(\Sigma)}^2.
\ee
Since $\L^{2, 1}(\Sigma)$ is dense in $\L^{2}(\Sigma)$, we deduce that for any $f\in \L^2(\Sigma)$, \eq{eqn56} admits a unique solution for $s \in \Omega$. This proves the theorem for $k=0$. Repeating the induction argument above, but using this as our starting point, we deduce the result.  
\end{enumerate}
\end{proof}

We are now ready to prove a version of the Fredholm alternative for $\hL$. This, together with Lemma \ref{Lslemma} immediately implies Theorem \ref{mainthm}. 
\begin{Theorem}\label{lhatthm}
Let $L$ be a strongly hyperbolic operator on a asymptotically AdS black hole $(\bhR, g)$ and fix boundary conditions on $\scri$ for $L$. Then for any integer $k \geq 0$ and any $s$ with $\Re(s) > \frac{1}{2} w_L - (k+\frac{1}{2})\varkappa$ one of the following holds:\\

Either:
\begin{enumerate}[1.]
\item  $\hL^{-1}$ exists as a bounded map from $\L^{2,k}(\Sigma)$ to $\H^{1, k}(\Sigma)$, \\ or:
\item There exists a finite-dimensional family of solutions to 
\be
\hL u = 0.
\ee
\end{enumerate}

Possibility 2.\ obtains only when $s \in \Lambda^k_{QNF}$, where $\Lambda^k_{QNF}$ is a discrete set of points which accumulates only at infinity.  The solutions $u$ in fact belong to $C^\infty_{bc}(\Sigma; \C^N)$, so that $\Lambda^{k}_{QNF} \subset \Lambda^{k+1}_{QNF}$. The function $s\mapsto \hL$ is meromorphic on $\{s:\Re(s) > \frac{1}{2} w_L - (k+\frac{1}{2})\varkappa\}$, with poles at $\Lambda^k_{QNF}$. The residues at the poles are finite rank operators. If $s \in \Lambda^k_{QNF}$ then the equation
\be
\hL u = f,
\ee
has a solution if and only if $f \in \mathrm{coKer}(\hL)^\perp$ where $\mathrm{coKer}(\hL)$ is a finite dimensional subspace of $\L^{2, k}(\Sigma)$ with dimension equal to that of $\mathrm{Ker}(\hL)$.
\end{Theorem}
\begin{proof}
Fix a compact connected subset $\Omega \subset \{s:\Re(s) > \frac{1}{2} w_L - (k+\frac{1}{2})\varkappa\}$. By Theorem \ref{Lsinv}, there exists $\lambda$ such that the operator
\be
B(s) := \lambda(\hL + \lambda)^{-1},
\ee
exists as a bounded operator everywhere on $\Omega$. Furthermore $B(s)$ maps $\L^{2, k}(\Sigma)$ into $\H^{1, k}(\Sigma, \kappa)$, so by the twisted version of the Rellich-Konrachov theorem (Theorem 4.1 of \cite{Holzegel:2012wt}) $B(s)$ is a compact operator. Finally the map $s \mapsto B(s)$ may be verified to be analytic on $\Omega$ with
\be
\lim_{s' \to s} \frac{B(s')-B(s)}{s'-s} = \lambda (\hL + \lambda)^{-1} (P_1 + 2 s) (\hL + \lambda)^{-1}, 
\ee
which is a bounded operator on $\L^{2, k}(\Sigma)$. $B(s)$ thus satisfies the conditions for the analytic Fredholm theorem \cite[Theorem 7.92]{renardy}. Note that
\be
\hL u = f \quad \iff \quad [1-B(s)] u = (\hL + \lambda)^{-1} f,
\ee
so $(1-B(s))^{-1}$ exists if and only if $\hL^{-1}$ exists. Recall that by Theorem \ref{gentheorem}, we can always guarantee that $\hL$ exists for some $s \in \Omega$ after enlarging $\Omega$ if necessary. The theorem then follows immediately from the analytic Fredholm theorem applied to $B(s)$, together with an elliptic regularity argument to justify that solutions of $\hL u =0$ are in fact smooth. To deduce that the dimension of the kernel and cokernel agree, we use the fact that $(1-B(s))$ is a compact perturbation of the identity and hence is Fredholm of index 0.
\end{proof}

Before we consider black holes which are merely locally stationary, we'll note a Lemma which will be of use when we relate our method to the approach based on resonances.
\begin{Lemma}\label{comp}
Suppose that $\bm{f}\in \bm{H}^k(\Sigma)$ has support only within a set $W \cc \Sigma\setminus \hor$. Then if $\Re(s)> \frac{1}{2} w_L - (k-\frac{1}{2})\varkappa$ and $s \not \in \Lambda_{QNF}$, we have that  $(\A - s)^{-1} \bm f \in \bm{H}^{k+1}(\Sigma)$, with the estimate
\be
\norm{(\A - s)^{-1} \bm f }{\bm{H}^{k+1}(\Sigma)} \leq C_{W, s} \norm{\bm{f}}{ \bm{H}^k(\Sigma)}.
\ee
\end{Lemma}
\begin{proof}
From above, we know that $(\A - s)^{-1} \bm f \in \bm{H}^{k}(\Sigma)$. An elliptic estimate shows that if  $f\in \L^{2, k-1}(\Sigma)$ is supported away from $\hor$ then $(\hL)^{-1} f\in \H^{1, k}(\Sigma)$, which suffices to prove the result.
\end{proof} 
We note that the restriction to functions supported away from the horizon is necessary. If Lemma \ref{comp} were to hold for any $f\in \bm{H}^{k}(\Sigma)$, then $\A$ would have compact resolvent, which is certainly not the case.

\section{Locally stationary black holes}\label{locstat}

In the previous section, we established that the quasinormal spectrum of a globally stationary asymptotically anti-de Sitter black hole is well defined via the regular semigroup approach, and furthermore it consists only of isolated quasinormal frequencies associated to smooth quasinormal modes. We noted that our proof readily extends to the case of an asymptotically de Sitter spacetime. Unfortunately the condition of global stationarity excludes many interesting examples of black holes, both de Sitter and anti-de Sitter. In this section we shall relax this condition to local stationarity.

Consider the Kerr-AdS black hole (see \cite{HolSmul,  Holzegel:2012wt, Holbnd} for a detailed treatment of the Klein-Gordon equation in this spacetime). When the rotation parameter is below the Hawking-Reall bound, the Kerr-AdS black hole is globally stationary. For a sufficiently rapidly rotating black hole, however, the Killing generator of the horizon, $T$, ceases to be everywhere timelike. There is necessarily an ergoregion in the spacetime. The spacetime is, however, locally stationary in the sense that for any point there exists a Killing field which is stationary in a neighbourhood of that point. Put another way, at every point, $p$ the subspace of $T_p \bhR$ spanned by vectors tangent to a Killing field is timelike (if $p$ is not on the horizon) or null (if $p$ is on the horizon). In order to capture the stationarity of a black hole containing an ergoregion, we must necessarily incorporate in our prescription some additional symmetries. 

\begin{Definition}\label{locstatbh}
We say that $(\bhR, \hor, \scri, \Sigma, g, r, T)$ is a locally stationary, asymptotically anti-de Sitter, black hole space time with AdS radius $l$ if $i)-vii)$ of Definition \ref{adsbh} hold, together with
\begin{enumerate}[i)]
\item[viii)'] $(\bhR, g)$ admits a faithful action by isometries of a compact Lie group $G$, the axial symmetry group, which commutes with $\varphi_t$ and preserves $\Sigma$. Furthermore if $\{\Phi_a\in \mathfrak{X}_\scri(\bhR) \}$ is a finite collection of Killing fields which generates the $G$ action, then $\mathrm{span} \{T, \Phi_a\}$ is timelike at all points in $\bhR\setminus \hor$ and null for points on $\hor$.
\end{enumerate}
The definition of a strongly hyperbolic operator remains unchanged in this case.
\end{Definition}

With this definition, we can essentially prove a similar result regarding the existence of a semigroup as for the globally stationary black hole, i.e.\ a strongly hyperbolic equation on this background generates a $C^0-$semigroup acting on $\mathbf{H}^k(\Sigma)$. In order to define the quasinormal spectrum, we note that the semigroup acts separately on each irreducible representation of $G$. The QNF corresponding to any finite sum of representations are isolated and have finite multiplicity, as in the globally stationary case. We would ideally like to show that the set of all QNF consists of isolated points (i.e.\ QNF corresponding to different irreducible representations of $G$ cannot accumulate) however, we are currently unable to show this.

Our results should be contrasted to the results of Vasy \cite{vasy10} which show discreteness of the QNF spectrum in the rotating case, under mild non-trapping assumptions, \emph{without} the assumption of symmetry and without restricting to a finite number of angular modes.

It is perhaps rather surprising that we are able to prove any such result since the generator $\A$ of the semigroup now fails badly to be elliptic at the boundary of the ergoregion. We are able to circumvent this by restricting our attention to functions which belong to a finite sum of irreducible representations of $G$. This restriction recovers enough ellipticity for the results to hold.

In order to avoid the issues of representation theory that arise in considering a general compact group $G$, we shall instead prove the result under the assumption that $G = U(1)$, which is the situation that usually occurs in practice (see, for example \cite{Alexakis:2010fk} for a proof that this occurs in the neighbourhood of a bifurcate horizon for stationary vacuum Einstein manifolds). There appears to be no obstacle to extending our proof to a general compact group.

We note that our definition is sufficiently general to include the most interesting case, that of the Kerr-AdS black hole:

\begin{Lemma}
The Kerr-AdS black hole (see \cite[\S 5]{Holzegel:2012wt} for a definition) with $\abs{a/l}<1$ is a locally stationary, asymptotically AdS black hole with $G = U(1)$. For $\abs{a} l < r_+^2$ it is in fact globally stationary.
\end{Lemma}

\subsection{Generalising the globally stationary Quasinormal Spectrum}

Rather than repeat in detail all of the arguments of \S\ref{hypsec}, \ref{QNMsec}, we shall instead briefly sketch how the proofs from the globally stationary case need to be modified in light of the fact that $T$ is no longer assumed to be everywhere timelike. The only place where the timelike nature of $T$ is made use of in \S\ref{hypsec} is in the Killing estimates of Theorem \ref{killthm}. Recall that $\bhR$ admits a $U(1)$ action, generated by $\Phi$ such that $\mathrm{span}\{T, \Phi\}$ is timelike. An immediate consequence of this together with the compactness of $\Sigma$ is:
\begin{Lemma}\label{modkil}
There exists a function $\chi \in C^\infty(\bhR)$, such that $T\chi = \Phi \chi = 0$ everywhere, $\chi$ vanishes near $\hor$, and finally
\be
\mathcal{T} = T + \chi \Phi,
\ee
belongs to $\mathfrak{X}_\scri(\bhR)$ and is timelike on $\bhR \setminus \hor$. The deformation tensor of $\mathcal{T}$ is given by
\be
{}^\mathcal{T}\Pi = d\chi \otimes_s \Phi^\flat,
\ee
and satisfies ${}^\mathcal{T}\Pi_\mu {}^\mu = 0$.
\end{Lemma}
Once we have made a choice of $\mathcal{T}$, we can define 
\begin{Definition}\label{Kildef2}
Fix homogeneous boundary conditions for the strongly hyperbolic operator $L$, as in Definition \ref{wpdefs}, and assume that the Robin functions are stationary and axismmetric, i.e. $\mathcal{L}_T \beta_I = \mathcal{L}_\Phi\beta_I = 0$. For a function $\phi\in C^\infty_{bc}(\bhR; \C^N)$, we define:
\be
(\curJ^\mathcal{T}_\gamma)^\mu[\phi] = \mathcal{T}^\nu \emT_\nu{}^\mu[\phi] - \frac{1}{2A} \gamma \abs{\phi}^2 \mathcal{T}^\mu.
\ee
We also define the Killing energy on $\Sigma_t$ to be
\be
E_\gamma(t)[\phi] = \int_{\Sigma_t} (\curJ^\mathcal{T}_\gamma)^\mu[\phi] n_\mu\ dS + \frac{1}{2}\sum_{I \in \{1, \dots, N\}\setminus \mathcal{D}} \int_{\scri \cap \Sigma_t}  \overline{\phi}^I \phi^I \beta_I r^{-2\kappa_I} \sqrt{A} d \mathcal{K}.
\ee
\end{Definition}
With this definition, we can prove a modified version of Theorem \ref{killthm}
\begin{Theorem}[The Killing estimate]\label{killthm2}
For $\phi\in C^\infty_{bc}(\bhR; \C^N)$, we have:
\begin{enumerate}[i)]
\item Given $\gamma \geq 0$, there exists $c>0$ independent of $\gamma$ such that
\be
E_\gamma(t)[\phi] \leq  c \left( \norm{\phi}{\H^1(\Sigma_t, \kappa)}^2 + \norm{T\phi}{\L^2(\Sigma_t)}^2+ \gamma \norm{\phi}{\L^2(\Sigma_t)}^2\right),
\ee
for all $\phi$.
\item There exists $\gamma_0$ such that for any $\gamma > \gamma_0$ and for any $X \in \mathfrak{X}_\hor(\bhR)$ we can find $C_{X, \gamma_0}>0$ independent of $\gamma$ such that
\be
\norm{\tilde{X} \phi}{\L^2(\Sigma_t)}^2+(\gamma-\gamma_0) \norm{\phi}{\L^2(\Sigma_t)}^2 \leq C_{X, \gamma_0} E_\gamma(t)[\phi],
\ee
for all $\phi$.
\item There exists a constant $C$ independent of $\gamma$ such that for any $\epsilon>0$ the following estimate holds:
\ben{killest22}
\frac{d}{dt} E_{\gamma}(t)[\phi] \leq \epsilon\left( \norm{\phi}{\H^1(\Sigma_t, \kappa)}^2 + \norm{(L+\gamma)\phi}{\L^2(\Sigma_t)}^2\right) + \frac{C}{\epsilon} \left (\norm{T\phi}{\L^2(\Sigma_t)}^2+\norm{\Phi\phi}{\L^2(\Sigma_t)}^2\right ).
\een
\item \emph{If $\phi$ is additionally assumed to vanish on the horizon}, there exists a constant $C$ independent of $\gamma$ such that for any $\epsilon>0$ the following estimate holds, 
\ben{killest222}
-\frac{d}{dt} E_{\gamma}(t)[\phi] \leq \epsilon\left( \norm{\phi}{\H^1(\Sigma_t, \kappa)}^2 + \norm{(L+\gamma)\phi}{\L^2(\Sigma_t)}^2\right)+ \frac{C}{\epsilon} \left (\norm{T\phi}{\L^2(\Sigma_t)}^2+\norm{\Phi\phi}{\L^2(\Sigma_t)}^2\right ).
\een
\end{enumerate}
\end{Theorem}
\begin{proof}
The first two parts are consequences of the fact that $\mathcal{T}$ is timelike, and becomes null on  the horizon. For part $iii)$, we again make use of the local in term divergence theorem applied to $\curJ^\mathcal{T}_\gamma$. The only term which cannot be dealt with precisely as in the proof of Theorem \ref{killthm} is ${}^\mathcal{T}\Pi^{\mu \nu} \emT_{\mu \nu}$. This, by Lemma \ref{modkil}, can be controlled by 
\be
\abs{\int_{\Sigma_t} {}^\mathcal{T}\Pi^{\mu \nu} \emT_{\mu \nu}\sqrt{A} dS} \leq \epsilon \norm{\phi}{\H^1(\Sigma, \kappa)}^2 + \frac{C}{\epsilon} \norm{\Phi \phi}{\L^2(\Sigma)},
\ee
whence we are done.
\end{proof}

Now Theorem \ref{redshift} holds as stated for the locally stationary case with stationary axisymmetric boundary conditions, provided we understand $N$ to be a timelike vector field which agrees with $n_{\Sigma_t}$ near $\hor$ and $\mathcal{T}$ near $\scri$. Thus Corollary \ref{cor2} part $i)$ holds. The commutation argument of Theorem \ref{commute} goes through unchanged, so in particular Corollary \ref{cor3} again holds in the locally stationary case. As a result, the definition of the solution semigroup in Definition \ref{Sdef} and Theorem \ref{semigroup} is valid.

At this stage, it will be useful to decompose the space $\mathbf{H}^k(\Sigma)$ into irreducible representations of the axial $U(1)$ symmetry group. We define linear subspaces
\be
\mathbf{H}^k_m(\Sigma) = \left\{ \bm{\psi} \in \bm{H}^k(\Sigma) : (\Phi  - i m) \bm{\psi}=0\right\}.
\ee
This is a closed subspace, provided we understand the equation to hold in a weak fashion. It is straightforward to show that
\be
\bm{H}^k(\Sigma) = \bigoplus_{m=-\infty}^\infty\mathbf{H}^k_m(\Sigma),
\ee
i.e.\ any $\bm{\psi}$ may be written
\be
\bm{\psi} = \sum_{m=-\infty}^\infty \bm{\psi}_m, \qquad \bm{\psi}_m \in \bm{H}^k_m(\Sigma),
\ee
where the sum converges in the $\bm{H}^k(\Sigma)$ norm and the projections $\bm{\psi}_m$ are determined uniquely. The spaces $\mathbf{H}^k_m(\Sigma)$ with different $m$ are orthogonal to one another. The reason to introduce this splitting is that the solution semigroup for a strongly hyperbolic operator with stationary axisymmetric boundary conditions respects the splitting:
\begin{Lemma}
If $L$ is a strongly hyperbolic operator on a locally stationary black hole $\bhR$ with stationary axisymmetric boundary conditions, then the family of solution operators $\mathcal{S}(t)$ defines a $C^0$-semigroup on $\bm{H}_m^k(\Sigma)$.
\end{Lemma}
\begin{proof}
By Theorem  \ref{semigroup}, $\mathcal{S}(t)$ define a $C^0-$semigroup on $\bm{H}^k(\Sigma)$. It remains to show that $\bm{H}^k_m(\Sigma)$ is an invariant subspace. Since $\Phi$ is a Killing vector, if $\bm{\psi}(t) = \mathcal{S}(t) \bm{\psi}_0$ is smooth then 
\be
(\Phi - i m)\bm{\psi}(t) = \mathcal{S}(t) (\Phi - i m) \bm{\psi}_0.
\ee
Thus if initially $ \bm{\psi}_0 \in \bm{H}^k_m(\Sigma)$, $ \bm{\psi}(t) \in \bm{H}^k_m(\Sigma)$ for all $t$. By approximation we may drop the requirement that $\bm{\psi}$ be smooth.
\end{proof}
As in definition \ref{Adef}, we define the unbounded operator $(D^k_{\leq M}(\A), A)$ to be the infinitesimal generator of the $C^0$-semigroup $\mathcal{S}(t)$ defined on 
\be
\bm{H}^k_{\leq M}(\Sigma):=  \bigoplus_{m=-M}^M\mathbf{H}^k_m(\Sigma).
\ee
We have the following result 
\begin{Theorem}\label{mainthm2}
If $L$ is a strongly hyperbolic operator on a locally stationary asymptotically AdS black hole $\bhR$, with stationary axisymmetric boundary conditions, then Theorem \ref{mainthm} holds as stated for $(D^k_{\leq M}(\A), A)$ the infinitesimal generator of the solution semigroup on $\bm{H}^k_{\leq M}(\Sigma)$.
\end{Theorem}
\begin{proof}
The proof follows in the same way as for Theorem \ref{mainthm}, however we use the Killing estimates of Theorem \ref{killthm2} in place of those of Theorem \ref{killthm}. Since for $u \in D^k_m(\A)$ we have $\Phi u = i m u$, we are able to convert the $\norm{\Phi u}{\L^2(\Sigma)}^2$ terms in our estimates into terms of the form $M^2 \norm{u}{\L^2(\Sigma)}$, which can be absorbed by taking $\gamma$ large enough in a similar fashion to how we handled terms of the form $\norm{T u}{\L^2(\Sigma)}^2$ in the proof of Theorem \ref{mainthm}.
\end{proof}

We would like to be able to show that this holds with $M=\infty$, which would correspond to being able to choose $\gamma$ uniformly in $M$. This would mean that the quasinormal modes for all $m$ cannot accumulate. Unfortunately with the assumptions we have made so far this appears not to be possible. It will not surprise those familiar with rotating black holes that the problem is in the `superradiant modes', i.e.\ those with large $m$ for fixed $s$. If these can be shown to be non-trapping there may be some hope of a result for all $M$.

\section{An example; (In)completeness of the quasinormal modes}\label{example}

In this section, we shall treat a simple $(1+1)$-dimensional asymptotically AdS black hole in considerable detail. The motivation for this is twofold. Firstly, it permits us to give a concrete application of the methods developed in the preceding sections. In particular, the relationship between regularity of data at the horizon and the quasinormal modes is somewhat more clearly shown. Secondly, our simple example provides a direct counterexample to any conjecture that the quasinormal modes are complete. In particular, we can exhibit solutions which remain arbitrarily far from the vacuum solution for an arbitrarily long time but whose projection onto the quasinormal modes vanishes.

Let us take $\bhR$ to be the manifold $\bhR = [0, 1]\times [0, \infty)$ with coordinates $(\rho, t), 0\leq\rho\leq 1, 0\leq t<\infty$. We identify the components of its boundary as $\scri = \{1\}\times [0, \infty)$, $\hor = \{0\}\times [0, \infty)$, $\Sigma = [0, 1]\times \{0\}$. We endow $\bhR$ with the metric
\ben{exmet}
g = \frac{1}{(1-\rho)^2} \left[ - (1-(1-\rho)^2)dt^2 + 2 (1-\rho)^2 dt d\rho + (1+(1-\rho)^2)d\rho^2 \right].
\een
We note that $T = \partial_t$ is a Killing vector, and we introduce $r = (1-\rho)^{-1}$. With these definitions, it is straightforward to show that
\begin{Lemma}
$(\bhR, \hor, \scri, \Sigma, g, r, T)$ is a globally stationary, asymptotically anti-de Sitter, black hole space time with AdS radius $l=1$ and surface gravity $\varkappa=1$. The black hole is foliated by spacelike surfaces $\Sigma_t  = [0, 1]\times \{t\}$.
\end{Lemma}
This metric in fact has constant curvature, so may be identified with a region of the global anti-de Sitter spacetime in $(1+1)$-dimensions. Most of what we shall say does not depend on the details of the metric, however this is a convenient choice since the null curves of this metric are straightforward to find. In fact, if we define
\be
u = t + \rho - 1, \qquad v = t + 2 \tanh^{-1}(1-\rho) + \rho-1,
\ee
then the metric takes the form
\ben{conf}
g = -\textrm{sech}^2\left(\frac{v-u}{2} \right) du dv.
\een
Thus $\bhR$ is conformally related to a region of (compactified) Minkowsi space, with $\hor$ mapped onto future null infinity and $\scri$ mapped onto the line $x=0$.

We shall consider solutions of the massless Klein-Gordon equation with Dirichlet boundary conditions:
\ben{exKG}
\Box_g \psi = 0, \qquad \psi(\cdot, 0) = \uppsi^0,\ \  \psi_t(\cdot, 0) = \uppsi^1,\ \  \psi(1, t) = 0.
\een
For convenience we assume that $\psi$ is complex. In $(1+1)$-dimensions, the massless Klein-Gordon equation corresponds to $\kappa=\frac{1}{2}$, for which value the twisted Sobolev space $\H^1(\Sigma, \kappa) = H^1(\Sigma)$, which simplifies matters. Furthermore, in this dimension the Klein-Gordon operator is conformally invariant, so we can readily write down the solutions of \eq{exKG} by making use of \eq{conf}:
\begin{Lemma}
The solution of \eq{exKG} is given by
\be
\psi(\rho, t) = \Psi[t + \rho - 1] - \Psi[t + 2 \tanh^{-1}(1-\rho) + \rho-1],
\ee
for a continuous function $\Psi:[-1, \infty) \to \C$. For $x\geq 0$, $\Psi$ is given by
\be
\Psi(x) =- \int_{0}^{P(x)}\left[ \frac{1+(1-y)^2}{2} (\uppsi^0_\rho(y) - \uppsi^1(y)) \right] dy,
\ee
where $P(x)$ is the unique solution of
\be
x = 2 \tanh^{-1}(1-P(x)) + (P(x)-1),
\ee
with $0\leq P(x) \leq 1$. For $-1\leq x < 0$, $\Psi$ is given by
\be
\Psi(x) =  -\int_{0}^1\left[ \frac{1+(1-y)^2}{2} (\uppsi^0_\rho(y) - \uppsi^1(y)) \right] dy - \int_{1+x}^1 \left[\uppsi^1(y)+\frac{1-(1-y)^2}{2} (\uppsi^0_\rho(y) - \uppsi^1(y)) \right] dy.
\ee
\end{Lemma}
\begin{proof}
It is immediate that a general solution of $\Box_g \psi = 0$ is given by $\psi_L(u) + \psi_R(v)$. Imposing the boundary conditions enables us to relate $\psi_L$ to $\psi_R$ and reduce dependence to a single function, which can be determined from the initial conditions.
\end{proof}

We're primarily interested in the behaviour of $\psi$ for late times. The late time behaviour is controlled by  $\Psi(x)$ for $x$ large, which corresponds to $P(x)$ approaching zero. In other words, the late time behaviour is governed by $\uppsi^0_\rho - \uppsi^1$ in the neighbourhood of $\hor$. This makes sense, because this is the right moving part of the initial wave. The longest surviving piece of the initial wave is that which starts just outside the horizon and moves away from the horizon. 

We'll now analyse $\Psi(x)$ for $x$ large. It's convenient to introduce
\be
f(y) = - \frac{1+(1-y)^2}{2} (\uppsi^0_\rho(y) - \uppsi^1(y)),
\ee
so that $\Psi(x) = \int_0^{P(x)} f(y) dy$. Let us first assume that the initial data belong to $\bm{H}^1(\Sigma) = H^1(\Sigma)\times L^2(\Sigma)$. Then $f(y)\in L^2(\Sigma)$ and we have the estimate
\be
\abs{\Psi(x)} \leq C_{\uppsi, \uppsi'} \sqrt{P(x)},
\ee
which follows easily from the Cauchy-Schwartz inequality. This estimate is essentially sharp: we can find initial data in $\bm{H}^1(\Sigma) $ such that $\Psi(x) = P(x)^{\frac{1}{2}+\epsilon}$ for any $\epsilon>0$. Now, we note that there exists $c$ such that
\be
P(x) \leq c e^{-x}.
\ee
We then immediately have that there exists $C$, depending on the initial data, and $t_0$ such that:
\be
\sup_{\rho \in \Sigma}\abs{\psi(\rho, t)} \leq C e^{-\frac{1}{2}t }, \qquad \textrm{for all } \ \ t  > t_0.
\ee
Thus for initial data in $\bm{H}^1(\Sigma)$, the optimal decay we can expect is $e^{-\frac{1}{2}t}$. Furthermore, this rate of decay is clearly determined by properties of the initial data at the horizon.

Let us now assume that the initial data belongs to $\bm{H}^k(\Sigma)$. This implies that $f\in H^{k-1}(\Sigma) \subset C^{k-2, \frac{1}{2}}(\Sigma)$. As a result, we have an expansion:
\be
f(y) = f_1 + 2 f_2 y + \ldots + (k-1)  f_{k-1} y^{k-2} + F_k(y),
\ee
where 
\be
\sup_{y\in[0, 1]} F_k(y) y^{2-k-\frac{1}{2}} <\infty.
\ee
Integrating, we deduce
\be
\Psi(x) = f_1 P(x) + f_2 P(x)^2 + \ldots + f_{k-1} P(x)^{k-1} + R_k(P(x)),
\ee
where
\be
R_k(y) \leq C y^{k-\frac{1}{2}}.
\ee
Now let us note that $P(x) = \sum_{n=1}^\infty P_n e^{-n x}$, with the sum converging uniformly for $x>X$. As a result, we deduce:
\begin{Lemma}\label{exQNMexp}
Suppose $\psi$ is a solution of \eq{exKG}, where $(\uppsi^0, \uppsi^1) \in \bm{H}^k(\Sigma)$. Then there exist constants $f'_i$ depending on the initial data such that
\be
\psi(\rho, t) = \sum_{n=1}^{k-1} f'_n e^{-n t} u_n(\rho) + \psi_k(\rho, t),
\ee
where
\ben{exQNM}
u_n(\rho) = e^{-n \rho}\left[1-\left(\frac{\rho}{2-\rho}\right)^n \right],
\een
and furthermore there exists $C$, depending on the initial data, such that:
\ben{decest}
\sup_{\rho \in \Sigma} \abs{\psi_k(\rho, t)} \leq C e^{-\left (k-\frac{1}{2}\right )t}.
\een
\end{Lemma}
Again, we can show that the estimate \eq{decest} is essentially sharp by choosing $f(y) = y^{k-\frac{3}{2}+\epsilon}$ for any $\epsilon>0$. There is no reason in general that $C$ should decrease as $k$ increases, so the expansion of Lemma \ref{exQNMexp} is at best an asymptotic expansion.

We note here that our intuition from the remarks following Definition \ref{QNSdef} appears to be sound. For data in $\bm{H}^k(\Sigma)$ there is an optimal decay rate of roughly $e^{-k \varkappa t}$ that can possibly be achieved. The barrier to improved decay is trapping at the horizon. By Laplace transforming, we have the following corollary:
\begin{Corollary}
Let $\hat{\psi}(\rho, s)$ be the Laplace transform of $\psi(\rho, t)$:
\be
\hat{\psi}(\rho, s) = \int_0^\infty \psi(\rho, t)e^{-s t} dt.
\ee
Then $\hat{\psi}(\rho, s)$ has a meromorphic extension to the half-plane $\Re(s) > -(k-\frac{1}{2})$, with poles at $s = -1, -2, \ldots, -(k-1)$. The residue at the pole at $s=-n$ is
\be
\mathrm{Res}(\hat{\psi}(\cdot, s), -n)= f'_n u_n.
\ee
\end{Corollary}
This result is closely related to that of Theorem \ref{mainthm}, since for a semigroup $\mathcal{S}(t)$, the resolvent acting on $u$, i.e.\ $(\A-s)^{-1}u$ is given by the Laplace transform of $\mathcal{S}(t) u$ for $\Re(s)$ sufficiently large. We know that $(\A-s)^{-1}$ is meromorphic on $\Re(s) > -(k-\frac{1}{2})$, so by the uniqueness of the meromorphic extension, $(\A-s)^{-1}$ has poles at $s = -1, -2, \ldots, -(k-1)$.

We may directly verify that $u_n(\rho)$ are the quasinormal modes of the strongly hyperbolic operator $L = A \Box_g$, and that the quasinormal frequencies, as we've defined them, are indeed $s=-n$. We can easily find $\hL$ to be
\be
\hL u =  \frac{-1}{1+(1-\rho)^2} \left[ \frac{\partial}{\partial \rho}\left(\rho(2-\rho) \frac{\partial u}{\partial \rho} \right) + s (1-\rho)^2 \frac{\partial u}{\partial \rho}+ s\frac{\partial}{\partial \rho}\left( (1-\rho)^2 u\right)  \right] + s^2 u.
\ee
This is an ordinary differential operator on $[0, 1]$, with a regular singular point at $\rho=0$ whose indicial equation has roots $0, -s$. If we seek a solution of $\hL u = 0$ in $H^k(\Sigma)$, then regularity at the origin permits us to discard the solution behaving like $\rho^{-s}$ provided that $\Re(s)>-(k-\frac{1}{2})$, \emph{and} that $-s \not \in \mathbb{N}$. For $\Re(s)<-(k-\frac{1}{2})$ both branches belong to $H^k(\Sigma)$. At $\rho =1$, we have a Dirichlet boundary condition. Our previous results imply that there exist at most a countable number of isolated points $s$ with $\Re(s)>-(k-\frac{1}{2})$ such that an $H^k(\Sigma)$ solution to the equation $\hL u =0$ satisfying the boundary conditions exists.

In fact, a general solution of $\hL u = 0$, for $s\neq 0$ is given by a linear combination of the functions
\be
w_1(\rho) = e^{s \rho}, \qquad w_2(\rho) = e^{s \rho}\left[ 1-\left(\frac{2-\rho}{\rho} \right)^s \right].
\ee
These are always linearly independent. Clearly $w_2$ is the unique solution, up to a constant multiple, which satisfies the Dirichlet condition at $\rho = 1$. This solution belongs to $H^k(\Sigma)$ if either $\Re(s)<-(k-\frac{1}{2})$ or if $s = -1-2, \ldots, -(k-1)$. The point spectrum\footnote{The case $s=0$ can be dealt with separately, and does not yield an eigenvalue.} of $(D^k(\A), \A)$ can be seen to be the union of these two sets. In fact, a little more work shows the line $\Re(s)=-(k-\frac{1}{2})$ belongs to the continuous spectrum. With our definition the quasinormal spectrum is simply
\be
\Lambda^k_{QNF} = \{-1-2, \ldots, -(k-1) \}, \qquad \Lambda_{QNF} = -\mathbb{N}.
\ee
As advertised above, each quasinormal frequency corresponds to a smooth quasinormal mode, given by \eq{exQNM}. The spectrum is shown schematically in Figure \ref{fig4}.

\begin{figure}[t]
\centering
\input{spectrum.tex}
\caption{The spectrum of $(D^k(\A),\A)$}
\label{fig4}
\end{figure}

\subsection{Completeness of the quasinormal mode spectrum}

We have now shown how our previous results apply to a concrete example for which calculations are relatively straightforward. We will conclude this section by discussing the completeness of the quasinormal modes. By analogy with the case of modes in a finite domain, it is sometimes supposed that the quasinormal modes are in some sense complete. In other words, that an arbitrary solution of a strongly hyperbolic equation may be expanded as
\be
\psi \stackrel{?}{=} \sum_{n=1}^\infty f_n e^{s_n t} u_n,
\ee
where $ \Lambda_{QNF}= \{s_n: n=1, \ldots\}$ with corresponding quasinormal modes $u_n$, where the sum should converge in some appropriate space. From the preceding discussion, we can deduce this that this is emphatically \emph{not} the case. It is clear that by taking initial data of compact support we can construct a solution whose quasinormal expansion is identically zero, but which remains arbitrarily large for an arbitrarily long time. The best we can achieve is an asymptotic expansion for large times. By restricting the data such that $\Psi$ is real analytic on $[-1, \infty]$ and equal to its expansion around $\infty$, we could produce solutions which are equal to their QNM expansions, however real analyticity is a very strong condition and is not consistent with Lorentzian causality.

The incompleteness is also apparent from the fact that $\A$ is not self-adjoint (or even normal), so we have no right to expect its spectrum to be complete. We conjecture that this is a general feature of quasinormal modes, i.e.\ the quasinormal spectrum is \emph{never} complete.

\section{Relation to resonances}\label{ressec}

In this section we shall relate our definition of the quasinormal modes back to the standard language of  resonances. We need to first adjust our definition of asymptotically AdS black holes. In order to define the unitary resolvent, whose meromorphic extension we wish to construct, we need to consider our black hole region $\bhR$ as embedded in a larger manifold, the extended black hole region, which includes a bifurcation surface. We do not need to consider a region with the full bifurcate horizon, but in practical examples, one would expect a white hole horizon to be present.

Having defined the extended black hole region $\bhR_e$, we are then able to define the unitary resolvent and show directly from our previous results that it admits a meromorphic extension.

\subsection{The extended black hole region}
Let us now define the extended black hole region.
\begin{Definition}\label{extadsbh}
We say that $(\bhR_e, \mathscr{H}, \scri, \Sigma', \Sigma, g, r, T)$ is an extended globally stationary, asymptotically anti-de Sitter, black hole space time with AdS radius $l$ if the following holds
\begin{enumerate}[i)]
\item $\bhR_e$ is a $(d+1)$-dimensional manifold with stratified boundary $ \mathscr{H}\cup \Sigma'\cup \scri$, where $\Sigma'$ is a compact manifold whose boundary has two components: $S:= \mathscr{H} \cap \Sigma'$ and $\Sigma' \cap\scri$ which are compact, connected, manifolds. We denote by $\hor$ the surface $ \mathscr{H} \setminus S$.
\item $g$ is a smooth Lorentzian metric on $\bhR_e\setminus \scri$.
\item $\scri$ is an asymptotically AdS end of $(\mathring{\bhR}_e, g)$ of AdS radius $l$, asymptotic radial coordinate $r$ and such that
\be
\bhR_e = \mathcal{D}^+(\Sigma').
\ee 
We assume $r$ extends to a smooth positive function throughout $\bhR_e$.
\item $\Sigma'$ is everywhere spacelike with respect to $g$, whereas $ \mathscr{H}$ is null. $\Sigma$ is an everywhere spacelike surface to the future of $\Sigma'$, intersecting $\hor$ and $\scri$, such that every future directed timelike curve which does not end on $ \mathscr{H}$ or $\scri$ eventually enters $\mathcal{D}^+(\Sigma) =: \bhR$.
\item $T$ is a Killing field of $g$ which vanishes on $S$, is normal to $\hor$,  transverse to $\Sigma'\setminus S$ and timelike in a neighbourhood of $\hor$, thus $\hor$ is a Killing horizon generated by $T$, which we assume to be a non-extremal black hole horizon.
\item $T\in\mathfrak{X}_\scri(\bhR)$, as a result, w.l.o.g. we may assume $\mathcal{L}_T r=0$. 
 \item If $\varphi_t$ is the one-parameter family of diffeomorphisms generated by $T$, then $\bhR_e \setminus  \mathscr{H}$ is smoothly foliated by  $\varphi_\tau(\Sigma'):=\Sigma'_\tau$, $\tau \geq 0$. $\bhR$ is smoothly foliated by $\varphi_t(\Sigma):=\Sigma_t$, $t\geq 0$, and $t$ extends to a time coordinate on $\bhR_e$ satisfying $\left. t\right|_{\Sigma_0}=0$, $T(t)=1$.
 \item $T$ is timelike on $\bhR_e \setminus\mathscr{H}$.
\end{enumerate}
\end{Definition}
We can extend this to a definition for locally stationary black holes by replacing assumption $viii)$ with assumption $viii)'$ of Definition \ref{locstatbh}. We will refer to the slicing by $\Sigma'_\tau$ as the unitary slicing, since for this slicing the $T-$energy is conserved for solutions of the Klein-Gordon equation. This is a consequence of the fact that all the $\Sigma'_\tau$ slices have a common boundary, $S$. We refer to the slicing by $\Sigma_t$ as the regular slicing, since it is regular on $\hor$. Figures \ref{fig5}, \ref{fig6} show these two slicings, together with the region they foliate.

\subsubsection{AdS-Schwarzschild}

Consider the region of the plane
\be
\Delta = \{ (U, V) \in \R^2: U \geq 0, V \geq U, UV \leq 1 \}.
\ee
We define a manifold with stratified boundary:
\be
\mathscr{R}_e = \Delta_{U,V} \times S^2_\omega,
\ee
where $\Delta$ carries the canonical differential structure inherited from $\R^2$, and we identify the boundary components as:
\begin{align*}
\mathscr{H} &= \{(U, V) \in\Delta, U=0\} \times S^2, \\
\scri &= \{(U, V) \in\Delta, UV=1\} \times S^2, \\
\Sigma' &= \{ (U, V) \in\Delta, V=U\} \times S^2, \\
S&= \{ (0,0)\in \Delta\} \times S^2.
\end{align*}
We take $r_+$ to be the largest positive root of the function
\be
f(x) = 1 - \frac{2m}{x} + \frac{x^2}{l^2}.
\ee
We note that $\lambda:=f'(r_+)>0$ so we may define a diffeomorphism $R: (0, 1) \to (r_+ , \infty)$ by the implicit condition
\be
-\lambda \int_{R(x)}^\infty \frac{ds}{f(s)}= \log x, \qquad \textrm{for all } x \in (0,1).
\ee
Since $f(s)$ has a simple zero at $s=r_+$, and for large $s$ we have $f(x) = l^{-2} s^2 + \O{1}$, we may easily show that
\be
R(x) = r_++ x + \O{x^2} \qquad \textrm{ as }x \to 0,
\ee
and
\be
\frac{1}{R(x)} = \frac{1-x}{l^2\lambda^2}+ \O{(1-x)^2} \qquad \textrm{ as }x \to 1.
\ee
so we may take as asymptotic radial coordinate the function $r: \mathscr{R}_e \setminus \scri \to [r_+, \infty)$ defined by:
\be
r(U, V, \omega) = R(UV)  
\ee
It is straightforward to verify that $r^{-1}$ is a boundary defining function for $\scri = \{UV=1\}$. We define the metric on $\bhR\setminus \scri$ to be
\be
g = \frac{4}{\lambda^2 }\frac{f(r)}{UV} dU dV + r^2 d\Omega_2^2.
\ee
We see that since $(R(x)-r_+)$ has a simple zero at $x=0$, this metric is regular at $\mathscr{H}$. $g$ admits a Killing vector, given by:
\be
T = \frac{\lambda}{2} \left(V \frac{\partial}{\partial V} - U \frac{\partial}{\partial U}\right).
\ee
A simple calculation shows that $g(T, T) = -f(r)$, so that $T$ is timelike in $\bhR_e \setminus \mathscr{H}$, tangent to $\scri$ and  null on $\hor$. Moreover $T$ vanishes on the sphere $S=\{U=V=0\}$. 

Let us define further the function
\be
t(U, V) = \frac{2}{\lambda} \log V - l \arctan \frac{R(UV)}{l}.
\ee
We leave it as an exercise for the interested reader to show that taking $(t, r, \omega)$ as coordinates on the region $\{V>0\}$, the metric has the form \eq{gdd} previously given for the Schwarzschild-AdS metric. Thus we may take $\Sigma = \{t=0\}$. Finally then we have

\begin{Lemma}
$(\bhR_e, \mathscr{H}, \scri, \Sigma', \Sigma, g, r, T)$ as defined above constitute an extended globally stationary, asymptotically anti-de Sitter black hole spacetime. The AdS-Schwarzschild black hole defined in \S\ref{schwarzschild section} may be identified with the region $\mathcal{D}^+(\Sigma)$ within this extended spacetime.
\end{Lemma}

\begin{figure}[t]
\begin{minipage}[b]{0.45\linewidth}
\centering
\input{exadsbh1.tex}
\caption{Schematic Penrose diagram showing the unitary slicing of $\bhR_e$}
\label{fig5}
\end{minipage}
\hspace{0.8cm}
\begin{minipage}[b]{0.45\linewidth}
\centering
\input{exadsbh2.tex}
\caption{Schematic Penrose diagram showing the regular slicing of $\bhR$}
\label{fig6}
\end{minipage}
\end{figure}

\subsection{The unitary resolvent}\label{unitary}

We'll consider again the initial-boundary value problem for a strongly hyperbolic operator $L$, defined as in previous sections. After fixing stationary homogeneous boundary conditions at infinity, we consider solutions of
\ben{resKG}
L \psi = 0\textrm{ in }\bhR_e, \qquad \left. \psi \right |_{\Sigma'_0} = \uppsi, \ \left.T\psi \right |_{\Sigma'_0} = \uppsi',
\een
We want to think of the evolution in $\tau$ of this system as a semigroup on a Hilbert space, as in previous sections. To do this, we note that $\uppsi \in H^1_{loc.}(\Sigma')$ and $\uppsi' \in L^2_{loc.}(\Sigma')$ determine a unique 1-jet, $\Psi$ defined on $\Sigma'_0$ such that
\be
\left. \Psi \right |_{\Sigma_0} = \uppsi, \qquad \left.T  \Psi \right |_{\Sigma_0} = \uppsi'.
\ee
We then define the $\bm{H}^1(\Sigma')$ norm to be
\be
\norm{(\uppsi, \uppsi')}{\bm{H}^1(\Sigma')}^2 = \norm{\Psi}{\H^1(\Sigma'_0, \kappa)}^2 +  \norm{\hat{n}_{\Sigma'_0} \Psi}{\L^2(\Sigma'_0)}^2,
\ee
where we recall that $\hat{n}_{\Sigma'_0}= r n_{\Sigma'_0}$ is the rescaled normal of $\Sigma'$. It's also useful to introduce the norms on the marginal spaces\footnote{We apologise for the proliferation of different function spaces. In this case they are made necessary by the fact that $T$ vanishes on $S$.} $\overline{H}^1(\Sigma', \kappa)$, $\overline{L}^2(\Sigma')$ by
\be
\norm{\psi}{\overline{H}^1(\Sigma', \kappa)} = \norm{(\psi, 0)}{\bm{H}^1(\Sigma')}, \qquad \norm{\psi}{\overline{L}^2(\Sigma')} = \norm{(0,\psi)}{\bm{H}^1(\Sigma')}.
\ee

By Theorem \ref{WPThm}, we know that a unique solution of \eq{resKG} exists for $(\uppsi, \uppsi') \in\bm{H}^1(\Sigma')$. Constructing a vector field which agrees with $n_{\Sigma'_\tau}$ near $S$ and $T$ near $\scri$ to use as a multiplier for the energy-momentum tensor, it is straightforward to show that:
\ben{reseng}
\norm{\left.(\psi,T\psi)\right|_{\Sigma'_\tau}}{\bm{H}^1(\Sigma')}^2 \leq Ce^{M\tau}\norm{(\uppsi, \uppsi')}{\bm{H}^1(\Sigma')}^2.
\een
Notice here that these are \emph{not} the spaces associated with the $T$-energy on $\Sigma'_\tau$, which is degenerate on the bifurcation surface $S$. Since we are going to be constructing a meromorphic extension by restricting to functions vanishing near the horizon, this is not especially important. 

An immediate consequence of this is the following:
\begin{Lemma}\label{unires}
Let $\overline{\mathcal{S}}(\tau) (\uppsi, \uppsi') = \left . (\psi, T\psi)\right|_{\Sigma'_\tau}$, where $\psi$ solves \eq{resKG}. Then $\overline{S}(\tau)$ is a $C^0$-semigroup on $\bm{H}^1(\Sigma')$. The infinitesimal generator of this semigroup, $(D(\overline{A}),\overline{\A})$ is a closed, unbounded, operator such that the resolvent $(\overline{\A}-s)^{-1}$ is holomorphic in the half-plane $Re(s)>M$. 
\end{Lemma}

We will refer to $\overline{\mathcal{S}}(\tau)$ as the unitary semigroup and $(\overline{\A}-s)^{-1}$ as the unitary resolvent. The reason for this is that the backward evolution problem for this slicing is well-posed, i.e.\ we can extend the definition of $\overline{\mathcal{S}}(\tau)$ to $-\infty<\tau<\infty$ such that:
\be
\overline{\mathcal{S}}(\tau_1)\overline{\mathcal{S}}(\tau_2) = \overline{\mathcal{S}}(\tau_1+\tau_2), \qquad \textrm{ for all } -\infty<\tau_i<\infty,
\ee
In other words, $\overline{\mathcal{S}}(\tau)$ can be viewed as a representation of the group $(\R, +)$ as operators on $\bm{H}^1(\Sigma)$. This captures the notion that `information is not lost' when evolving with $\overline{\mathcal{S}}(\tau)$, whereas it inevitably is when evolving on the regular slicing with $\mathcal{S}(t)$. To justify the use of the nomenclature unitary, consider an element $\bm{u}$ of $\bm{H}^1(\Sigma'_\tau)$. Our usual approach is to use the isometry $\Sigma'_\tau \cong \Sigma'_0$, together with the embedding map $\iota:\Sigma' \to \Sigma'_0$ to induce an isometry $\bm{H}^1(\Sigma'_\tau)\cong \bm{H}^1(\Sigma')$. We could instead identify $\bm{u}\in \bm{H}^1(\Sigma'_\tau)$ with $(\overline{\mathcal{S}}(-\tau)\bm{u})\circ\iota \in \bm{H}^1(\Sigma)$. With respect to this identification of $\bm{H}^1(\Sigma'_\tau) \cong \bm{H}^1(\Sigma')$, the semigroup $\overline{\mathcal{S}}(\tau)$ is represented by unitary operators (in fact, the identity).

More usefully, for suitable $L$ (such as the Klein-Gordon operator with positive mass), the $T-$energy is conserved and positive definite, so $\overline{\mathcal{S}}(\tau)$ is represented by unitary operators on the Hilbert space derived from the $T-$energy. For general $L$ this need not be the case.

We can express $\overline{\A}$ in terms of the operator $L$ as follows. First we note that $L$ may be decomposed as
\be
L\psi  = \Omega^2 \left( \overline{P}_2 \psi + T \overline{P}_1 \psi +TT\psi \right ),
\ee
where $\overline{P}_i$ are differential operators on $\Sigma'$ of order $i$ and $\Omega$ is defined to be
\be
\Omega^2:= \frac{g^{-1}(d\tau, d\tau)}{g^{-1}(dt,dt)}.
\ee
$\Omega$ is smooth and bounded on any compact subset of $\Sigma'$ which avoids $S$. In terms of the operators $\overline{P}_i$, we may write
\ben{Akdef2}
\overline{\A}\left(\begin{array}{c}\psi \\ \psi' \end{array}\right)= \left(\begin{array}{cc}0 & 1 \\-\overline{P}_2 & -\overline{P}_1\end{array}\right)\left(\begin{array}{c}\psi \\ \psi' \end{array}\right).
\een

It will be useful to define the (rescaled) Laplace transform of $L$ with respect to the time function $\tau$ by
\be
\overline{L}_s u := \left.  \Omega^{-2} e^{-s \tau}L e^{s \tau} \upsilon \right|_{\Sigma'_\tau} = \overline{P}_2 u + s \overline{P}_1 u +s^2 u.
\ee
This operator acts on functions defined on $\Sigma'$ and we again denote by $\upsilon$ the lift of the function $u$ which satisfies $\left. \upsilon \right|_{\Sigma'_0} = u$, $T\upsilon = 0$. As for the case of the regular slicing, we have:
\begin{Lemma}\label{Lslemma2}
The resolvent $(\overline{\A} - s)^{-1}$ exists and is a bounded linear transformation of $\bm{H}^1(\Sigma')$ onto $D(\overline{\A})$ if and only if $\overline{L}_s^{-1}:\overline{L}^{2}(\Sigma) \to D(\overline{L}_s)$ exists as a bounded operator and furthermore $D(\overline{L}_s) \subset \overline{H}^{1}(\Sigma, \kappa)$. Furthermore, we have
\ben{residen2}
(\overline{\A} - s)^{-1} = \left(\begin{array}{cc}-1 & 0 \\-s & 1\end{array}\right) \left(\begin{array}{cc}(\overline{L}_s)^{-1} & 0 \\0 & 1\end{array}\right)\left(\begin{array}{cc}\overline{P}_1 + s & 1 \\1 & 0\end{array}\right).
\een
\end{Lemma}

It is the resolvent $(\overline{L}_s)^{-1}$ whose meromorphic extension is usually used to define the quasinormal modes. We will take the liberty of also referring to this as the unitary resolvent, since it is so closely related to $(\overline{\A}-s)^{-1}$. We note that often $\overline{L}_s$ is a self-adjoint operator with respect to an appropriate $L^2$ inner product for functions on $\Sigma'$.

\subsection{The meromorphic extension}

An immediate consequence of Lemmas \ref{unires}, \ref{Lslemma2} is that $(\overline{L}_s)^{-1}:\overline{L}^2(\Sigma') \to \overline{H}^1(\Sigma', \kappa)$ exists as a bounded operator, holomorphic in $s$ for $\Re(s) > M$. Now, unlike in the case of $(\hL)^{-1}$, we cannot improve the range on which $(\overline{L}_s)^{-1}$ exists by considering $\overline{L}_s$ as an operator on a more regular subspace of $\overline{L}^2(\Sigma')$. This approach fails because of the unsuitability of the $\Sigma'_\tau$ slicing for capturing behaviour at $\hor$. We instead seek a meromorphic extension of $(\overline{L}_s)^{-1}:\overline{L}^2_c(\Sigma') \to \overline{L}^2_{loc.}(\Sigma')$. Here $\overline{L}^2_c(\Sigma')$ denotes the space of functions in $\overline{L}^2(\Sigma')$ whose essential support is contained in a compact set away from $S$ and $\overline{L}^2_{loc.}(\Sigma')$ is the space of functions $u$ on $\Sigma'$ such that $\chi u \in \overline{L}^2(\Sigma')$, where $\chi$ is a cut-off function vanishing near $S$. To do this, let us first introduce the following map which takes functions defined a.e.\ on $\Sigma'\setminus \mathscr{H}$ to functions defined a.e.\ on $\Sigma\setminus \mathscr{H}$. Suppose $u:\Sigma' \setminus \mathscr{H}\to \C^N$ is a function defined a.e.\ on $\Sigma'\setminus \mathscr{H}$. We define the operator $Q_s$ by:
\be
Q_s  :  u \mapsto \left. e^{s \tau} \upsilon\right |_{\Sigma_0}.
\ee
We will require the following properties of $Q_s$:
\begin{Lemma}\label{changeslice}
The map $Q_s$ is invertible, with $(Q_s)^{-1}$ taking functions defined a.e.\  on $\Sigma\setminus \mathscr{H}$ to functions defined a.e.\ on $\Sigma'\setminus \mathscr{H}$. Let $\chi:\Sigma' \to \R$ be any smooth cut-off function which is equal to zero near $S$. Then $Q_s \circ \chi : \overline{L}^2(\Sigma') \to \L^2(\Sigma)$ is a holomorphic family of bounded operators, as is $\chi \circ (Q_s)^{-1} : \L^2(\Sigma) \to \overline{L}^2(\Sigma')$. Furthermore
\be
 \Omega^{2}\overline{L}_s  = (Q_s)^{-1} \circ \hL \circ Q_s.
\ee
\end{Lemma}

At this stage, we need to prove a meromorphic extension result for the resolvent on $\Sigma$. This is a straightforward consequence of the resolvent identity, together with the results we have previously established.
\begin{Lemma}
The resolvent $ (\A-s)^{-1} : \bm{H}_c^1(\Sigma) \to \bm{H}^1(\Sigma)$, which is meromorphic on  $\{\Re(s)>-\frac{1}{2}\varkappa\}$, admits a meromorphic extension, $R(s, \A)$, to $\C$ with poles of finite order, whose residues are finite rank operators. The poles may only occur at points of $\Lambda_{QNF}$, but there may be points of $\Lambda_{QNF}$ at which the meromorphic extension of $R(s, \A)$ is regular.
\end{Lemma}
\begin{proof}
First note that Theorem \ref{mainthm} immediately gives us that $R(s, \A)$ is meromorphic on $\{\Re(s)>-\frac{1}{2}\varkappa\}$. Fix a set $W\cc \Sigma \setminus \hor$, and assume $\bm{u}\in \bm{H}^2(\Sigma)$ is supported in $W$. Suppose now that $-\frac{3}{2}\varkappa<\Re(s) \leq \frac{1}{2}\varkappa$ and that $s \not \in \Lambda_{QNF}$. We know from Theorem \ref{mainthm} that $R(s, \A)u$ exists and belongs to $\bm{H}^2(\Sigma)$. Pick some $s' \in \{\Re(s)>-\frac{1}{2}\varkappa, s \not \in \Lambda_{QNF}\}$. We know from Lemma \ref{comp} that $R(s', \A)\bm{u}$ exists and furthermore
\be
\norm{R(s', \A)\bm{u}}{\bm{H}^2(\Sigma)} \leq C_{W, s'} \norm{\bm{u}}{\H^1(\Sigma)}.
\ee
Consider now the resolvent identity
\ben{resident}
R(s; \A) = R(s'; A) + (s-s') R(s; \A) R(s'; \A).
\een
From here we deduce that
\be
\norm{R(s; \A)\bm{u}}{\bm{H}^2(\Sigma)} \leq C_{W,s'} \norm{R(s; \A)}{\bm{H}^2(\Sigma) \to \bm{H}^2(\Sigma)}\norm{\bm{u}}{\bm{H}^1(\Sigma)}.
\ee
Thus by continuity we can define $R(s; \A): \bm{H}_c^1(\Sigma) \to \bm{H}^1(\Sigma)$ in $\{\Re(s)>-\frac{3}{2}\varkappa, s \not \in \Lambda_{QNF}\}$, and furthermore by \eq{resident},  $R(s; \A)$ is holomorphic in this domain. If $R(s; \A): \bm{H}^2(\Sigma) \to \bm{H}^2(\Sigma)$ has a pole of order $k$ at some point $s\in \{\Re(s)>-\frac{3}{2}\varkappa, s \in \Lambda_{QNF}\}$, then the meromorphic extension $R(s; \A): \bm{H}_c^1(\Sigma) \to \bm{H}^1(\Sigma)$ has a pole of order $\leq k$ at $s$.
As a result, the set of poles of the meromorphic extension is a subset of $\Lambda_{QNF}$. Working inductively, we may extend this proof to $\C$.
\end{proof}

Now, since $R(s; \A)$ and $\hL^{-1}$ are closely related, an equivalent statement holds for $\hL^{-1}$. Putting this together with Lemma \ref{changeslice}, we have the following theorem
\begin{Theorem}\label{ancon}
The unitary resolvent $(\overline{L}_s)^{-1}:\overline{L}^2(\Sigma')\to \overline{H}^1(\Sigma', \kappa)$, defined for $\Re(s)> M$ admits a meromorphic extension to $\C$ as an operator\footnote{i.e.\ $\chi_1 (\overline{L}_s)^{-1} \chi_2:\L^2(\Sigma) \to \L^2(\Sigma)$ is a meromorphic operator for any cut-off cuntions $\chi_i$ vanishing near the horizon} $(\overline{L}_s)^{-1}:\overline{L}^2_c(\Sigma') \to \overline{L}^2_{loc.}(\Sigma')$. The set of poles of this operator is a subset of $\Lambda_{QNF}$, the poles have finite order and the residues at the poles are finite rank operators. To each pole is associated a finite dimensional space of smooth solutions satisfying the boundary conditions at $\scri$ to the equation: 
\be
\overline{L}_s u = 0, \qquad \textrm{ such that }e^{s \tau} u \textrm{ is regular at }\hor.
\ee
\end{Theorem}
Note that some QNF may not appear as poles of the resolvent $(\overline{L}_s)^{-1}$. This in fact occurs in the example of \S\ref{example}. The QNF we found previously are are not poles of this meromorphic extension of $(\overline{L}_s)^{-1}$. Our analysis of that problem shows that they nevertheless deserve to be considered `honest' quasinormal frequencies. The reason that we can lose some QNM in using the resolvent approach is that to meromorphically extend the resolvent, we have to restrict to functions supported away from the horizon. At a particular QNF, $s$, It may happen that such functions are all in the range of $(\A-s)$, even though the range is not all of the space $\bm{H}^k(\Sigma)$. In order to characterise more precisely which of the QNF appear as poles of the resolvent, we introduce the idea of `ingoing' boundary conditions.

\subsection{Ingoing boundary conditions} The `ingoing' boundary conditions may be introduced as  follows. Let $X$ be the set of initial data on $\Sigma$ which vanish near the horizon and launch a smooth solution of the Klein-Gordon equation on $\bhR$. We define $Y = \{ \mathcal{S}(t) \bm{\psi} : \bm{\psi} \in X, t\geq 0\}$. Clearly $Y$ is a linear subspace of $\bm{H}^k(\Sigma)$, which is preserved by $\mathcal{S}(t)$. We define $\bm{H}^k_{in.}(\Sigma)$ to be the closure of $Y$ with respect to the norm $\bm{H}^k(\Sigma)$, and $(D^k_{in.}(\A), \A)$ to be the closure of $(Y, \A)$ as an unbounded operator on $\bm{H}_{in.}^k(\Sigma)$. We can think of $\bm{H}^k_{in.}(\Sigma)$ as the smallest closed subspace of $\bm{H}^k(\Sigma)$  containing $\bm{H}^k_c(\Sigma)$ which is preserved by $\mathcal{S}(t)$.

As a result of $\bm{H}_{in.}^k(\Sigma)$ being preserved by $\mathcal{S}(t)$, we can define a restricted resolvent $(\A-s)^{-1}: \bm{H}^k_{in.}(\Sigma)\to D^k_{in.}(\A)$. We state without proof the following result:
\begin{Lemma}\label{reslem}
The restricted resolvent $(\A-s)^{-1}: \bm{H}^k_{in.}(\Sigma)\to D^k_{in.}(\A)$ is meromorphic for $\Re(s) >-(k-\frac{1}{2})\varkappa$. The restricted resolvent has poles only at the location of the poles of the unrestricted resolvent $(\A-s)^{-1}: \bm{H}^k(\Sigma)\to D^k(\A)$, but they may be of lower order. In particular the restricted resolvent may be holomorphic at some of the quasinormal frequencies. Denoting by $\Lambda^k_{QNF, in.}$ the locations of the poles of the restricted resolvent, we have
\be
\Lambda_{QNF, in.}^k \subset \Lambda_{QNF}^k.
\ee
Furthermore, to each pole of the restricted resolvent is associated a finite dimensional space of smooth solutions to
\be
(\A-s) \bm{u} = 0,
\ee
with $\bm{u} \in \bm{H}^k_{in.}(\Sigma)$ for all $k$. We refer to such $\bm{u}$ as ingoing quasinormal modes.
\end{Lemma}

We say a smooth solution $u$ to $\overline{L}_s u=0$ is `ingoing with order $k$ at the horizon' if $(e^{s \tau}u, s e^{s \tau}u)|_{\Sigma_0}\in \bm{H}^k_{in.}(\Sigma)$. In the Schwarzschild-AdS case, a function is ingoing at the horizon to all orders iff:
\be
e^{m r_*} \left(\frac{\partial u}{\partial r_*} -s u \right) \to 0, \qquad \textrm{ as } r_* \to -\infty,  \textrm{ for any } m .
\ee
Here $r_*$ is the Regge-Wheeler tortoise coordinate which tends to minus infinity on $\hor$ (see \cite{Holbnd}, \S2.1). An analogous definition can be constructed for other black holes.

We may relate this back to the meromorphic extension of $(\overline{L}_s)^{-1}:\overline{L}^2_c(\Sigma') \to \overline{L}^2_{loc.}(\Sigma')$ as follows:
\begin{Lemma}
The meromorphic extension of the resolvent $(\overline{L}_s)^{-1}:\overline{L}^2_c(\Sigma') \to \overline{L}^2_{loc.}(\Sigma')$. Has poles at precisely those points $s\in \Lambda_{QNF, in.}$. To each such pole is associated a finite number of smooth solutions to 
\be
 \overline{L}_s u=0,
\ee
where $u$ is ingoing at the horizon to all orders and satisfies the boundary conditions at $\scri$. Such solutions are in one-to-one correspondence with the ingoing quasinormal modes.
\end{Lemma}
\begin{proof}
Suppose $s \in \Lambda_{QNF, in.}$ but $s$ is not a pole of $(\overline{L}_s)^{-1}:\overline{L}^2_c(\Sigma') \to \overline{L}^2_{loc.}(\Sigma')$. This, making use of $Q_s$, implies that there exists a proper subspace of $\bm{H}^k_{in.}(\Sigma)$ which contains $Y$ and is preserved by $\mathcal{S}(t)$. By the construction of $\bm{H}^k_{in.}(\Sigma)$ this is absurd.
\end{proof}

Thus by using the restricted resolvent we recover completely the traditional definition of quasinormal modes as `ingoing' solutions to $\overline{L}_s u = 0$. The price we pay however is that we must consider solutions of the Klein-Gordon equation which arise from perturbations which vanish on the horizon. This is a fairly strong (and unphysical) restriction, which is readily seen to be unnecessary with our more general definition of quasinormal modes.

\subsection{Asymptotically hyperbolic manifolds}

In this subsection we shall briefly sketch how our results may be applied to show the meromorphicity of the resolvent for asymptotically hyperbolic manifolds. This is a classic result of Melrose and Mazzeo \cite{Mazzeo1987260}, see also the refinement of this result due to Guillarmou \cite{MR2153454}. Our approach exploits a close relationship between static Killing horizons and asymptotically hyperbolic manifolds  \cite{Gibbons:2008hb}. 

We will show how to recover the results for even manifolds (in a sense we shall define). We note that the results of \cite{MR2153454} which obtain meromorphicity in a half-plane requiring only even expansions up to a certain order should also be obtainable with our method, since so long as we work only to a finite order in our Sobolev spaces, the full smoothness of the metric is not required. As we have proven our results everywhere assuming smooth metrics we will not pursue this possibility.

Let $X$ be a $(n+1)-$dimensional compact manifold with boundary, and let $x$ be a boundary defining function for $\partial X$. Let $h$ be a  smooth Riemannian metric on $\mathring{X}$. We say that $(X, h)$ is an asymptotically hyperbolic manifold if $x^{2} h$ extends smoothly to $\partial X$ and furthermore all the sectional curvatures of $h$ approach $-1$ near the boundary. We say that $(X, h)$ is even if there exist local coordinates $y^i$ on the surfaces $x=\textrm{const.}$ such that we have
\be
h = \frac{dx^2 + h_{ij}(y, x^2)dy^i dy^j}{x^2},
\ee
for some smooth functions $h_{ij}$.

Let us define $\mathscr{M}:= [0, \infty)_\tau \times \mathring{X}$ and consider the following Lorentzian metric, defined on $\mathscr{M}$:
\be
\tilde{g} = -d\tau^2 + h.
\ee
On this background we will consider the Klein-Gordon equation:
\be
\Box_{\tilde{g}} \phi - \frac{n^2}{4} \phi = \frac{\partial^2 \phi }{\partial \tau^2}+ \Delta_h\phi - \frac{n^2}{4} \phi =0.
\ee
Laplace transforming in $\tau$, we see that the resolvent we should study for this problem is
\ben{hyperbolic resolvent}
\left( \Delta_h + s^2 -  \frac{n^2}{4} \right)^{-1} = \left( \Delta_h - \lambda(n-\lambda) \right)^{-1},
\een
where we have made the trivial shift in the spectral parameter $\lambda = s +\frac{n}{2}$. This is the object studied in \cite{MR2153454}.

In order to relate the problem on $(\mathscr{M}, \tilde{g})$ to a black hole problem, we perform a conformal transformation. Let us introduce the metric $g$, defined on $\mathscr{M}$ by
\be
g = x^2 \tilde{g} = -x^2 d\tau^2 + x^2 h
\ee
Near $\partial X$ we have:
\be
g  = -x^2 d\tau^2 + dx^2 + h_{ij}(y, x^2)dy^i dy^j.
\ee
Making the change of variables near $[0, \infty)_\tau \times \partial X$ given by $u(x,t) = x e^{-\tau}$, $v(x,t) = x e^{\tau}$, we have 
\be
g  =   du dv + h_{ij}(y^i, uv)dy^i dy^j.
\ee
Now, let us define a manifold with stratified boundary by glueing the limit points $u \to 0$ to $\mathscr{M}$ as a boundary: $\bhR_e:=\mathscr{M}\cup \{u=0\}$. We extend the differentiable structure of $\mathscr{M}$ to $\bhR_e$ by defining the smooth functions on $\bhR_e$ to be those which extend to $\{u=0\}$ as smooth functions of $u, v, y^i$. It is here that the evenness condition is required: if $h_{ij}$ were permitted to depend on odd powers of $x$ we would have terms in the expansion of $g$ near the boundary of the form $u^{n+\frac{1}{2}}$.

On $\bhR_e$ we have a Killing field given by
\be
T = v \frac{\partial}{\partial v} - u \frac{\partial}{\partial v},
\ee
which is null on $\hor = \{u=0, v>0\}$, vanishes on $S = \{u=0,v=0\}$ and is timelike everywhere else. The surface $\Sigma' = \{\tau = 0\}$ is everywhere timelike, and by extending the surface given near the boundary by $\{v-u =1\}$ we can construct a surface $\Sigma$ which is timelike and everywhere to the future of $\Sigma'$. We define $\mathscr{H}:= \{u=0, v\geq0\} = \hor\cup S$

As a consequence, $(\bhR_e, \mathscr{H}, \Sigma, \Sigma', g, T)$ fits our definition of extended asymptotically anti-de Sitter black holes, if we take the trivial choices $r=1$ and $\scri = \emptyset$. This is a considerable abuse of notation, since this manifold is really asymptotically de Sitter and contains no anti-de Sitter end, but our results nevertheless apply. 

We can verify that
\be
\left[\Box_{\tilde{g}} - \frac{n^2}{4} \right ] (\phi) = \rho^{1+\frac{n}{4}}\left[ \Box_g + \frac{n}{4 (n+1)} V\right] \left( \rho^{-\frac{n}{4}} \phi \right),
\ee
where
\be
 V =  R_{g} - \frac{R_{\tilde{g}} + n(n+1)}{x^2}.
\ee
Since the sectional curvatures of $h$ approach $-1$ to $\O{x^2}$ near $\partial X$ we have that $V$ is smooth on $\bhR_e$. We can thus apply our results to study solutions of
\be
\left( \Box_g + \frac{n}{4 (n+1)} V\right) \psi = 0.
\ee
In particular, the unitary resolvent associated to this equation will be meromorphic. Since we restrict to functions vanishing near the horizon, this can be seen to imply the meromorphicity of \eq{hyperbolic resolvent} as an operator $L^2_c(X) \to L^2_{loc.}(X)$.

\newpage
 \vfill
\appendix 
\section{Purely oscillatory QNM of hermitian operators}\label{oscapp}
In this brief appendix we shall prove that for \emph{hermitian} strongly hyperbolic operators satisfying a fairly weak condition on the horizon, quasinormal modes cannot be purely oscillatory. In particular, this will allow us to `upgrade' the boundedness proofs of \cite{Holzegel:2012wt} to establish decay of solutions, albeit without a quantitative decay rate. The result may be seen as a generalisation of Theorem 1.2 of \cite{HolSmul} to more general black holes and operators.

As usual, we consider a strongly hyperbolic operator, $L$, on a globally stationary asymptotically anti-de Sitter background with stationary, homogeneous, boundary conditions prescribed at infinity. We say that $L$ is \emph{hermitian} if $L$ is formally self-adjoint, i.e.\ if for any function $\phi \in C^\infty_{bc}(\mathring{\bhR}, \C^N)$, we have $L\phi = L^\dagger \phi$. Note that $L$ is \emph{not} self-adjoint, since the operator and its adjoint have different domains. In particular the operator $L$ associated with the scalar Klein-Gordon equation is hermitian.

\begin{Lemma}\label{hermlem}
Let $L$ be a hermitian strongly hyperbolic operator on an asymptotically AdS black hole $\bhR$. Suppose furthermore that $T^\mu W_\mu = 0$ on $\hor$, where $W$ is the vector field part of the operator $L$, as in Definition \ref{LLdef}. Then $L$ has no non-zero, purely imaginary, quasinormal frequencies.
\end{Lemma}
\begin{proof}
Suppose $L$ has a purely imaginary quasinormal frequency $s\neq 0$, with corresponding smooth quasinormal mode $u$. If $\upsilon$ is the stationary lift of $u$ to $\bhR$, we have $L e^{s t} \upsilon = 0$. Now apply Corollary \ref{adjrel} with $\phi_1 = \phi_2 = e^{s t} \upsilon$. Both $L\phi_1$ and $L^\dagger \phi_2 = L \phi_2 $ vanish. Since $s$ is purely imaginary, the current $K$ is stationary, and the boundary conditions imply $K^\mu m_\mu = 0$ on $\scri$. Thus the adjointness relation reduces to the simple statement:
\be
\Im(s) \int_\hor \abs{u}^2 d\sigma = 0,
\ee
whence $u$ vanishes on the horizon. 

Repeatedly differentiating the equation $\hL u = 0$, we can show that this implies that $u$ vanishes to all orders on $\hor$. Restricting attention to a coordinate patch near the horizon, $\hL$ is a Fuchsian type operator of the kind considered in \cite{roberts} (of the form of Example (1)). Invoking the unique continuation result proven in that paper\footnote{The result in \cite{roberts} applies to \emph{scalar} equations, however our assumptions are such that the principle part of the operator is scalar, and so the microlocal factorisation approach in this paper can be extended to our situation.}, we can conclude that $u=0$ in a neighbourhood of the horizon. Since $\hL$ is elliptic away from the horizon, the Calder\'{o}n uniqueness theorem \cite{calderon} for the Cauchy problem for elliptic operators implies that $u\equiv0$.
\end{proof}

\section{Supplementary proofs} \label{supplapp}
\subsection{Proof of Lemma \ref{emprop}} \label{Lemma appendix}
For convenience we re-state the lemma:
\begin{Lem*}[Lemma \ref{emprop}]
For a sufficiently regular $\phi$, we have
\ben{divt2}
\nabla_\mu \emT^\mu{}_\nu[\phi] = -\Re\left[\tn_\nu \overline{\phi}\cdot \left(\Box_g\phi+ \frac{1}{2} (V+V^*) \phi  \right) \right]+ \emS_\nu[\phi],
\een
where
\be
2 \emS_\nu[\phi] = \Re\left\{ \overline{\phi} \cdot \left[ \tn^\dagger_\nu ( F f) f^{-1} \phi\right]+  (\tn_\sigma \overline{\phi}) \cdot \left[f^{-1} (\tn^\dagger_\nu  f)  (\tn^\sigma \phi)\right]\right \}.
\ee
\end{Lem*}
\begin{proof}
First let us note that, by expanding with Leibniz rule and using the diagonal form of $f$, we have:
\bea
\tn^\dagger_\mu \tn_\nu \phi &=& -f^{-1} \nabla_\mu\left[f^2 \nabla_{\nu} \left(f^{-1} \phi\right) \right] \nonumber \\
&=& - \nabla_\mu \nabla_\nu \phi + \left(f^{-1}\nabla_\mu \nabla_\nu f \right) \phi + f^{-1}\nabla_\nu f \nabla_\mu \phi - f^{-1}\nabla_\mu f \nabla_\nu \phi, \label{double twisted}
\eea
so that
\be
\tn^\dagger_\mu \tn^\mu = \Box_g +  \left(f^{-1}\nabla_\mu \nabla^\mu f \right).
\ee
Now let us consider each term in $\nabla_\mu \emT^\mu{}_\nu$ in turn. Firstly we have, using that $f$ is real and diagonal:
\bean
 \nabla^{\mu}\left[  \tn_{(\mu} \overline{\phi} \cdot \tn_{\nu)} \phi \right] &=&  \frac{1}{2} \nabla^{\mu} \left[ (f\tn_{\mu}\phi)^* f^{-2} (f \tn_{\nu} \phi) + (f\tn_{\nu} \phi)^*  f^{-2} (f \tn_{\mu} \phi) \right] \\&=& \frac{1}{2}  \big [ (\nabla^{\mu} f\tn_{\mu}\phi)^* f^{-2} (f \tn_{\nu} \phi) + (f\tn_{\mu}\phi)^* f^{-2} (\nabla^\mu f \tn_{\nu} \phi)\\&&\quad  + (\nabla^\mu f\tn_{\nu} \phi)^*  f^{-2} (f \tn_{\mu} \phi) +  (f\tn_{\nu} \phi)^*  f^{-2} ( \nabla^\mu f \tn_{\mu} \phi)
 \\&&\quad + (f\tn_{\mu}\phi)^* (\nabla^\mu f^{-2}) (f \tn_{\nu} \phi) + (f\tn_{\nu} \phi)^*  (\nabla^\mu f^{-2}) (f \tn_{\mu} \phi)  \big ]
 \\&=& \Re\Big[ (f\tn_{\nu} \phi)^*  f^{-2} ( \nabla^\mu f \tn_{\mu} \phi)+ (f\tn_{\mu}\phi)^* f^{-2} (\nabla^\mu f \tn_{\nu} \phi)\\&& \quad + (f\tn_{\mu}\phi)^* (\nabla^\mu f^{-2}) (f \tn_{\nu} \phi) \Big] \\
 &=& \Re \left[ -(\tn_{\nu} \phi)^*   ( \tn_\mu^\dagger  \tn^{\mu} \phi)-(\tn^{\mu}\phi)^*  (\tn_\mu^\dagger  \tn_{\nu} \phi) + (f\tn_{\mu}\phi)^* (\nabla^\mu f^{-2}) (f \tn_{\nu} \phi)\right] .
\eean
Now, let us consider:
\bean 
 \nabla^{\mu}\left[  g_{\mu \nu} \tn_\sigma  \overline{\phi} \cdot \tn^\sigma \phi \right] &=&  \nabla_{\nu}\left[  (f \tn_\sigma \phi)^* f^{-2}  (f \tn^\sigma \phi) \right] \\&=& (\nabla_\nu f \tn_\sigma \phi)^* f^{-2}  (f \tn^\sigma \phi) + (f \tn_\sigma \phi)^* f^{-2}  (\nabla_\nu f \tn^\sigma \phi)  \\&&\quad +(f \tn_\sigma \phi)^* (\nabla_\nu f^{-2} ) (f \tn^\sigma \phi)\\ &=& \Re\left[2 (f \tn_\sigma \phi)^* f^{-2}  (\nabla_\nu f \tn^\sigma \phi)+ (f \tn_\sigma \phi)^* (\nabla_\nu f^{-2} ) (f \tn^\sigma \phi) \right] \\&=& \Re\left[-2  (\tn^\mu \phi)^*   (\tn^\dagger_\nu  \tn_\mu \phi)+ (f \tn_\sigma \phi)^* (\nabla_\nu f^{-2} ) (f \tn^\sigma \phi) \right].
\eean
Putting these two calculations together, and using \eq{double twisted} we find:
\bean
\nabla^{\mu}\left[  \tn_{(\mu} \overline{\phi} \cdot \tn_{\nu)} \phi - \frac{1}{2}   g_{\mu \nu} \tn_\sigma  \overline{\phi} \cdot \tn^\sigma \phi  \right] &=& \Re \Big[ -(\tn_{\nu} \phi)^*   ( \tn_\mu^\dagger  \tn^{\mu} \phi) \\ && \quad - (\tn^{\mu}\phi)^*  (\tn_\mu^\dagger  \tn_{\nu} \phi) + (\tn^\mu \phi)^*   (\tn^\dagger_\nu  \tn_\mu \phi)
\\&& \quad  + (f\tn_{\mu}\phi)^* (\nabla^\mu f^{-2}) (f \tn_{\nu} \phi) \\&& \quad - \frac{1}{2}  (f \tn_\mu \phi)^* (\nabla_\nu f^{-2} ) (f \tn^\mu \phi) \Big] \\ &=& -\Re\left[\tn_\nu \overline{\phi}\cdot \left(\Box_g\phi+ \left(f^{-1}\nabla_\mu \nabla^\mu f \right) \phi  \right) \right] 
\\ &&\quad + 2 \Re\left [(\tn^\mu \phi)^* (f^{-1}\nabla_\mu f) \tn_\nu \phi - (\tn^\mu \phi)^*(f^{-1}\nabla_\nu f)(  \tn_\mu \phi)\right ] \\ && \quad + 2 \Re \Big [-(\tn_{\mu}\phi)^* (f^{-1} \nabla^\mu f) ( \tn_{\nu} \phi) \\&& \quad +  \frac{1}{2} ( \tn_\mu \phi)^* (f^{-1} \nabla_\nu f ) (\tn^\mu \phi) \Big]
\\ &=& -\Re\left[\tn_\nu \overline{\phi}\cdot \left(\Box_g\phi+ \left(f^{-1}\nabla_\mu \nabla^\mu f \right)  \right) \phi  \right] 
\\ &&\quad -  \Re\left [ ( \tn_\mu \overline{\phi}) \cdot \left[ (f^{-1} \nabla_\nu f ) (\tn^\mu \phi)\right] \right].
\eean
Now, finally we consider the last term in  $\nabla_\mu \emT^\mu{}_\nu$. We have:
\bean
 \nabla^{\mu}\left[  g_{\mu \nu}\overline{\phi}\cdot  F\phi \right] &=&   \nabla_{\nu}\left[  (f^{-1} \phi)^*  (f Ff)  (f^{-1} \phi) \right]  \\&=& \big[  (\nabla_{\nu} f^{-1} \phi)^*  (f Ff)  (f^{-1} \phi)+ (f^{-1} \phi)^*  (f Ff)  (\nabla_{\nu} f^{-1} \phi)\\&&\quad + (f^{-1} \phi)^*  (\nabla_{\nu} f Ff)  (f^{-1} \phi)\big] \\&=& 2 \Re \left[ (\tn_\nu \phi)^* F \phi\right] - \phi^* \left [\tn_\nu^{\dagger} (Ff) f^{-1} \right] \phi.
\eean
Taking everything together, we have
\bean
\nabla_\mu \emT^\mu{}_\nu 
&=& -\Re\left[\tn_\nu \overline{\phi}\cdot \left(\Box_g\phi+ \left(f^{-1}\nabla_\mu \nabla^\mu f \right) + F  \right) \phi  \right]  \\&&\quad +  \Re\left [\frac{1}{2}\overline{\phi}\cdot \left [\tn_\nu^{\dagger} (Ff) f^{-1} \right] \phi -( \tn_\mu \overline{\phi}) \cdot \left[ (f^{-1} \nabla_\nu f ) (\tn^\mu \phi)\right] \right].
\eean
Writing $(f^{-1} \nabla_\nu f ) = -\frac{1}{2} (f^{-1} \tn^\dagger_\nu f )$, and recalling
\be
F =  \frac{1}{2} (V+V^*)- f^{-1}\nabla_\mu \nabla^\mu f,
\ee
and 
\be
2 \emS_\nu[\phi] = \Re\left\{ \overline{\phi} \cdot \left[ \tn^\dagger_\nu ( F f) f^{-1} \phi\right]+  (\tn_\sigma \overline{\phi}) \cdot \left[f^{-1} (\tn^\dagger_\nu  f)  (\tn^\sigma \phi)\right]\right \}.
\ee
We finally have:
\be
\nabla_\mu \emT^\mu{}_\nu[\phi] = -\Re\left[\tn_\nu \overline{\phi}\cdot \left(\Box_g\phi+ \frac{1}{2} (V+V^*) \phi  \right) \right]+ \emS_\nu[\phi],
\ee
and we're done.
\end{proof}
\subsection{Proof of the twisted trace identity}

For completeness we include the twisted trace identity, which was stated as part of Lemma 4.2.1 of \cite{Warnick:2012fi} but not explicitly proven there.

We assume $(\mathring{\mathscr{M}}^{d+1}, g)$ to be a time oriented Lorentzian manifold with an asymptotically AdS end, with asymptotic radial coordinate $r$ which we assume extends as a smooth positive function throughout $\mathring{\mathscr{M}}^{d+1}$. We take $\Sigma$ to be a spacelike surface which extends to the conformal infinity of the asymptotically AdS end, $\scri$, and meets $\scri$ orthogonally with respect to the conformal metric $r^{-2} g$. We finally assume that $\scri \cap\Sigma$ is compact. We will show:

\begin{Lemma}\label{twisted trace}
Suppose that $\phi \in \H_\mathcal{D}^1(\Sigma, \kappa)$ and let $I \in \{1, \ldots, N\}\setminus \mathcal{D}$. Then for any $\delta>0$, there exists $C_\delta$ such that
\ben{apptraceeq}
 \int_{\scri \cap \Sigma} \abs{\phi_I}^2 r^{-2\kappa_I} \sqrt{A} d \mathcal{K} \leq \delta \norm{\chi \phi}{\H^1(\Sigma, \kappa)}^2 + C_\delta \norm{\chi \phi}{\L^2(\Sigma)}^2
\een
Where $\chi$ is some smooth function equal to $1$ near $\scri$ and vanishing outside a neighbourhood of $\scri$. 
\end{Lemma}
\begin{proof}
We first note that it will suffice to prove the result for $\phi \in C^\infty_\mathcal{D}(\Sigma; \C^N)$, defined to be the set of $\phi$ such that
\begin{align*}
 &\phi = e^{(-\frac{d}{2}\iota+\kappa)\log r} \phi_++e^{(-\frac{d}{2}\iota-\kappa)\log r} \phi_-,  \\ &\phi_{\pm} \in C^\infty(\Sigma; \C^N), \qquad  \phi_+^I = 0\   \textrm{ for }I \in \mathcal{D}, \\&\qquad \qquad \norm{\phi}{\H^1(\Sigma, \kappa)}<\infty,
\end{align*}
as such functions are dense in $\H_\mathcal{D}^1(\Sigma, \kappa)$.

From the definition of the asymptotically flat end and the compactness assumption, we can cover a neighbourhood of $\scri \cap \Sigma$ with a finite number of coordinate patches with coordinates $(\rho, x^\alpha)$, with $r = \rho^{-1}$, such that if $h$ is the metric induced from $g$ on $\Sigma$ we have
\be
h_{\rho \rho} = \frac{l^2}{\rho^2} + \O{1}, \quad h_{\rho \alpha} =  \O{1}, \quad h_{\alpha \beta} = \frac{\mathfrak{h}_{\alpha \beta}}{\rho^2} + \O{1},
\ee
as $\rho \to 0$, where $\mathfrak{h}$ is a Riemannian metric. Without loss of generality, we may assume that $\rho \in [0, \epsilon)$ and $x^\alpha$ take values in the unit $(d-1)$-ball, $B_1$. By considering a partition of unity, we may assume that $\phi$ is supported in one such coordinate chart.

A brief calculation shows that in these coordinates the measure induced on a surface $\{\rho = \textrm{const.} \}\cap \Sigma$ satisfies
\be
d \mathcal{K} = \sqrt{\mathfrak{h}} \rho^{-d+1} dx
\ee
so that
\be
\sqrt{A} d \mathcal{K} = \nu \rho^{-d} dx
\ee
where $\nu$ is a smooth positive function such that $\nu, 1/\nu$ are both bounded. Similarly, we have that
\be
dS_{\Sigma} = \nu' \rho^{-d} dx d\rho
\ee
for some other smooth positive function  $\nu'$  such that $\nu', 1/\nu'$ are both bounded.

Now let us fix $I\in \{1, \ldots, N\} \setminus \mathcal{D}$, so that $0<\kappa_I<1$, and consider the quantity
\be
\eta(\rho)^2 = \int_{B_1} \abs{ \rho^{\kappa_I} \phi_I(\rho, x)}^2 \rho^{-d} dx.
\ee
We note that for $\phi \in C^\infty_\mathcal{D}(\Sigma; \C^N)$ this expression extends to a smooth and bounded function on $[0, \epsilon)$ and that moreover
\be
 \int_{\scri \cap \Sigma} \abs{\phi_I}^2 r^{-2\kappa_I} \sqrt{A} d \mathcal{K} \leq C \eta(0)^2.
\ee
for some constant depending only on $g, \Sigma, \kappa$. A short calculation, making use of the Cauchy-Schwarz inequality, allows us to estimate:
\begin{align*}
\abs{\eta'(\rho)} &= \abs{\frac{1}{\eta(\rho)} \Re \int_{B_1} \overline{\partial_\rho \left( \rho^{\kappa_I - \frac{d}{2}} \phi_I \right)}\left( \rho^{\kappa_I - \frac{d}{2}} \phi_I dx \right)} \\ &\leq \left( \int_{B_1}   \abs{\partial_\rho \left( \rho^{\kappa_I - \frac{d}{2}} \phi_I \right)}^2 dx \right)^{\frac{1}{2}}.
\end{align*}
Next we calculate
\begin{align*}
\abs{\eta(0) - \eta(\rho)} &= \abs{\int_0^\rho \eta'(s) ds} \leq \int_0^\rho \abs{\eta'(s)} ds \\ 
&\leq \int_0^\rho \left( \int_{B_1}   \abs{\partial_s \left( s^{\kappa_I - \frac{d}{2}} \phi_I(s, x) \right)}^2 dx \right)^{\frac{1}{2}} ds \\ 
&\leq \left(\int_0^\rho \int_{B_1}   \abs{s^{\frac{d}{2}-\kappa_I} \partial_s \left( s^{\kappa_I - \frac{d}{2}} \phi_I(s, x) \right)}^2 s^{-d+1} dx  ds\ \cdot  \int_0^\rho s^{2\kappa_I - 1} ds\right)^{\frac{1}{2}} \\
&\leq C \rho^{\kappa_I} \norm{\phi}{\H^1(\Sigma, \kappa)}, 
\end{align*}
for some constant $C$ depending only on $g, \Sigma, \kappa$. From here we can estimate
\be
\eta(0)^2 \leq 2 C^2 \rho^{2 \kappa_I} \norm{\phi}{\H^1(\Sigma, \kappa)}^2 + 2 \eta(\rho)^2.
\ee
Now, multiplying by $\rho^{1-2\kappa_I}$ and integrating in $\rho$ over $[0, \delta]$ for some $\delta<\epsilon$, we have
\be
\int_0^\delta \eta(0)^2  \rho^{1-2\kappa_I} d\rho = \eta(0)^2 \frac{\delta^{2-2\kappa_I}}{2-2\kappa_I},
\ee
and
\be
\int_0^\delta   \rho^{2 \kappa_I} \norm{\phi}{\H^1(\Sigma, \kappa)}^2 \rho^{1-2\kappa_I} d\rho = \frac{1}{2} \delta^{2} \norm{\phi}{\H^1(\Sigma, \kappa)}^2 .
\ee
Finally
\be
\int_0^\delta \eta(\rho)^2  \rho^{1-2\kappa_I} d\rho =  \int_0^\delta \int_{B_1} \abs{  \phi_I(\rho, x)}^2 \rho^{-d+1} dx \leq C' \norm{\phi}{\L^2(\Sigma)}^2
\ee
Putting these estimates together, we have
\be
\eta(0)^2 \leq C\left( \delta^{2 \kappa_I}\norm{\phi}{\H^1(\Sigma, \kappa)}^2  + \delta^{2\kappa_I - 2}  \norm{\phi}{\L^2(\Sigma)}^2\right) 
\ee
for some constant $C$ depending only on $g, \Sigma, \kappa$. By taking $\delta$ small, we can make the coefficient in front of the $\H^1$ norm arbitrarily small. Combining this estimate with a partition of unity we are done.
\end{proof}

\newpage
\section{Glossary of symbols} \label{glossary}
{\small
\begin{longtable}{ll}
$A$ &The lapse function of the stationary metric, \eq{sigA}, p\pageref{sigA} \\
$\A$ &The generator of the wave semi-group, Def. \ref{Adef}, p\pageref{Adef}\\
$\beta_I$ & The Robin functions, Def. \ref{wpdefs}, p\pageref{wpdefs} \\
$C$ & A constant, which may vary from line to line. \\
$C^\infty_{bc}(\Sigma; \C^N)$ & Smooth functions obeying boundary conditions at $\scri$, Def. \ref{wpdefs}, p\pageref{wpdefs} \\
$d+1$ & The dimension of the spacetime manifold, Def. \ref{adsenddef}, p\pageref{adsenddef}\\
$D^k(\A)$ &The domain of $\A:D^k(\A) \to \mathbf{H}^k(\Sigma)$, Def. \ref{Adef}, p\pageref{Adef}\\
$\mathcal{D}^+(\Sigma)$ &The future domain of dependence of $\Sigma$, Def. \ref{wpdefs}, p\pageref{wpdefs} \\
$\mathcal{D}$ & The set of components obeying Dirichlet conditions, Def. \ref{wpdefs}, p\pageref{wpdefs} \\
$E_\gamma(t)[\phi]$ & The Killing energy, Def. \ref{Kildef}, p\pageref{Kildef} \\
$\mathcal{E}_\gamma(t)[\phi]$ & The redshift energy, Def. \ref{RSdef}, p\pageref{RSdef} \\
$f$ & The twisting matrix, \eq{fdef}, p\pageref{fdef} \\
G & The axial symmetry group of a locally stationary BH, Def. \ref{locstatbh}, p\pageref{locstatbh}  \\
$\H^1(\Sigma, \kappa)$,  &The twisted Sobolev space, Def. \ref{wpdefs}, p\pageref{wpdefs} \\
$\overline{H}^1(\Sigma', \kappa)$,  &The twisted Sobolev space for the unitary slicing,  \S\ref{unitary}, p\pageref{unitary} \\
$\H^{1, k}(\Sigma, \kappa)$,  &A higher twisted Sobolev space, \S\ref{LTop}, p\pageref{LTop}\\
$\H^1_{\mathcal{D}}(\Sigma, \kappa)$ &The twisted Sobolev space, with Dirichlet conditions\\&  for the $\mathcal{D}$ components Def. \ref{wpdefs}, p\pageref{wpdefs} \\
$\mathbf{H}^k(\Sigma)$ & The higher regularity Sobolev spaces, Def. \ref{Hkdef}, p\pageref{Hkdef} \\
$\hor$ &The black hole horizon, Def. \ref{adsbh}, p\pageref{adsbh}\\
$\scri$ &The conformal infinity of an aAdS spacetime, Def. \ref{adsenddef}, p\pageref{adsenddef} \\
$\curJ^T_\gamma[\phi]$ & The Killing energy current, Def. \ref{Kildef}, p\pageref{Kildef} \\
$\curJ^N_\gamma[\phi]$ & The redshift energy current, Def. \ref{RSdef}, p\pageref{RSdef} \\
$\kappa$ & The matrix of asymptotic exponents, \eq{kappadef}, p\pageref{kappadef} \\
$\varkappa$ &The surface gravity, \eq{sgdef}, p\pageref{sgdef} \\
$l$ & The AdS radius, Def. \ref{adsenddef}, p\pageref{adsenddef} \\
$L$, $L^\dagger$ &A strongly hyperbolic operator and its adjoint, Def. \ref{LLdef}, p\pageref{LLdef}\\
$\hL, \hL^\dagger$ &The Laplace transform of $L$, $L^\dagger$, \S\ref{LTop}, p\pageref{LTop}\\
$\L^2(\Sigma)$ &The renormalised $L^2$ space, Def. \ref{wpdefs}, p\pageref{wpdefs} \\
$\overline{L}^2(\Sigma')$ &The renormalised $L^2$ space for the unitary slicing, \S\ref{unitary}, p\pageref{unitary} \\
$\L^{2,k}(\Sigma)$ & A higher twisted Sobolev space, \S\ref{LTop}, p\pageref{LTop}\\
$\tn_\mu$, $\tn_\mu^\dagger$ & The twisted derivative and its adjoint, \eq{twist}, p\pageref{twist} \\
$\Phi$ & An axial killing field of a locally stationary BH, Def. \ref{locstatbh}, p\pageref{locstatbh}  \\
$r$ &The asymptotic radial coordinate of an aAdS spacetime, Def. \ref{adsenddef}, p\pageref{adsenddef} \\
$\bhR$ &A stationary black hole region, Def. \ref{adsbh}, p\pageref{adsbh} (see also Def. \ref{locstatbh}, p\pageref{locstatbh})\\
$\bhR_e$ &An extended stationary black hole region, Def. \ref{extadsbh}, p\pageref{extadsbh} \\
$S$ &The bifurcation surface of an extended black hole region, Def. \ref{extadsbh}, p\pageref{extadsbh} \\
$\cS(t)$& The solution operator, Def. \ref{Sdef}, p\pageref{Sdef}\\
$\sigma$ &The norm of the Killing field $T$, \eq{sigA}, p\pageref{sigA} \\
$\Sigma$ &A (typically spacelike) hypersurface\\
$\Sigma_t$ &A family of spacelike surfaces foliating a black hole, Def. \ref{adsbh}, p\pageref{adsbh} \\
$\Sigma'_\tau$ &The unitary slicing of an extended black hole region, Def. \ref{extadsbh}, p\pageref{extadsbh} \\
$T$ &The Killing generator of the black hole horizon, Def. \ref{adsbh}, p\pageref{adsbh}\\
$\emT_{\mu \nu}[\phi]$ & The twisted energy-momentum tensor of $\phi$, Def. \ref{emtdef}, p\pageref{emtdef} \\
$w_L, w_L^*$ & Bounds on the cross-term, Def. \ref{RSdef}, p\pageref{RSdef} \\
$\mathfrak{X}(\mathscr{M})$ &The smooth vector fields on $\mathscr{M}$, Def. \ref{adsenddef}, p\pageref{adsenddef} \\
$\mathfrak{X}_\Sigma(\mathscr{M})$ &The smooth vector fields on $\mathscr{M}$ tangent to $\Sigma$, Def. \ref{adsenddef}, p\pageref{adsenddef} \\
$\mathfrak{X}^*(\mathscr{M})$ &The smooth one-form fields on $\mathscr{M}$, Def. \ref{adsenddef}, p\pageref{adsenddef} \\
\end{longtable}

}

\providecommand{\href}[2]{#2}\begingroup\raggedright\endgroup

\subsection*{Acknowledgements} I would like to thank Gustav Holzegel and Mihalis Dafermos for helpful advice and comments. I am grateful to PIMS and NSERC for funding. Part of this work was undertaken at Perimeter Institute. Research at Perimeter Institute is supported by the Government of Canada through Industry Canada and by the Province of Ontario through the Ministry of Economic Development \& Innovation. I am also very grateful to the anonymous referees for a thorough and helpful review of the manuscript.

\end{document}